%% file: main.tex
\documentclass[acmsmall,screen]{acmart}
\setcopyright{rightsretained}
\acmPrice{}
\acmDOI{10.1145/3563331}
\acmYear{2022}
\copyrightyear{2022}
\acmSubmissionID{oopslab22main-p418-p}
\acmJournal{PACMPL}
\acmVolume{6}
\acmNumber{OOPSLA2}
\acmArticle{168}
\acmMonth{10}

\input{packages}

\declaretheoremstyle[
  headfont=\normalfont\scshape,
  notefont=\normalfont\scshape,
  notebraces={(}{)},
  bodyfont=\normalfont\itshape
]{thmstyle}
\declaretheoremstyle[
  headfont=\normalfont\scshape,
  notefont=\normalfont\scshape,
  notebraces={(}{)},
  bodyfont=\normalfont
]{defstyle}

\theoremstyle{defstyle}
\newtheorem{definition}{Definition}[section]

\newacronym{fg}{FG}{Featherweight Go}
\newacronym{fgg}{FGG}{Featherweight Generic Go}
\newacronym{gotogo}{\texttt{go2go}}{Go prototype translator}
\newacronym{dict}{\texttt{dict}}{Dictionary-Passing Translation}
\newacronym{mono}{\texttt{mono}}{Monomorphisation Translation}
\newacronym{erasure}{\texttt{erasure}}{Erasure Translation}

\newcommand{\dictpassddaggerunicodesymbol}{}
\DeclareUnicodeCharacter{01C2}{\dictpassddaggerunicodesymbol}

\input{macro}

\input{lstMacros}

\externalcitedocument{supplement}

\begin{document}

\title{Generic Go to Go}                
\subtitle{Dictionary-Passing, Monomorphisation, and Hybrid}                     



\author{Stephen Ellis}
\email{stephen.ellis20@imperial.ac.uk}
\orcid{0000-0002-3556-2201}
\affiliation{%
    \institution{Imperial College London}
    \city{London}
    \country{United Kingdom}
}
\author{Shuofei Zhu}
\authornote{Stephen Ellis and Shuofei Zhu contributed equally in this work.}
\email{sfzhu@psu.edu}
\orcid{0000-0003-3689-7668}
\affiliation{%
    \institution{The Pennsylvania State University}
    \city{University Park, PA}
    \country{United States}
}
\author{Nobuko Yoshida}
    \email{nobuko.yoshida@cs.ox.ac.uk}
    \orcid{0000-0002-3925-8557}
\affiliation{%
    \institution{University of Oxford}
    \city{Oxford}
    \country{United Kingdom}
}
\author{Linhai Song}
\email{songlh@ist.psu.edu}
\orcid{0000-0002-3185-9278}
\affiliation{%
    \institution{The Pennsylvania State University}
    \city{University Park, PA}
    \country{United States}
}

\begin{abstract}
    \input{sections/abstract}
\end{abstract}



\maketitle

\input{sections/introduction}

\input{sections/fg}

\input{sections/fgg}

\input{sections/nomono}

\input{sections/dict}

\input{sections/properties}

\input{sections/experiment-0}

\input{sections/experiment-1}

\input{sections/experiment-2}

\input{sections/experiment-3}

\input{sections/related}

\input{sections/conclusion}

\input{ack}

\bibliography{fcggreferences,session,related,asplos}

\ifnotsplit{
  \include{appendix}
}{}

\end{document}

%% file: packages.tex
\usepackage{xr-hyper}
\usepackage{xr}

\providecommand{\ifnotsplit}[2]{#1}
\usepackage{afterpage}

\usepackage{amssymb}
\usepackage{pifont}

\usepackage{booktabs}   
\usepackage{subcaption} 
\usepackage{tabularx}
\usepackage{multirow} 
\usepackage{soul}
\usepackage{proof}
\usepackage{mathpartir}
\usepackage{ifthen}
\usepackage{thm-restate}

\usepackage{bbding}

\usepackage{pfsteps} 
\usepackage[shortlabels, inline]{enumitem}
\usepackage[acronym]{glossaries}

\usepackage{tikz}
\usepackage{float}
\usetikzlibrary{positioning, arrows, automata, calc}
\usetikzlibrary{calc} 
\usepackage{tikz-cd}
\usetikzlibrary{decorations.pathmorphing}
\usepackage{listings}
\usepackage{wrapfig}
\usepackage{xcolor}
\usepackage{xspace}
\usepackage{stmaryrd}
\usepackage{mathtools}
\usepackage{cancel}
\usepackage{caption}

%% file: macro.tex
\newcommand{\gomacro}{Go~1.18\xspace}

\newcommand{\lit}[1]{\ensuremath{\operatorname{\text{\sffamily\bfseries{#1}}}}}
\newcommand{\old}[1]{{\color{gray}#1}}

\newcommand{\sytxRulepr}[2]{{\footnotesize#1}&{\footnotesize$#2$}}
\newcommand{\sytxRuleprMulti}[2]{{\footnotesize#1}&\multicolumn{3}{l}{\footnotesize$#2$}}
\newcommand{\sytxRuleIndentpr}[2]{\quad{\footnotesize#1}&\quad{\footnotesize$#2$}}
\newcommand{\sytxRuleIndentprMulti}[2]{\quad{\footnotesize#1}&\multicolumn{3}{l}{\quad{\footnotesize$#2$}}}

\newcommand{\oldSytxRulepr}[2]{\old{\footnotesize #1} & \old{\footnotesize$#2$}}
\newcommand{\oldSytxRuleIndentpr}[2]{\quad\old{\footnotesize#1} & \quad \old{\footnotesize$#2$}}

\newcommand{\sytxBrace}[1]          {\texttt{\{}#1\texttt{\}}}

\newcommand{\sType}[1]{\ensuremath{#1_S}}
\newcommand{\sTypeInit}[2]{\ensuremath{\sType{#1}\sytxBrace{#2}}}
\newcommand{\iType}[1]{\ensuremath{#1_I}}
\newcommand{\jType}[1]{\ensuremath{#1_J}}

\newcommand{\var}[1]{\mathit{#1}}
\newcommand{\multi}[1]{\overline{#1}}

\newcommand{\struct}[1]     {\lit{struct}       \sytxBrace{#1}}
\newcommand{\interface}[1]  {\lit{interface}   \sytxBrace{#1}}
\newcommand{\func}          {\lit{func}}
\newcommand{\funcDelc}[5]   {\func~(\this~#1)~#2(#3)~#4~\sytxBrace{#5}}
\newcommand{\type}          {\lit{type}}
\newcommand{\return}       {\!\lit{return}}

\newcommand{\typeFormal}[1][\multi{\alpha~\iType{\tau}}]{#1}

\newcommand{\prog}{P}

\newcommand{\program}[2][D]{\multi{#1} \triangleright #2}

\newcommand{\typeActualMethod}      {[\psi]}
\newcommand{\typeActualReceive}     {[\phi]}
\newcommand{\typeFormalMethod}      {[\Psi]}
\newcommand{\typeFormalType}      {[\Phi]} 
\newcommand{\typeFormalReceive}     {[\multi{\alpha}]}

\newcommand{\this}{\texttt{this}}

\newcommand{\dom}[1]{\operatorname{\var{dom}}\!\left({#1}\right)}
\newcommand{\fv}[1]{\operatorname{\var{fv}}\!\left({#1}\right)}

\newcommand{\panic}{\ensuremath{\mathsf{panic}}}

\definecolor{ruleColor}{rgb}{0.1, 0.3, 0.1}
\newcommand{\rulename}[1]{{\color{ruleColor}\text{\small [#1]}}}
\newcommand{\rulenamesmall}[1]{{\color{ruleColor}\text{\small [#1]}}}

\newcommand{\axiom}[1]{\infer{#1}{\strut}}

\newcommand{\namedRule}[2]{
   \begin{array}{l}
      \rulenamesmall{#1}\\
       #2
   \end{array}
}

\newcommand{\namedRuleTwo}[3]{
    \begin{array}{lr}
      \rulenamesmall{#1} & #2\\
      \multicolumn{2}{r}{#3}
    \end{array}
}

\newcommand{\ok}{\operatorname{\var{ok}}}
\newcommand{\wellTyped}[3][\Gamma]{#1 \vdash #2 : #3}
\newcommand{\wellTypedMulti}[3][\Gamma]{#1 \vdash \multi{#2 : #3}}
\newcommand{\wellFormed}[2][\Gamma]{#1 \vdash #2 \ok}
\newcommand{\wellFormedMulti}[2][\Gamma]{#1 \vdash \multi{#2 \ok}}

\newcommand{\subtype}[3][\Delta]{#1\vdash#2<:#3}
\newcommand{\subtypeMulti}[3][\Delta]{#1\vdash\multi{#2<:#3}}

\newcommand{\subParam}{<:-param}
\newcommand{\subS}{<:s}
\newcommand{\subI}{<:i}
\newcommand{\subFormal}{<:-formal}
\newcommand{\tParam}{t-param}
\newcommand{\tNamed}{t-named}
\newcommand{\tActual}{t-actual}
\newcommand{\tformal}{t-formal}
\newcommand{\tnested}{t-nested}
\newcommand{\tspecification}{t-specification}
\newcommand{\tstruct}{t-struct}
\newcommand{\tinterface}{t-interface}
\newcommand{\ttype}{t-type}
\newcommand{\tfunc}{t-func}
\newcommand{\tvar}{t-var}
\newcommand{\tliteral}{t-literal}
\newcommand{\tcall}{t-call}
\newcommand{\tfield}{t-field}
\newcommand{\tasserts}{t-assert$_S$}
\newcommand{\tasserti}{t-assert$_I$}
\newcommand{\tstupid}{t-stupid}
\newcommand{\tprog}{t-prog}

\newcommand{\fields}  {\var{fields}}
\newcommand{\body}    {\var{body}}
\newcommand{\vtype}   {\var{type}}
\newcommand{\unique}  {\var{unique}}
\newcommand{\tdecls}  {\var{tdecls}}
\newcommand{\mdecls}  {\var{mdecls}}
\newcommand{\methods}[1][] {\var{methods}\ifthenelse{\equal{#1}{}}{}{_#1}}
\newcommand{\distinct} {\var{distinct}}

\newcommand{\bounds}{\var{bounds}_\Delta}

\newcommand{\reduction}[2]{#1 \longrightarrow #2}
\newcommand{\red}{\longrightarrow}
\newcommand{\dicttrans}{\Longrightarrow}
\newcommand{\dictred}{\rightarrowtriangle}
\newcommand{\dictredleft}{\leftarrowtriangle}

\newcommand{\by}{\mathbin{:=}}
\newcommand{\noteq}{\mathbin{!=}}
\newcommand{\hole}{\square}

\newcommand{\rfields}{r-fields}
\newcommand{\rcall}{r-call}
\newcommand{\rassert}{r-assert}
\newcommand{\rcontext}{r-context}

\newcommand{\prepre}{\red_{\text{e}}}

\newcommand{\prepostsim}{\red_{\text{s}}}
\newcommand{\prepostdict}{\dictred}

\newcommand{\precongdict}{\dictred}

\newcommand{\dict}[3][\Delta; \eta ;\Gamma]{#1 \vdash #2 \Mapsto #3}
\newcommand{\dictMulti}[3][\Delta; \eta ;\Gamma]{#1 \vdash \multi{#2 \Mapsto #3}}

\newcommand{\dictDict}[2]{\vdash #1 \Mapsto_{\text{dict}} #2}
\newcommand{\dictMultiDict}[2]{\vdash \multi{#1 \Mapsto_{\text{dict}}
    #2}}

\newcommand{\lex}[1]{#1^\ddagger}

\newcommand{\dprogram}{d-program}
\newcommand{\dinterface}{d-interface}
\newcommand{\dmeth}{d-meth}
\newcommand{\dspec}{d-spec}
\newcommand{\ddict}{d-dict}
\newcommand{\dstruct}{d-struct}
\newcommand{\dvar}{d-var}
\newcommand{\dfield}{d-field}
\newcommand{\dvalue}{d-value}
\newcommand{\dassert}{d-assert}
\newcommand{\ddictcall}{d-dictcall}
\newcommand{\dcall}{d-call}

\newcommand{\apply}{\texttt{Apply}}
\newcommand{\trycast}{\texttt{tryCast}}

\newcommand{\any}{\texttt{Any}}
\newcommand{\rec}{\texttt{rec}}
\newcommand{\dictlit}{\texttt{dict}}
\newcommand{\typeField}{\texttt{\_type}}
\newcommand{\paramTypeMeta}{\texttt{param\_index}}
\newcommand{\fnMeta}[1]{\texttt{spec\_mdata}_{#1}}
\newcommand{\nAryFunction}[1]{\texttt{Function}_#1}
\newcommand{\typeMetadataLit}{\texttt{\_type\_mdata}}

\newcommand{\typemeta}[2][\zeta]{\ensuremath{\operatorname{\var{type\_meta}}_{#1}({#2})}}

\newcommand{\asParameters}{\method{asParam}}

\newcommand{\arities}{\method{arity}}
\newcommand{\maxFormals}{\method{maxFormal}}
\newcommand{\makeDictMeth}{\method{meth\_ptr}}
\newcommand{\metadataName}[1]{\method{mdata\_name}{{#1}}}
\newcommand{\signatureMeta}[2][\zeta]{\ensuremath{\operatorname{\var{sig\_mdata}}_{#1}({#2})}}

\newcommand{\specMetadata}{\method{spec\_mdata}}
\newcommand{\specName}{\method{spec\_name}}
\newcommand{\dictName}[1]{\method{typeDict}{#1}}
\newcommand{\methName}{\method{mName}}

\newcommand{\makeDict}[2][\eta, \Delta]{\ensuremath{\operatorname{\var{makeDict}}_{#1}({#2})}}

\newcommand{\method}[2]{\operatorname{\var{#1}}({#2})}

\newcommand{\caseof}[1]{  \resetpfcounter\textbf{Case : } Rule \rulename{#1}}
\newcommand{\caseStd}{  \resetpfcounter\textbf{Case : }}
\newcommand{\pfSubcase}[1]{\\\hspace{-1cm}\textbf{Subcase : } $#1$ }
\newcommand{\pf}[2][]{\item $#2$ \pflabel{#1}}
\newcommand{\pfstep}[3]{\\\vspace{-0.2cm}\hrulefill\vspace{-0.05cm}\item $#3$ \pflabel{#2} \BY{#1}}

\newcommand{\inversion}[1]{inversion~on~\rulename{#1}}
\newcommand{\pfbreakline}{$\\\qquad$}

\newcommand{\map}[2][\Delta; \eta ;\Gamma]{\ensuremath{\llbracket#2\rrbracket_{#1}}}

\newcommand{\congEval}{C}

\newcommand{\dictpattern}{\ensuremath{\rho_{\text{dict}}}}
\newcommand{\assertpattern}{\ensuremath{\rho_{\text{sim}}}}
\newcommand{\erasepattern}{\ensuremath{\rho_{\text{erase}}}}

\newcommand{\myparagraph}[1]{{\textbf{\emph{#1}}.\ }}

\newcommand\ie{\emph{i.e.,}\xspace}
\newcommand\eg{\emph{e.g.,}\xspace}

\newcommand{\mycircledtext}[1]{\raisebox{.5pt}{\textcircled{\raisebox{-.5pt} {\textit{#1}}}}}

\newcommand{\nomono}{{\it nomono}}
\newcommand{\bnfsep}{\mathbin{\;\big|\;}}%

\definecolor{benchnameColor}{rgb}{0.3, 0.2, 0.3}
\newcommand{\benchname}[1]{{\color{benchnameColor}{{\small \texttt{#1}}}}}
\newcommand{\smallgls}[1]{{\small \glsentryshort{#1}}}

\newcommand{\cmark}{\ding{51}}%
\newcommand{\xmark}{\ding{55}}%

\newcounter{observationcounter}

\newcommand{\gocode}[1]{\inlinelstfcgg{#1}}

\newcommand{\ENCan}[1]{\langle{#1}\rangle}

\newcommand{\sectSpace}{}
\newcommand{\subsectSpace}{}
\newcommand{\subsubsectSpace}{}

\newcommand{\mapother}[1]{\llparenthesis #1 \rrparenthesis}

\newcommand{\tick}{\CheckmarkBold}
\newcommand{\cross}{\XSolidBrush}

%% file: lstMacros.tex
\definecolor{turkishrose}{rgb}{0.71, 0.45, 0.51}

\definecolor{dkgreen}{rgb}{0,0.6,0}
\definecolor{gray}{rgb}{0.5,0.5,0.5}
\definecolor{mauve}{rgb}{0.58,0,0.82}
\definecolor{Cerulean}{rgb}{0.0, 0.48, 0.65}
\definecolor{PineGreen}{rgb}{0.0, 0.47, 0.44}
\definecolor{Violet}{rgb}{0.56, 0.0, 1.0}
\definecolor{Brown}{rgb}{0.59, 0.29, 0.0}

\definecolor{primekeywords}{rgb}{0.13,0.13,1}
\definecolor{secondkeywords}{rgb}{0,0.5,0}
\definecolor{greycomments}{rgb}{0.6,0.6,0.6}

\newcommand{\CODESTYLE}{\ttfamily}

\lstdefinelanguage{FCGG}{
      morekeywords=[1]{ 
        let, new, match, with, open, module, 
        namespace, type, of, member, for, while, 
        true, false, fun, return, yield,
        try, mutable, if, then, else, cloud, async, 
        static, use, abstract, interface, inherit, 
        finally, elif, not,forall,exists, public, 
        switch, case, default, func, var, go, defer, 
        struct, select, chan, new, assert, exists, and, not, import},
    emph={go, chan, <-, close, select, case, Lock, Unlock},
    sensitive=true,%
    morestring=[b]",%
    morestring=[s]{`}{`},%
}

\makeatletter
\newenvironment{btHighlight}[1][]
{\begingroup\tikzset{bt@Highlight@par/.style={#1}}\begin{lrbox}{\@tempboxa}}
{\end{lrbox}\bt@HL@box[bt@Highlight@par]{\@tempboxa}\endgroup}

\newcommand\btHL[1][]{%
  \begin{btHighlight}[#1]\bgroup\aftergroup\bt@HL@endenv%
}
\def\bt@HL@endenv{%
  \end{btHighlight}%
  \egroup
}
\newcommand{\bt@HL@box}[2][]{%
  \tikz[#1]{%
    \pgfpathrectangle{\pgfpoint{1pt}{0pt}}{\pgfpoint{\wd #2}{\ht #2}}%
    \pgfusepath{use as bounding box}%
    \node[anchor=base west, fill=orange!30,outer sep=0pt,inner xsep=1pt, inner ysep=0pt, rounded corners=3pt, minimum height=\ht\strutbox+1pt,#1]{\raisebox{1pt}{\strut}\strut\usebox{#2}};
  }%
}
\makeatother

\lstnewenvironment{lstfcgg}{}{}
\newcommand{\inlinelstfcgg}[1]{\lstinline[mathescape=true,basicstyle=\ttfamily]!#1!}

%% file: sections/abstract.tex
Go is a popular statically-typed industrial programming language.
To aid the type safe reuse of code, the recent  
Go release (\gomacro) published early 2022 
includes \emph{bounded parametric polymorphism} via \emph{generic types}.
\gomacro implements generic types 
using a combination of \emph{monomorphisation} and 
\emph{call-graph based dictionary-passing} called \emph{hybrid}.
This hybrid approach can be viewed as an optimised form of monomorphisation
that statically generates specialised methods and types
based on possible instantiations. 
A monolithic dictionary supplements information lost 
during monomorphisation, and is structured according 
to the program’s call graph. Unfortunately, the hybrid 
approach still suffers from code bloat, poor compilation 
speed, and limited code coverage.

In this paper we
propose and formalise a new \emph{non-specialising call-site based dictionary-passing} 
translation. 
Our call-site based translation creates individual dictionaries for each type parameter,
with dictionary construction occurring in place of instantiation, 
overcoming the limitations of hybrid.
We prove it correct
using a novel and general bisimulation up to technique.
To better understand how different generics translation approaches work in practice, 
we benchmark five 
translators, \gomacro, two existing 
monomorphisation translators, 
our dictionary-passing translator, 
and an erasure translator. 
Our findings reveal 
several suggestions for improvements 
for \gomacro --- specifically how to overcome the expressiveness limitations of 
generic Go and
improve compile time and compiled code size performance of \gomacro.

%% file: sections/introduction.tex
\section{Introduction}
\label{sec:introduction}
Since its creation in 2009, the Go programming language 
has placed a key emphasis on simplicity, safety, and efficiency. 
Based on the 
\citet{stackoverflow-developer-survey} survey, Go is the 5th most beloved 
language, and is used to build 
large systems, \eg \citet{docker}, 
\citet{kubernetes}, and \citet{grpc}.
The recent Go release (\gomacro released on the 15th of March 2022) 
added \emph{generics}, which has been considered Go's 
most critical missing and 
long awaited feature by 
Go programmers and developers~\cite{go-developer-survey}.  
\citet{go-release-notes}, however, has posted 
that much work is still needed to ensure that 
generics in Go are well-implemented. 

The work on implementing generics in Go began in earnest with 
\citet{griesemer2020featherweight},  
in which they formalised two core calculi of (generic) Go; \gls{fgg} and \gls{fg},
as well as formalising 
a \emph{monomorphisation translation} from \gls{fgg} to \gls{fg}.
Monomorphisation statically explores 
a program's call graph and generates multiple 
implementations of each generic 
type and method according to each 
specialisation of that type, or method, required at runtime.

The Go team informally proposed three approaches; 
\begin{enumerate*}
  \item Stencilling (monomorphisation)~\cite{google-mono},
  \item Call-graph dictionary-passing~\cite{go-dict-proposal}, and 
  \item GC shape stencilling (hybrid of (1) and (2))~\cite{google-hybrid}. 
\end{enumerate*}
A monomorphisation-based source-to-source prototype (\gls{gotogo}) 
has been implemented by 
\citet{gotogo}, following the stencilling proposal (1) and 
\cite{griesemer2020featherweight}.  
The current \gomacro implementation
extends (3)~\cite{go118}. 
Unlike more traditional \emph{non-specialising}
dictionary approaches 
(\eg dictionary-passing in Haskell and vtables in C++), 
\gomacro uses an optimised form of monomorphisation to allow types 
in the same GC shape group to share specialised method and type instances. 
In theory, all objects in a GC shape group have an equivalent 
memory footprint and layout, although currently, \gomacro  
only groups pointers.
As multiple types may share the same GC shape group,
their dictionaries provide information lost during monomorphisation, \eg
concrete types and method pointers.
Moreover, \gomacro builds a monolithic dictionary based on 
the program's \emph{call-graph}.
Monomorphisation has a 
number of well-known limitations;
it can substantially increase code 
size, it can be prohibitively slow
during compilation~\cite{jones1995dictionary, StroustrupB:cpppl},
and it does not cover all programs~\cite{griesemer2020featherweight}.
Concretely, there are two core limitations with all the Go team proposals
(1--3), the current \gomacro implementation, and the proposal of \citet{griesemer2020featherweight}.

\input{figs/fgg/fgg-nomono.tex}

\textit{1) Non-monomorphisable programs.\ } 
All current implementations and proposals for generics in 
Go suffer from the inability 
to handle a class of programs that 
use recursive instantiations, \eg the 
list permutation example\footnote{
  See \cite{gitchanderpermute} for an 
  efficient but type unsafe implementation of list permutation.} 
provided in 
Figure~\ref{fig:code:list:perm}.
This program cannot be monomorphised, as 
a list of integers  
\inlinelstfcgg{List[int]} has a 
\inlinelstfcgg{permute} method which returns a list of type 
\inlinelstfcgg{List[List[int]]}, which in turn has a 
\inlinelstfcgg{permute} method 
that returns type \inlinelstfcgg{List[List[List[int]]]}, and on
\emph{ad infinitum}. 
Monomorphisation cannot explore 
this infinite set of types in finite time, and 
so cannot specialise a method for each instance.

{\textit{2) Specialising translation.\ } 
All currently realised approaches to generics in Go are 
based on method/type specialisation. 
This stands in contrast 
to the approaches taken by other languages with automatic 
memory management, such as Haskell, C\#, and Java. 
Go uses garbage collection for automatic 
memory management. 
In the top  16 statically typed languages 
with generics~\cite{TopProgr17:online}, 
we find a constant theme; 
languages with automatic memory management use 
non-specialising implementations such as 
dictionary-passing or erasure, and those without 
use monomorphisation (see
\ifnotsplit{Appendix~\ref{app:implementations}}{the full version of this paper~\cite{fullversion}} for a breakdown of language implementations).

\myparagraph{Challenges and contributions}
We develop and implement a new non-specialising, call-site dictionary-passing translation 
from Go with generics (\gls{fgg})
to Go (\gls{fg}), and prove its correctness. We then create micro and 
real-world benchmarks for generic Go, 
and examine the trade-offs
between the different translations to suggest 
improvements for \gomacro.

{\textit{1) The first challenge is to design and build a 
non-specialising call-site dictionary-passing translation for Go.\ } 
Go's distinctive 
structural subtyping adds an extra level of complexity that requires careful consideration. 
Our first contribution in \S~\ref{section:dictionary} and \S~\ref{subsec:imple}
is the formalisation and implementation of a new dictionary-passing translation 
that is specifically designed for the unique qualities of Go.

{\textit{2) The second challenge is to overcome the 
non-monomorphisability limitation} of 
the current implementations and translate previously untranslatable
programs such as \inlinelstfcgg{permute}. 
A key aspect of our dictionary design is \emph{call-site}---each polymorphic type parameter 
is represented by its own dictionary, which in turn is created at 
the call-site where that type parameter would have been instantiated. 
This allows any well-formed \gls{fgg} program to be translated.

{\textit{3) The third challenge we meet is  
to establish semantic correctness of our
translation}. Historically, dictionary-passing translations 
have been proven correct using value preservation~\cite{yu2004formalization,yu2004thesis,sulzmann2021dictionary}, 
an approach that cannot ensure 
termination preservation or
generalise to more advanced language 
features (\eg concurrency in Go).
We instead use a fine-grained behavioural equivalence
guided by the work of \citet{Igarashi99FJ}.
Unfortunately, proving the \emph{bisimulation} result in  
\cite[Theorem 5.4]{griesemer2020featherweight} 
is insufficient due to intermediate states 
created by dictionary-passing. 
We propose a novel \emph{bisimulation up to dictionary resolution} reduction,
and use this relation to prove that the translation preserves 
essential properties of the source language (\S~\ref{section:properties}).  
This proof technique is general and translation-agnostic, 
and is useful in other contexts where a standard bisimulation 
is inadequate.

{\textit{4) The fourth challenge is to find an effective evaluation for implementations of 
generics in Go}. 
We compare the five implementations---
\begin{enumerate*}
\item our call-site, non-specialising dictionary-passing
translation;
\item an erasure translation 
built by us for empirical evaluation;
\item a monomorphisation translation by \citet{griesemer2020featherweight}; 
\item the initial source-to-source monomorphisation prototype translation \gls{gotogo} by the Go team; and 
\item \gomacro 
\end{enumerate*}
---along three dimensionalities;
\begin{enumerate*}
  \item complication time,
  \item translated code size, and
  \item performance of compiled executables.
\end{enumerate*}
As \gomacro was just released, 
\emph{there currently exists no real-world Go program
  with generics}.
In \S~\ref{subsec:evaluation}, we contribute a number of benchmarks to overcome this
deficit: we construct micro benchmarks to examine the effect of different 
forms of complexity in generic programs; 
and reimplement the real-world benchmarks from
\cite{odersky2000two,ureche2013miniboxing} in Go.

{\textit{5) The final challenge is to examine 
the trade-offs between the different translations, which 
suggest future improvements of \gomacro}. 
We observe, 
in general, that monomorphisation leads 
to better execution performance, while 
non-specialisation (dictionary-passing) produces 
smaller executables in less compilation time.
We also observe that on the micro benchmarks our dictionary-passing 
translation can generate programs that are comparable in efficiency 
to \gomacro. 
Overall, our results show that \gomacro has much scope for improvement and 
the usefulness of non-specialised call-site dictionary-passing translations for languages such as Go.
We provide concrete suggestions in \S~\ref{section:discussion}.

\myparagraph{\emph{Outline}} 
\S~\ref{section:fg} and \S~\ref{section:fgg} summarise \gls{fg} and
\gls{fgg};  \S~\ref{section:dictionary} proposes a new 
dictionary-passing translation; 
\S~\ref{section:properties} proves 
its semantic correctness; 
\S~\ref{sec:exp} describes our implementations 
and measures the trade-offs between the five translators;
\S~\ref{section:related} gives related work; and \S~\ref{section:conclusion}
concludes. 
Proofs and  
omitted definitions can
be found in 
\ifnotsplit{the Appendix to this paper.}{the full version of the paper \cite{fullversion}.}
The dictionary-passing/erasure translators  
and benchmarks are available in the artifact to this paper~\mbox{\cite{aritfact-entry}}.
Source code is available on GitHub~\cite{zhu22github} 
and Software Heritage~\cite{zhu22heritage}.

%% file: figs/fgg/fgg-nomono.tex
\begin{wrapfigure}{r}{0.52\linewidth}
  \vspace*{-.4cm}
  \lstset{xleftmargin=7pt}
  \begin{lstfcgg}
type List[T Any] interface { 
  permute() List[List[T]]; insert(val T, i int) List[T]; 
  map[R Any](func(T) R) List[R] ;  len() int }
type Cons[T Any] struct { head T; tail List[T] }
type Nil[T Any] struct {}
func (this Cons[T]) permute() List[List[T]] { 
  if this.len() == 1 { return Cons{this, Nil{}}
  } else {
    return flatten(this.tail.permute().Map(
      func(l List[T]) List[List[T]]{
        var l_new List[List[T]] = Nil[List[T]]{}
        for i := 0; i <= l.len(); i++ {
          l_new = Cons{l.insert(this.head, i), l_new}
        }
        return l_new
}))}}
func (this Nil[T]) permute() List[List[T]] { 
  return Nil[List[T]]{} 
}
  \end{lstfcgg}
  \vspace*{-7mm}
  \caption{List permutation example}
  \label{fig:code:list:perm}
  \vspace*{-3mm}
\end{wrapfigure}

%% file: sections/fg.tex
\section{\glsentrylong{fg}}
\label{section:fg}
\label{sec:fg}
We briefly summarise the \glsfirst{fg} language~\cite[\S~3]{griesemer2020featherweight};
specifically highlighting the key points 
related to dictionary translation.

\subsection{\glsentrylong{fg} by Examples}
\label{sec:fg:example}
\input{figs/fg/fg-code-function}
\gls{fg} is a core subset 
of the (non-generic) Go 1.16 language containing \emph{structures}, \emph{interfaces}, 
\emph{methods}, and \emph{type assertions}. 
In \gls{fg}, there are two kinds of named types; 
\emph{Interfaces} (\inlinelstfcgg{interface}) specify a collection 
of methods which any implementing type must also possess, and 
\emph{structures} (\inlinelstfcgg{struct}) which are 
data objects containing a fixed 
collection of typed fields. 
\emph{Methods} are functions that apply to a 
specific structure, called the method's \emph{receiver}. 
Finally, \emph{type assertions} ask whether a structure can be used 
as a specific type. If it cannot, then \gls{fg} will produce a 
\emph{type assertion error}.  

In contrast to nominally typed languages, Go uses 
\emph{structural subtyping}. 
As we shall see in \S~\ref{section:dictionary},
it is this 
distinctive feature that makes our dictionary-passing translation 
challenging and non-trivial. 
In a nominally typed language, such as Java, one type implements (subtypes)
another when it explicitly declares such. 
In Go, we do not declare that one type implements another.
Rather, one type implements another precisely when it implements 
(at least) all of the prescribed methods. 

Consider the example Go code in 
Figure~\ref{code:fg:example}, which 
simulates higher order functions, lists, and mapping.  
For simplicity of presentation, we assume 
that there are primitive \inlinelstfcgg{int} and \inlinelstfcgg{bool} 
types along with a $<$ operation. 
The \inlinelstfcgg{Any} interface does not specify any methods; as such, 
all other types are its subtypes, meaning that any object may be used 
when an \inlinelstfcgg{Any} is expected, but also that we cannot 
apply any methods to an \inlinelstfcgg{Any} object without first 
asserting it to some more specific type -- an action which may fail at runtime. 

The \inlinelstfcgg{Function} interface specifies a single method,
which is given by the \emph{method signature} 
\inlinelstfcgg{Apply(x Any) Any}. 
Any structure implementing 
an \inlinelstfcgg{Apply} method
that takes an argument of type \inlinelstfcgg{Any}
and returns a value, also of type \inlinelstfcgg{Any},
is said to 
implement the \inlinelstfcgg{Function} interface. 
Our example code simulates the \emph{greater than} function
as a structure (\inlinelstfcgg{GtFunc}) containing a single \inlinelstfcgg{Ord}
field. Its \inlinelstfcgg{Apply} method then calls the \inlinelstfcgg{Gt} 
method provided by struct's field. 
The \inlinelstfcgg{Ord} interface, however, specifies that \inlinelstfcgg{Gt} 
should accept a single argument of type \inlinelstfcgg{Ord}. 
Before the \inlinelstfcgg{Apply} method of \inlinelstfcgg{GtFunc} can call
\inlinelstfcgg{Gt} it must, then, assert its argument to 
type \inlinelstfcgg{Ord}. 
If the argument does not implement \inlinelstfcgg{Ord}, then a \emph{type assertion error} 
occurs. 
We assume that only one implementation of \inlinelstfcgg{Ord} exists, that being 
\inlinelstfcgg{int}, which itself uses a risky type assertion. 

The example also includes a \inlinelstfcgg{List} interface specifying
a \inlinelstfcgg{Map} method. We provide a cons list implementation 
of \inlinelstfcgg{List}. 
In \gls{fg}, there is a single top-level \inlinelstfcgg{main} function 
that acts as the program's entrance. 
Our program initially builds a simple three 
value \inlinelstfcgg{int} list on line~\ref{fg:code:function:build}, 
and then uses the simulated greater than function (\inlinelstfcgg{GtFunc}) to 
map the list to a \inlinelstfcgg{bool} list.
When, however, we 
attempt to map this \inlinelstfcgg{bool} list using the same function, we 
encounter a runtime type assertion error on line~\ref{fg:code:function:topanic}. 
While we could catch this error at compile time by 
increasing the specificity of the \inlinelstfcgg{Apply}, \inlinelstfcgg{Gt}, and 
\inlinelstfcgg{Map} functions using \inlinelstfcgg{int} and 
\inlinelstfcgg{bool} instead of \inlinelstfcgg{Any}, 
this would severely limit 
code reusability.

\subsectSpace
\subsection{\glsentrylong{fg} Syntax and Semantics}
\label{sec:fg:syntax}
\input{figs/fg/fg-syntax}
Figure~\ref{fig:fg:syntax} presents the syntax of \gls{fg} from  
\cite{griesemer2020featherweight}.
We use the $\multi{x}$ notation for a sequences of $x$, namely $x_0, x_1,
\dots, x_n$.    
A program ($\prog$) is given by a sequence of declarations ($\multi{D}$)
along with a {\bf main} function which acts as the top-level expression ($e$). 
Shortened as $\prog = \program{e}$.


\gls{fg} is statically typed: 
all \gls{fg} typing rules follow the Go 1.16 specification.  
If, in the variable-type
environment $\Gamma$, 
an expression $e$ is of type $t$, 
then it satisfies the judgement $\wellTyped[\Gamma]{e}{t}$. 
We assume that all programs $\prog$ are \emph{well-formed}, written 
$\prog\ok$. 
Since the rules/notations are identical to those 
in \cite{griesemer2020featherweight}, 
we omit them here, but provide
definitions and details in
\ifnotsplit{Appendix~\ref{app:fg}}{the full version of this paper~\cite{fullversion}}.

\label{section:fg:reduction}
Figure~\ref{fig:fg:semantics} presents the \gls{fg} semantics with values and 
evaluation contexts.
\textbf{\emph{Evaluation context}} $E$ defines the left-to-right call-by-value semantics
for expressions. 
\textbf{\emph{Reductions}} are defined by the field selection rule 
\rulename{r-fields}, type assertion rule \rulename{r-assert}, 
and the method 
invocation \rulename{r-call}, with  \rulename{r-context} for the context
evaluation. We use $\red^\ast$ to denote a multi-step reduction.
\gls{fg} satisfies type preservation and progress properties 
(see \cite[Theorems 3.3 and 3.4]{griesemer2020featherweight}).

\input{figs/fg/fg-semantics}

%% file: figs/fg/fg-code-function.tex
\begin{figure}
\begin{minipage}[t]{0.48\linewidth }
    \vspace{-3mm}  
\begin{center}
\begin{lstfcgg}
type Any interface {}
type Function interface { Apply(x Any) Any } (*\label{fg:code:function:function}*)
type Ord interface { Gt(x Ord) bool }
type List interface { Map(f Function) List }
type Nil struct {} 
type Cons struct { head Any ; tail List } 
func main() {
    _ = Cons{1, Cons{7, Cons{3, Nil{}}}} (*\label{fg:code:function:build}*)
    .Map(GtFunc{5}) (*\label{fg:code:function:tobool}*)
    .Map(GtFunc{5}) // PANIC (*\label{fg:code:function:topanic}*)
} //   Unable to assert bool as type Ord
\end{lstfcgg}
\end{center}
\end{minipage}
\begin{minipage}[t] {0.47\linewidth }
    \vspace{-3mm}  
\lstset{firstnumber=12}
\begin{lstfcgg}
type GtFunc struct { val Ord }
func (this GtFunc) Apply(x Any) Any {  (*\label{fg:code:function:apply}*)
    return this.val.Gt(x.(Ord)) //Gt needs an Ord arg
}
func (this int) Gt(x Ord) bool { 
    return x.(int) < this // < needs an int value
}
func (this Nil) Map(f Function) List { return Nil{} }
func (this Cons) Map(f Function) List {
    return Cons{ f.Apply(this.head), this.tail.Map(f) }
}
\end{lstfcgg}
\end{minipage}
\vspace*{-.5cm}
    \caption{\glsentryshort{fg} List example adapted from \cite[Figures~1~\&~3]{griesemer2020featherweight}}
    \label{code:fg:example}
\end{figure}

%% file: figs/fg/fg-syntax.tex
\begin{figure}[t]
\begin{tabular}{ll@{\hskip 24pt}ll@{\hskip 24pt}ll}
     \sytxRulepr{Field name}{f}&
       \sytxRulepr{Type name}{t, u ::= t_S \mid t_I }&
       \sytxRulepr{Expression} {e ::=}
       \\[-0.6mm]
       \sytxRulepr{Variable name}{x}&
       \sytxRulepr{Method name}{m}&
       \sytxRuleIndentpr{Method call}{e.m(\multi{e})}\\[-0.6mm]
        \sytxRulepr{Interface type}{\hspace{-0.3cm}\text{name} \  \iType{t}, \iType{u}} &
        \sytxRulepr{Method signature}{M ::= (\multi{x\ t})\ t} & 
        \sytxRuleIndentpr{Variable}{x}\\[-0.6mm]
        \sytxRulepr{Structure type}{\hspace{-0.3cm}\text{name} \  \sType{t}, \sType{u}}&
        \sytxRulepr{Method specification}{S ::= mM} & 
        \sytxRuleIndentpr{Type assertion}{e.(t)} \\[-0.6mm]
        \sytxRulepr{Type literal}{T ::=} &  \sytxRulepr{Declaration}{D::=} & 
        \sytxRuleIndentpr{Field select}{e.f}\\[-0.6mm]
            \sytxRuleIndentpr{Structure}{\struct{\multi{f\ t}}} & \sytxRuleIndentpr{Type decl}{\type t~T} & 
            \sytxRuleIndentpr{Structure
              literal}{\sTypeInit{t}{\multi{e}}}\\[-0.6mm]
            \sytxRuleIndentpr{Interface}{\interface{\multi{S}}} & 
            \sytxRuleIndentprMulti{Method decl}{\func (x\ \sType{t})\ mM\ \sytxBrace{\return e}}\\[-0.6mm]
        \sytxRuleprMulti{Program}{\prog ::=\textbf{ package main};\ \multi{D}\ \textbf{func main}() \sytxBrace{\_ = e}}\\[-0.6mm]

    \end{tabular}
\vspace{-3mm}
    \caption{Syntax of \glsentrylong{fg}}
\vspace{-3mm}
    \label{fig:fg:syntax}
\end{figure}
  

%% file: figs/fg/fg-semantics.tex
\begin{figure}
    {\small
        \begin{tabular}{llllll}
            \sytxRuleprMulti{Value}{v ::= \sTypeInit{t}{\multi{v}}}                    \\[-0.6mm]
            \sytxRulepr{Evaluation context}{E ::=}                               &   & \\[-0.6mm]
            \sytxRuleIndentpr{Hole}{\hole}                                       &
            \sytxRuleIndentpr{Structure}{\sTypeInit{t}{\multi{v}, E, \multi{e}}} &
            \sytxRuleIndentpr{Type assertion}{E.(t)}                                   \\[-0.6mm]
            \sytxRuleIndentpr{Select}{E.f}
                                                                                 &
            \sytxRuleIndentpr{Call receiver}{E.m(\multi{e})}                     &
            \sytxRuleIndentpr{Call arguments }{v.m(\multi{v}, E, \multi{e})}
        \end{tabular}
    }
    {\footnotesize \begin{equation*}
            \begin{gathered}
                \namedRule{\rfields}{
                    \infer{
                        \reduction{
                            \sTypeInit{t}{\multi{v}}.f_i
                        }{
                            v_i
                        }
                    }{
                        (\multi{f~t}) = \fields(\sType{t})
                    }
                }\quad
                \namedRule{\rcall}{
                    \infer{
                        \reduction{
                            v.m(\multi{v})
                        } {
                            e[x \by v, \multi{x \by v}]
                        }
                    }{
                        (x : \sType{t}, \multi{x:t}).e = \body(\vtype(v).m)
                    }
                }\quad
                \namedRule{\rassert}
                {
                \infer
                {
                \reduction{v.({t})}{v}
                }{
                \vtype(v) <: {t}
                }
                }
                \quad
                \namedRule{\rcontext}
                {
                    \infer{
                        \reduction{E[e]}{E[e']}
                    }{
                        \reduction{e }{e'}
                    }
                }
                \\[-0.6mm]
                \vtype(\sTypeInit{t}{\multi{v}}) = \sType{t}
                \quad
                \infer{
                    \body(\sType{t}.m) = (x:\sType{t}, \multi{x:t}).e
                }{
                    (\funcDelc{\sType{t}}{m}{\multi{x~t}}{t}{\return e}) \in \multi{D}
                }
                \quad 
                \infer{
                    \fields (t_S) = \multi{f~t}
                } {
                    (\type t_S \struct{\multi{f~t}}) \in \multi{D}
                }
            \end{gathered}
        \end{equation*}
    }
    \vspace{-4.5mm}
    \caption{Reduction semantics of \glsentrylong{fg}}
    \label{fig:fg:semantics}
\end{figure}

%% file: sections/fgg.tex
\section{\glsentrylong{fgg} and the limitations of monomorphisation and
\gomacro}
\label{section:fgg}
As with \S~\ref{section:fg}, we briefly summarise the
\glsentryfirst{fgg} language~
\cite[\S~4]{griesemer2020featherweight}.
This section concludes with 
a discussion of limitations in existing generic Go translations and \gomacro.

\subsectSpace
\subsection{\glsentrylong{fgg} by Example}
\input{figs/fgg/fg-code-function}

Figure~\ref{code:fgg:example} extends 
Figure~\ref{code:fg:example} with generics. 
As we saw in \S~\ref{sec:fg:example}, there was
a critical flaw in the original, non-generic, \gls{fg}
code. One part of the logic was polymorphic 
(\ie \inlinelstfcgg{Map} is a natural transformation) while the
other was not (\ie \inlinelstfcgg{Gt}). We concluded that 
section by observing the
two options; either we cater to the strict type 
discipline demanded by \inlinelstfcgg{Gt}, reducing
reusability, or force an excessively permissive 
polymorphism on \inlinelstfcgg{Gt} and risk runtime type assertion errors. 

Generics, or bounded parametric polymorphism, 
provide us with a third solution via the 
precise definition and tracking of polymorphic types in
structures, interfaces, and methods. 
As we shall see momentarily, in \gls{fgg}, each of 
these constructs may now accept any number of 
type variables (type parameters) as a type 
formal, which must then be instantiated upon use. 
Each type variable has a bound, an interface, that 
any instantiating type must satisfy, \ie be an instance of. 
Type formal \inlinelstfcgg{[T Any]} is read as type parameter 
\inlinelstfcgg{T} is bound by type \inlinelstfcgg{Any}.
Objects with a generic type can use all methods 
specified by the type variable's bound. 
Type variables can be bound by any interface type, and may be mutually recursive within a type formal. 
Take, for example, the type bound of \inlinelstfcgg{Ord} in Figure~\ref{code:fgg:example}. 
\inlinelstfcgg{Ord} is bound by \inlinelstfcgg{Ord} itself and is used recursively in the
type bound for \inlinelstfcgg{GtFunc}. 
For a type (\eg \inlinelstfcgg{int}) to instantiate type variable 
\inlinelstfcgg{T} in \inlinelstfcgg{[T Ord[T]]}, 
its \inlinelstfcgg{Gt} method must not only take an argument of \inlinelstfcgg{Ord}, 
but must be precisely the same \inlinelstfcgg{Ord}-implementing type.
This kind of self-referential type bound is known as 
\emph{F-bounded polymorphism} \cite{canning1989f}.

The interface \inlinelstfcgg{Function} 
is now defined over two type variables (\inlinelstfcgg{T} and 
\inlinelstfcgg{R}, both bounded by \inlinelstfcgg{Any}), 
which are used by the specified \inlinelstfcgg{Apply}
method to type the simulated function’s domain and codomain, respectively, 
e.g., a type implementing \inlinelstfcgg{Function[int, bool]} must 
implement the method \inlinelstfcgg{Apply(x int) bool}. 
Unlike the original \gls{fg} code, we do not need \inlinelstfcgg{GtFunc} 
to simulate any arbitrary function, but rather just functions from 
some generic \inlinelstfcgg{Ord} type
to \inlinelstfcgg{bool}. 
    Instantiating \inlinelstfcgg{GtFunc} with \inlinelstfcgg{int}, 
    written \inlinelstfcgg{GtFunc[int]}, 
    gives an implementation of \inlinelstfcgg{Function[int,bool]}.

A type bound not only limits which types may specialise a 
type parameter, but also what methods are available to 
polymorphic values, \ie 
given that all valid specialisations of \inlinelstfcgg{T} 
in \inlinelstfcgg{GtFunc[T]} 
must implement \inlinelstfcgg{Ord[T]}, we know that the 
\inlinelstfcgg{val} field must always 
possess the \inlinelstfcgg{Gt} method, allowing 
us to call to \inlinelstfcgg{Gt} on line~\ref{fgg:code:function:apply:gt}
without a type assertion.

The definition of \inlinelstfcgg{List} tracks not only the type of 
the list, but also the type of the list created 
by \inlinelstfcgg{Map}. The \inlinelstfcgg{Map} 
method accepts a type parameter along 
with a \inlinelstfcgg{Function} argument; this type parameter is then 
used as the codomain of the \inlinelstfcgg{Function} argument, and 
instantiates the \inlinelstfcgg{List} return type.
Line~\ref{fgg:code:function:typefail} thus fails during type
checking because \inlinelstfcgg{GtFunc} does not 
implement \inlinelstfcgg{Function[bool, bool]}.

\subsectSpace
\subsection{\glsentrylong{fgg} Syntax and Semantics}
\label{section:fgg:syntax}

\input{figs/fgg/fgg-syntax}

Figure~\ref{fig:fgg:syntax} presents the syntax of \gls{fgg}. 
The key differences from \gls{fg} are the addition of types formal
($\Psi, \Phi$) for method signatures and declarations. 
A type formal ($\typeFormal$) is a sequence of pairs, 
each of which contains 
a type parameter ($\alpha$) and parameter bound ($\iType{\tau}$).
Type bounds are interface types that 
can be mutually recursive, in that any bound in a type
formal may depend upon any type parameter in that type formal, including itself.
Type parameters are instantiated by a type actual ($\psi, \phi$) -- a
sequence of types that satisfy the requirements imposed by the type
formal. A type ($\tau$)  
in \gls{fgg} is either a type parameter or 
a declared type that has been instantiated
($t\typeActualReceive$).
We simplify method declaration from 
\gls{fgg}~\cite{griesemer2020featherweight}, following 
the \gomacro syntax. 

\label{section:fcgg:typing}

The type system in \gls{fgg} extends 
\gls{fg} with the addition of a new type variable context $\Delta$ 
mapping type variable to its bound. 
Expression $e$ of type $\tau$ is now given by the judgement 
$\wellTyped[\Delta;\Gamma]{e}{\tau}$. Program well-formedness is given 
by $\prog \ok$.
The typing rules follow those
given in \cite[Figure~15]{griesemer2020featherweight}, which can 
be found in
\ifnotsplit{Appendix~\ref{appendix:fgg}}{the full version of this paper~\cite{fullversion}}.

\label{section:fgg:reduction}
The reduction semantics of \gls{fgg} are defined in 
Figure~\ref{fig:fgg:semantics}. 
They extend those of \gls{fg}; notably, \rulename{\rcall}
(via the $\body$ auxiliary function) specialises generic types 
in the resolved method body. 
\gls{fgg} satisfies 
type preservation and progress properties 
(see \cite[Theorems 4.3 and 4.4]{griesemer2020featherweight}).

\input{figs/fgg/fgg-semantics}

%% file: figs/fgg/fg-code-function.tex
\begin{figure}
    \begin{minipage}[t]{0.44\linewidth }
        \vspace{-3mm}  
    \begin{center}
    \begin{lstfcgg}
type Any interface {}
type Function[T Any, R Any] interface { 
    Apply(x T) R 
}
type Ord[T Ord[T]] interface { Gt(x T) bool }
type List[T Any] interface { 
    Map[R Any](f Function[T, R]) List[R] 
}
type Nil[T Any] struct {} 
type Cons[T Any] struct { head T;tail List[T]} 
func main() {
 _ = Cons[int]{1, Cons[int]{7, Cons[int]{3, Nil[int]{}}}} // : List[int]
        .Map[bool](GtFunc[int]{5}) // : List[bool]
        .Map[bool](GtFunc[int]{5}) // This line doesn't pass type checking since   (*\label{fgg:code:function:typefail}*)
} // GtFunc does not implement the Function[bool, bool] interface
    \end{lstfcgg}
    \end{center}
    \end{minipage}
    \begin{minipage}[t] {0.52\linewidth }
        \vspace{-3mm}  
    \lstset{firstnumber=16}
    \begin{lstfcgg}
type GtFunc[T Ord[T]] struct { val T }
func (this GtFunc[T]) Apply(x T) bool {  (*\label{fgg:code:function:apply}*)
    return this.val.Gt(x) (*\label{fgg:code:function:apply:gt}*)
}
func (this int) Gt(x int) bool { return x < this }
func (this Nil[T]) Map[R Any](f Function[T, R]) List[R] {
    return Nil[R]{}
}
func (this Cons[T]) Map[R Any](f Function[T, R]) List[R] { (*\label{code:con:map}*)
    return Cons[R]{ f.Apply(this.head), this.tail.Map[R](f)}
}
    \end{lstfcgg}
    \end{minipage}
    \vspace*{-.5cm}
        \caption{\glsentryshort{fgg} List example adapted from \cite[Figures~4~\&~6]{griesemer2020featherweight}}
        \label{code:fgg:example}
        \vspace*{-.5cm}
    \end{figure}

%% file: figs/fgg/fgg-syntax.tex
{
\setlength{\tabcolsep}{.9mm}
\begin{figure}[t]
    \begin{tabular}{ll@{\hskip 24pt}ll@{\hskip 24pt}ll}
        \oldSytxRulepr{Field name}{f}                                                                   & \sytxRulepr{Type}{\tau, \sigma ::=}                                                         & \sytxRulepr{Expression}{e ::=}                                    \\[-0.6mm]
        \oldSytxRulepr{Variable name}{x}                                                                & \sytxRuleIndentpr{Type parameter}{\alpha}                                                   & \sytxRuleIndentpr{Method call}{e.m[\multi{\tau}](\multi{e})}      \\[-0.6mm]
        \oldSytxRulepr{Method name}{m}                                                                  & \sytxRuleIndentpr{Named type}{t[\multi{\tau}]}                                              & \oldSytxRuleIndentpr{Variable}{x}                                 \\[-0.6mm]
        \oldSytxRulepr{Interface type}{\hspace{-0.3cm}\text{name}\ \iType{t}, \iType{u}}                & \sytxRulepr{Interface type}{\iType{\tau}, \iType{\sigma} ::= \iType{t}[\multi{\tau}]}       & \sytxRuleIndentpr{Type assertion}{e.(\tau)}                       \\[-0.6mm]
        \oldSytxRulepr{Structure type}{\hspace{-0.3cm}\text{name}\ \sType{t}, \sType{u}}                & \sytxRulepr{Structure type}{\sType{\tau}, \sType{\sigma} ::= \sType{t}[\multi{\tau}]}       & \oldSytxRuleIndentpr{Field select}{e.f}                           \\[-0.6mm]
        \oldSytxRulepr{Type name}{t,u ::= \iType{t} \mid \sType{t}}                                     & \sytxRulepr{Interface-like type}{\jType{\tau}, \jType{\sigma} ::= \alpha \mid \iType{\tau}} & \sytxRuleIndentpr{Structure literal}{\sTypeInit{\tau}{\multi{e}}} \\[-0.6mm]
        \sytxRulepr{Type parameter}{\alpha}                                                             & \sytxRulepr{Type formal}{\Phi, \Psi ::= \typeFormal}                                        & \oldSytxRulepr{Type literal}{T ::=}                               \\[-0.6mm]
        \sytxRulepr{Declaration}{D::=}                                                                  & \sytxRulepr{Type actual}{\phi, \psi ::= \multi{\tau}}                                       & \oldSytxRuleIndentpr{Structure}{{\struct{\multi{f\ t}}}}          \\[-0.6mm]
        \sytxRuleIndentprMulti{Type decl.}{\type t[\Phi]~T}                                                                                                         & \oldSytxRuleIndentpr{Interface}{\interface{\multi{S}}}            \\[-0.6mm]
        \sytxRuleIndentprMulti{Method decl.}{\func (x\ \sType{t}[\multi{\alpha}])\ mM \sytxBrace{\return e}} & \oldSytxRulepr{Meth. spec.}{S ::= mM}                                                                                              \\[-0.6mm]
        \sytxRuleprMulti{Program}{\prog ::=\textbf{ package main};\ \multi{D}\ \textbf{func main}() \sytxBrace{\_ = e}}        & \sytxRulepr{Meth. signature}{M ::= [\Psi](\multi{x\ t})\ t}                                                                                                                                     
    \end{tabular}
    \vspace{-3mm}
    \caption{Syntax of \glsentrylong{fgg}}
    \vspace{-4mm}
    \label{fig:fgg:syntax}
\end{figure}

}

%% file: figs/fgg/fgg-semantics.tex
\begin{figure}[t]
   {
    \setlength{\tabcolsep}{.9mm}
    \begin{tabular}{llllll}
        \sytxRuleprMulti{Value}{v ::= \sTypeInit{\tau}{\multi{v}}}\\[-0.6mm]
        \sytxRulepr{Evaluation context}{E ::=}&&\\[-0.6mm]
            \sytxRuleIndentpr{Hole}{\hole}&
            \sytxRuleIndentpr{Structure}{\sTypeInit{\tau}{\multi{v}, E, \multi{e}}} & 
            \sytxRuleIndentpr{Call receiver}{E.m\typeActualMethod(\multi{e})}\\[-0.6mm]
            \sytxRuleIndentpr{Select}{E.f} & 
            \sytxRuleIndentpr{Type assertion}{E.(\tau)}&
            \sytxRuleIndentpr{Call arguments }{v.m\typeActualMethod(\multi{v}, E, \multi{e})}
    \end{tabular}
    }

    {\footnotesize
    \begin{equation*}
        \begin{gathered}
            \namedRule{\rfields}{
                \infer{
                    \reduction{
                        \sTypeInit{\tau}{\multi{v}}.f_i
                    }{
                        v_i
                    }
                } {
                    (\multi{f~\tau}) = \fields(\sType{\tau})
                }
            }\quad 
            \namedRule{\rcall}{
                \infer{
                    \reduction{
                        v.m\typeActualMethod(\multi{v}) 
                    } {
                        e[x \by v, \multi{x \by v}]
                    }
                }{
                    (x : \sType{\tau}, \multi{x:\tau}).e = \body(\vtype(v).m\typeActualMethod)
                }
            }\quad 
            \namedRule{\rassert}{
                \infer{
                    \reduction{v.({\tau})}{v}
                }{
                    \subtype[\emptyset]{\vtype(v)}{{\tau}}
                }
            }    \quad 
            \namedRule{\rcontext}{
                \infer{
                    \reduction{E[e]}{E[e']}
                }{
                    \reduction{e }{e'}
                }
            } 
            \\   
                    \vtype(\sTypeInit{\tau}{\multi{v}}) = \sType{\tau}
                    \quad 
                    \infer{
                        \body(\sType{t}\typeActualReceive.m\typeActualMethod) = (x:\sType{t}\typeActualReceive, \multi{x:\tau}).e[\theta]
                    }{
                        (\funcDelc{\sType{t}\typeFormalReceive}{m\typeFormalMethod}{\multi{x~\tau}}{\tau}{\return e}) \in \multi{D}
                        & \theta = (\multi{\alpha},\Psi \by \phi,\psi)
                    }
                    \\
                \infer{
                    \fields (\sType{t}\typeActualReceive) = \multi{f~\tau}[\eta]
                } {
                    (\type \sType{t}\typeFormalType \struct{\multi{f~\tau}}) \in \multi{D}
                    & \eta = (\Phi \by \phi)
                }
                \end{gathered}
            \end{equation*}
\vspace{-5mm}
        \caption{Reduction semantics of \glsentrylong{fgg}}
        \label{fig:fgg:semantics}
    }
\vspace{-3mm}
\end{figure}

%% file: sections/nomono.tex
\subsection{The Limitation of  Monomorphisation}
\label{section:nomono}

\citet{griesemer2020featherweight} 
define 
a class of programs that their monomorphisation approach
cannot translate.
This limitation also applies to the \gomacro call-graph based dictionary 
implementation for the same rationale.  
Consider the model non-monomorphisable 
program in Figure~\ref{fig:example:nomono}.

\begin{wrapfigure}{l}{0.32\textwidth}
\vspace*{-.2cm}
    \begin{lstfcgg}
type Box[$\alpha$ Any] struct { value $\alpha$ }
func (b Box[$\alpha$]) Nest(n int) Any {
  if (n > 0) { 
    return Box[Box[$\alpha$]]{b}.Nest(n-1) 
  } else { return b } 
}
    \end{lstfcgg}
\vspace*{-.3cm}
\caption{\inlinelstfcgg{Box} example\\\cite[Figure~10]{griesemer2020featherweight}}
\label{fig:example:nomono}
\vspace*{-.3cm}
\end{wrapfigure}

Intuitively, the fundamental issue with this 
deceptively simple program 
is that \emph{instance set discovery} is 
non-terminating. 
To monomorphise a program, 
we first need to discover all possible type instantiations used in 
said program. 
Perfectly well-behaved programs may however 
produce infinitely many type instantiations. 
This occurs when an instance of a (mutually) 
recursive method eventually depends upon a greater 
instantiation of itself, which in turn depends on an even
greater instantiation of itself \textit{ad infinitum}, \eg
\inlinelstfcgg{Box[int].Nest()} depends 
upon the specialisation \inlinelstfcgg{Box[Box[int]].Nest()}.
In \cite{griesemer2020featherweight}, such programs are called
\nomono.

\subsection{\gomacro Implementation}
The official release of \gomacro uses an optimised version of monomorphisation called
\emph{dictionaries and GC shape stenciling}~\cite{go118}.
When possible, their implementation reuses monomorphised functions to reduce code size. 
Two objects may share the same specialised method 
implementation when they have the same GC shape.
In the current implementation, the criteria of having the same 
GC shape means they are of the same data type, or both are pointers.
Each function therefore must to have a dictionary to differentiate
 concrete types at runtime.
A dictionary contains (1) the runtime type information of 
generic type parameters, as well as (2) their derived types used in the function. 
In the function body, each generic function call that 
depends on the generic type parameters also needs a dictionary;  
(3) these sub-dictionaries required by the method calls are also provided in the dictionary.
Additionally, the dictionary provides each generic object 
with (4) the data structure that Go runtime uses to conduct method calls.

\gomacro would also need to create an infinite call-graph dictionary 
for the \inlinelstfcgg{Box} example in Figure~\ref{fig:example:nomono},
as well as for 
the \inlinelstfcgg{permute} example in 
Figure~\ref{fig:code:list:perm}. Hence, \gomacro cannot handle either
example. Our call-site dictionary-passing approach does 
not suffer this limitation.

%% file: sections/dict.tex
\section{Call-Site, Non-Specialising Dictionary-Passing Translation}
\label{section:dictionary}
This section presents 
our new dictionary-passing translation from \gls{fgg} to \gls{fg}. 

\myparagraph{High level overview} 
Our call-site, non-specialising dictionary-passing translation can be 
split into a number of parts, each tackling a 
different challenge. Specifically, we consider:
the preservation of typeability, 
the use of dictionaries to resolve 
generic method implementations, 
the creation of dictionaries, and 
the preservation of  type assertion behaviour.
These challenges may have been discussed in other works, 
yet the structural type system of Go serves to 
hinder any existing solutions. 
We explain the key ideas and challenges in \S~\ref{sec:dictexample},
and detail the formal translation rules in \S~\ref{sec:dictexp}. 

\subsectSpace
\subsection{Dictionary-Passing by Example}
\label{sec:dictexample}

\subsubsection{\bf Structural Subtyping and Type Erasure}
\label{paragraph:structure}
The first challenge we encounter is that 
subtypes must be preserved. If, in
the source program, expression $e$ can 
be used as an argument to \inlinelstfcgg{Foo}, then 
the translation of $e$ should likewise
be usable as an argument to the translation of \inlinelstfcgg{Foo}. 
We should also desire that any non-subtypes 
are preserved, we leave this challenge to \S~\ref{subsubsec:typecollision}.

As a first naive attempt at removing 
polymorphic types, we might observe that 
regardless of the value we pass to a 
polymorphic argument, it must implement the \inlinelstfcgg{Any} type. 
From this, we could -- again, naively --
conclude that lifting all polymorphic 
arguments to the \inlinelstfcgg{Any} type solves our problem. 
Unfortunately, such a solution fails upon closer inspection. 
Consider the code in Figure~\ref{code:fgg:example}. 
By erasing the polymorphic types in
\inlinelstfcgg{Function}, we lose the subtype 
\inlinelstfcgg{GtFunc[int]} $<:$ \inlinelstfcgg{Function[int, bool]} 
(The naively erased \inlinelstfcgg{GtFunc} implements \inlinelstfcgg{Apply(in Any) bool},
while the erased 
\inlinelstfcgg{Function} demands \inlinelstfcgg{Apply(in Any) Any}). 
This issue is noted in \citet[\S~4.4]{Igarashi99FJ}.

Their solution, however, is inappropriate in a 
structurally typed language such as Go. 
In nominally typed languages like Java, 
it is clear that one type subtypes another. 
One need only inspect the implementing type’s 
declaration, as a subtype exists only when 
it is explicitly declared. 
\citet{Igarashi99FJ} insert \emph{bridge methods} to
handle cases such as the \inlinelstfcgg{GtFunc}-\inlinelstfcgg{Function} example. 
A bridge method is an overloaded method added 
to the subtype whose type matches the erased 
method as specified by the supertype, \ie{} adding an 
overloaded method of type \inlinelstfcgg{Apply(in Any) Any} to \inlinelstfcgg{GtFunc}. 
This method is immediately inappropriate as Go does 
not allow method overloading. 

The bridge method solution would still be inappropriate
were we to simulate 
overloading using name mangling.  
To add bridge methods, we need to 
know -- statically -- that a subtype exists. 
In \gls{fgg}, we need to know how two 
types are instantiated before we can conclude 
that a subtype relation exists. 
This requires the kind of potentially infinite 
whole program analysis (\S~\ref{section:nomono})
that we wished to avoid in our dictionary-passing translation. 
Instead, we ensure that subtypes are 
preserved by erasing \emph{all} method types, rather 
than just polymorphic types. 
As with \inlinelstfcgg{GtFunc}'s \inlinelstfcgg{Apply} 
method in Figure~\ref{code:fg:example}, 
when a variable
of a known type is used, we assert it to that 
type; although unlike Figure~\ref{code:fg:example}, the \gls{fgg} type 
checker has already ensured the safety of these synthetic assertions.

\subsubsection{\bf Dictionaries}
\label{subsubsec:dictionaries}

\begin{figure}
\begin{minipage}[t]{0.4\linewidth }
    \vspace{-3mm}  
\begin{center}
\begin{lstfcgg}
type Ord[T Ord[T]] interface { 
    Gt[](that T) bool 
}
type GtFunc[T Ord[T]] struct { val T }
func (this GtFunc[T]) Apply(in T) bool {
    return this.val.Gt[](in)
}
type Max struct {}
func (this Max) Of[T Ord[T]](l T, r T) T {
    $\cdots$ l.Gt(r) $\cdots$
}
func main() { GtFunc[int]{5}.Apply(7) }
\end{lstfcgg}
\end{center}
\end{minipage}\hspace*{-2mm}
\begin{minipage}[t] {0.59\linewidth }
    \vspace{-3mm}  
\lstset{firstnumber=1}
\begin{lstfcgg}
type Ord interface{  Gt(that Any) Any }
type OrdDict struct { 
    Gt func(rec Any, in Any) Any ; //Gt method pointer
    /*Simulated type*/ }
type GtFunc struct { val Any ; dict OrdDict }
func (this GtFunc) Apply(in Any) Any {
    return this.dict.Gt(this.val /*Receiver*/, in) (*\label{code:dictexample:resolve1}*)
}
func (this Max) Of(dict OrdDict, l Any, r Any) Any {
    $\cdots$ dict.Gt(l, r) $\cdots$ }(*\label{code:dictexample:resolve2}*)
func main() {
  od := OrdDict{Gt: func(rec Any, in Any) Any { rec.(int).Gt(in)}}
  GtFunc{5, od}.Apply(7) }
\end{lstfcgg}
\end{minipage}
\vspace*{-.5cm}
\caption{Dictionary-passing translation example extending 
Figure~\ref{code:fg:example}. 
\glsentryshort{fgg} source (Left),  
\glsentryshort{fg} translation (Right)}
\vspace*{-.4cm}
\label{code:dict:passing:example}
\end{figure}

We are now confronted with the primary 
challenge of concern to dictionary-passing 
translations; how do we resolve generic 
method calls without polymorphic type information? 
A dictionary is, at its simplest, 
a map from method names to their 
specific implementation for some type. 
A dictionary-passing translation, then, 
is one which substitutes the specialisation 
of type parameters with the passing of 
dictionaries as supplementary value-argument. 
One may then resolve a method call on 
a generic value by performing a dictionary lookup. 

Presently, we consider the structure 
and usage of dictionaries while delaying 
our discussion of call-site dictionary 
construction and type simulation until 
\S~\ref{subsubsec:dynamicdict} and \S~\ref{subsubsec:typecollision}, \emph{resp}.
Consider Figure~\ref{code:dict:passing:example} (left) extending a fragment of 
 Figure~\ref{code:fg:example} with a \inlinelstfcgg{Max.Of} method. 
 For us to call \inlinelstfcgg{Gt} in \inlinelstfcgg{GtFunc[T].Apply} 
 or \inlinelstfcgg{Max.Of[T]}, we need to know the 
 concrete type of \inlinelstfcgg{T}. This information 
is lost during erasure. 

The translation (right) includes a 
fresh struct \inlinelstfcgg{OrdDict} which is, 
quite naturally, the dictionary for 
\inlinelstfcgg{Ord} bounded type parameters. 
Dictionaries contain a method pointer 
field for each method in the original 
interface, along with a \emph{type-rep} which 
shall be discussed in \S~\ref{subsubsec:typecollision}. 
\gls{fg} does not include method pointers; 
instead, we must simulate them using
 higher order functions with the 
 first argument being the receiver. 
While this adds a small amount of 
complexity to the final correctness 
proofs, we see this as a worthwhile 
compromise, as it allows us to focus 
on the translation of generics alone, 
rather than on generics \emph{and} on
 a translation to some low level language. 
By containing each method specified
 by the \gls{fgg} bounding
interface, dictionaries have a fixed internal representation.
This reflects real-world dictionary-passing implementations
and allows entries to be accessed efficiently~\cite{driesen1996direct}.

\begin{wrapfigure}{r}{0.42\linewidth}
    \vspace*{-.7cm}
    \lstset{xleftmargin=5pt}
    \begin{lstfcgg}
type Foo[$\alpha$ Any] interface { 
    do[$\beta$ Any](a $\beta$, b bool) $\alpha$ 
} 
type Bar[$\alpha$ Any] struct {}
func (x Bar[$\alpha$]) do[$\beta$ Any](a $\beta$, b $\alpha$) int {$\cdots$}
func main() { 
    Bar[bool]{}.(Foo[int]); (*\label{code:assertion:source}*) 
    Bar[bool]{}.(Foo[bool]) (*\label{code:assertion:source2}*) 
}
    \end{lstfcgg}
    \vspace*{-.5cm}
    \caption{Type-rep example. \glsentryshort{fgg} source}
    \label{fig:code:dict:type-rep:fgg}
    \vspace*{-.3cm}
    \end{wrapfigure}
Dictionaries are passed to methods via 
two mechanisms, namely the method's receiver, 
and as regular value-arguments. 
Generic structures, \eg \inlinelstfcgg{GtFunc}, 
possess a dictionary for each type parameter.
When used as a receiver, these dictionaries can 
be accessed using standard field destructuring. 
Method dispatch then takes the form of a dictionary
lookup and method invocation as seen on lines~\ref{code:dictexample:resolve1}
and \ref{code:dictexample:resolve2} (right).

\subsubsection{\bf Type Collision}
\label{subsubsec:typecollision}

Here we consider the challenge of ensuring that 
type assertion behaviour is preserved by our translation.
Erasing type parameters may 
introduce new subtypes which did not
exist in the source program. 
Consider the expression 
\inlinelstfcgg{GtFunc[int]\{5\}.(Function[bool, bool])} 
where \inlinelstfcgg{GtFunc} and \inlinelstfcgg{Function} are 
defined in Figure~\ref{code:fgg:example}. 
Upon evaluation, this expression 
produces a runtime type assertion 
error as \inlinelstfcgg{GtFunc[int]\{5\}} is
not a subtype of \inlinelstfcgg{Function[bool, bool]}. 
The erased types as described in 
\S~\ref{paragraph:structure}, however, form a subtype relation,
meaning the error will not 
occur in the translated code. 
This behaviour would be incorrect. 
To ensure that type assertion errors 
are correctly preserved we simulate the FGG type
assertion system inside the 
translated \gls{fg} code via type-reps~\cite{crary1998intensional}. 
A simulated \gls{fgg} type implements \inlinelstfcgg{_type_metadata}
by specifying a method, 
\inlinelstfcgg{tryCast}, which throws an error 
if and only if the \gls{fgg} assertion would have failed.

\begin{figure}
    \lstset{xleftmargin=10pt}
    \begin{lstfcgg}
type _type_metadata interface { tryCast (in Any) Any }
type AnyDict struct {_type _type_metadata}
type Foo interface { do(dict$_0$ Anydict, in Any) Any ; spec_do() spec_metadata$_4$ }
type Foo_meta struct { _type$_0$ _type_metadata }
func (this Foo_meta) tryCast(x Any) Any { (*\label{code:trycast}*) // Type formal, Parametrised arg, Literal arg, return type $\alpha$
  if (x.(Foo).spec_do() $!=$  spec_metadata$_4${Any_meta{}, param_index$_0${}, Bool_meta{}, this._type$_0$ }) { panic } 
  return x }
type Bar struct {dict$_0$ AnyDict}
func (this Bar) spec_do() spec_metadata$_4$ { // Type formal, Parametrised arg, Arg type $\alpha$, return type literal
  return spec_metadata$_4${Any_meta{}, param_index$_0${}, this.dict$_0$._type, Int_meta{}}}(*\label{code:specdo}*)
func main() { 
  Foo_meta{Int_meta{}}.tryCast(Bar{AnyDict{Bool_meta{}}}) (*\label{code:assertion:target}*)
  Foo_meta{Bool_meta{}}.tryCast(Bar{AnyDict{Bool_meta{}}}) } (*\label{code:assertion:target}*)
    \end{lstfcgg}
    \vspace*{-.5cm}
\caption{Type-rep example. \glsentryshort{fg} translation}
\label{fig:code:dict:type-rep:fg}
 \vspace*{-.4cm}
\end{figure}

Consider the code in Figure~\ref{fig:code:dict:type-rep:fgg}.
The source \gls{fgg} code contains two assertions; 
the one on line~\ref{code:assertion:source} 
passes, while line~\ref{code:assertion:source2} 
produces a type assertion error.

A struct implements an 
interface when it correctly implements
each method specified by the interface. 
This means that not only does the struct
define a method of the same name, but 
also of precisely the same type. 
Assertion to an interface, then, need
only ensure that each method is correctly implemented. 
Assertion to a structure is a simple type equality check.

The translated interface, Figure~\ref{fig:code:dict:type-rep:fg}, 
now includes the meta method \inlinelstfcgg{spec_do}, 
returning simulated \gls{fgg} type information for a struct's 
\inlinelstfcgg{do} implementation.

The \inlinelstfcgg{spec_metadata$_4${}} object 
returned by \inlinelstfcgg{spec_do} on 
line~\ref{code:specdo} of the target code is a four-element 
tuple containing: type parameter bounds, 
argument types, and the return type. This object 
simulates the \gls{fgg} method type for 
\inlinelstfcgg{do} on \inlinelstfcgg{Bar[$\tau$]} for some $\tau$, \ie
\inlinelstfcgg{do[$\beta\ $ Any](a $\ \beta$, b $\ \tau$) Int[]}.
The first entry \inlinelstfcgg{Any_meta\{\}} gives the simulated 
type bound of the source method's type parameter $\beta$. 
The next gives the type of argument \inlinelstfcgg{a}, namely $\beta$. 
As there is no suitable concrete metadata type for $\beta$, we 
use an index \inlinelstfcgg{param_index$_0${}} to indicate
that \inlinelstfcgg{a}'s type is the method's first type parameter. 
The third, that of \inlinelstfcgg{b}, is not known at compile time, 
but is rather given by the
type parameter of the receiver. 
Finally, the return type is given by the constant \inlinelstfcgg{Int_metadata}.  

The type assertion on line~\ref{code:assertion:target}
uses the \inlinelstfcgg{Foo_meta}'s \inlinelstfcgg{tryCast} method defined on line~\ref{code:trycast}.
This method first checks that the erased types are compatible, \ie{} that \inlinelstfcgg{Bar} 
implements all erased methods in \inlinelstfcgg{Foo}. The \inlinelstfcgg{spec_do} method is then
used to check the simulated method type matches the interface specification.
If any of these checks is failed then the assertion 
fails and a {$\panic$} is thrown.

\subsubsection{\bf Call-Site Dictionary Creation}
\label{subsubsec:dynamicdict}

As discussed in \S~\ref{section:nomono}, 
the approach taken by \gomacro 
is fundamentally limited by its use of call-graph based 
dictionary construction. 
In contrast we consider the challenge of the call-site 
construction of dictionaries 
in a structurally typed language. 
Our approach 
overcomes the aforementioned limitation of
 \cite{griesemer2020featherweight} and \gomacro.

We note a few key facts. 
A $\tau$-dictionary provides all 
the methods specified by the 
type bound $\tau$, and we may build a 
dictionary for any specialising type 
which is a subtype of $\tau$. 
We can also use a type variable to 
specialise some other type variable as 
long as the bound of the later is a 
supertype of the former. In a translation 
this \emph{dictionary-supertyping} 
involves using a $\tau$-dictionary to 
build a potentially different $\sigma$-dictionary. 
In a nominally typed 
language the explicit, and fixed, 
hierarchy allows a dictionary-passing translation
to easily structure and construct dictionaries
according to the subtype hierarchy. 
Dictionary-supertyping in nominally typed languages 
is generally a 
matter of extracting the appropriate sub-dictionary~\cite{bottu2019coherence}. 

In a structurally typed language, however, there is not 
a fixed subtype hierarchy. Recall that in order
to infer subtype relationships, we first need the specific
type instances. We 
have two choices: either explore the entire 
call graph to discover all type instantiations and 
construct our dictionaries according to the call-graph, 
or construct/supertype our dictionaries at the call-site where
specialisation would have happened. 
The former approach was taken by \gomacro 
and beyond the significant static analysis 
required, this approach also suffers from the 
same finiteness limitation encountered 
by monomorphisation approaches \cite{griesemer2020featherweight}.

\begin{figure}
    \noindent
\begin{minipage}{0.33\linewidth}
    \begin{lstfcgg}
type Eq[$\alpha$ Eq[$\alpha$]] interface {
    Equal(that $\alpha$) bool
}
type Ord[$\alpha$ Ord[$\alpha$]] interface {
    Gt(that $\alpha$) bool; 
    Equal(that $\alpha$) bool
}
func Foo[$\beta$ Ord[$\beta$]](val $\beta$) Any {
    return Bar[$\beta$](val)
}
func Bar[$\beta$ Eq[$\beta$]](val $\beta$) Any {$\cdots$}
func main() { Foo[int](5) }
    \end{lstfcgg}
\end{minipage}
\begin{minipage}{0.64\linewidth}
    \begin{lstfcgg}
type EqDict struct { Equal func(rec Any, that Any) Any }
type OrdDict struct {
  Equal func(rec Any, that Any) Any ;
  Gt func(rec Any, that Any) Any  
}
func Foo(dict OrdDict, val Any) Any {return Bar(EqDict{dict.Equal}, val)}
func Bar(dict EqDict, val Any) Any { $\cdots$ }
func main() {
  old_dict := OrdDict{
    Equal : func(rec Any, that Any) Any { return rec.(int).Equal(that) }
    Gt : func(rec Any, that Any) Any { return rec.(int).Gt(that) } }
  Foo(old_dict, 5) }
    \end{lstfcgg}
\end{minipage}
\vspace*{-.3cm}
\caption{Call-site dictionary creation example. 
\glsentryshort{fgg} source (Left). 
\glsentryshort{fg} translation (Right)}
\label{fig:code:dict:dynamic}
\vspace*{-.4cm}
\end{figure}

We demonstrate our call-site approach in Figure~\ref{fig:code:dict:dynamic}.
This example consists 
of two interfaces, \inlinelstfcgg{Eq} and 
\inlinelstfcgg{Ord}, which form a 
subtype relation along with a method \inlinelstfcgg{Foo} which 
uses a type parameter bounded by \inlinelstfcgg{Ord} to 
instantiate a type parameter bounded by \inlinelstfcgg{Eq}. 

If, in the source program, there are two types $\sigma$ and $\tau$ 
where there exists an instantiation creating a subtype relation, then 
the two erased types form a subtype relation.
This is precisely the result discussed in \S~\ref{paragraph:structure}. 
When initially creating a dictionary, we
populate it with the required method pointers
for the known instantiating type.
If, however, we are creating a 
$\tau$-dictionary for type parameter $\beta$ bounded by $\sigma$, 
then the method contained by the supertyping 
$\tau$-dictionary (\inlinelstfcgg{Eq}) is a subset of 
the $\sigma$-dictionary (\inlinelstfcgg{Ord}) for type parameter $\alpha$. 
Dictionary-supertyping then consists of destructuring the 
subtype's dictionary and -- along with the type-rep --
adding all required method pointers to a 
new supertype-dictionary. 

While conceptually simple, 
our call-site approach directly addresses the 
unique issues raised by structural typing systems and 
allows us to 
overcome the limitation discussed in \S~\ref{section:nomono} 
that afflicts 
both monomorphisation~\cite{griesemer2020featherweight} and \gomacro.

\subsection{Dictionary-Passing Judgement}
\label{sec:dictexp}

This subsection is technical:  
readers who are not interested in the formal translation rules 
can safely skip this subsection.

We define the judgement $\dict[]{\prog}{\lex{\prog}}$ as the 
dictionary-passing translation from $\prog$ in \gls{fgg} to
$\lex{\prog}$ in \gls{fg}. 
The expression judgement $\dict{e}{\lex{e}}$ is parametrised by
variable and type variable environments 
($\Gamma$ and $\Delta$ \emph{resp.})
as well as a dictionary map $\eta$ from type variable names to 
dictionary variables. 
We provide auxiliary functions in Figure~\ref{fig:dict:aux} and translation rules in 
Figure~\ref{fig:dict:prog}. 

\myparagraph{Name constants}
We introduce a set of maps from name constants in \gls{fgg} to unique
 \gls{fg} names which are assumed 
to never produce a collision, 
    \begin{enumerate*}
        \item $\dictName{\iType{t}}$ --- from a type bound (interface) to the 
            dictionary struct name for that bound; 
        \item $\metadataName{t}$ --- from a type name to a simulated type name; 
        \item $\specName{m}$ --- from a method name to a method producing simulated specification; and
        \item  $\methName{t,m}$ --- the method applicator (pointer) for method $m$ on type $t$. 
    \end{enumerate*}

\myparagraph{Auxilary functions}
\input{figs/dict/aux}
Figure \ref{fig:dict:aux} provides a number of auxiliary functions used in 
the dictionary-passing translation. 
The overloaded $\arities{}$ function 
computes the number of type and value parameters required by each 
method signature, including method signatures in an interface's specifications. 
Function $\maxFormals{D}$ computes the number of type parameters expected 
by the largest type formal. 
Function $\asParameters{\Phi}$ converts a type formal into dictionary arguments.
The function $\makeDictMeth{t, mM}$ constructs the 
simulated method pointer struct and implementation 
-- called the \emph{abstractor/applicator pair} --  for method $m$ on type $t$.

To build a type simulation of type $\tau$ we call $\typemeta{\tau}$ where $\zeta$ is a 
map from type variables to existing simulated types. 
When simulating the type assertion to an interface in 
\S~\ref{subsubsec:typecollision}, we used the  
$\specName{m}$ method \inlinelstfcgg{spec_do} to produce
the instantiated simulated signature for method $m$. 
The $\specMetadata{mM}$ function takes an interface's method specification 
and produces the specification for the $\specName{m}$ method. 
Simulated method signatures are built using $\signatureMeta{M}$. 
This function takes a map $\zeta$ 
from type variables to simulated types, 
and extends $\zeta$ with the indexing structs for the method's type formal. 

A set of simulated method signature ($\fnMeta{n}$) and type parameter index 
($\paramTypeMeta_i$) structs are created by the program translation 
rule (\rulename{\dprogram}); $\arities{\multi{D}}$ and $\maxFormals{D}$ 
are used to ensure that all needed structs are constructed. $\fnMeta{n}$ is an $n+1$ tuple 
used in interface assertion simulation, and describes a method signature of arity $n$
and gives 
type parameter bounds, 
value argument types, and the method's return type. 
To allow interface assertion simulation to reference type 
variables, we use $\paramTypeMeta_i\{\}$ to reference a method's $i^{\text{\tiny th}}$ 
type parameter. 

Given a type environment $\Delta$ and a map $\eta$ from type variables to existing dictionaries,
we build a $\iType{\tau}$-dictionary for type~$\sigma$ using the 
\makeDict{\sigma, \iType{\tau}} function. In the case that $\sigma$ is already 
a type variable $\alpha$, then the map $\eta$ must contain a dictionary for $\alpha$. 
When $\alpha$ is bounded by $\iType{\tau}$ in $\delta$, we are done, 
whereas if $\iType{\tau}$ is a subtype of $\alpha$, 
but not 
$\iType{\tau} = \alpha$, then we need copy 
method pointers required by the new (and smaller) $\iType{\tau}$-dictionary.
A new dictionary is built for a constant type $\sigma$ by providing a method pointer (abstractor) for each 
method specified by $\iType{\tau}$ and the simulated type of $\sigma$.

\input{figs/dict/trans}

\myparagraph{Program translation} Rule \rulename{\dprogram}
introduces new declarations required for method pointers 
and type simulations as described in \S~\ref{sec:dictexample}, and 
the \any\ interface to provide a uniform, erased, type representation. 
Each method applicator must implement an $n$-arity function interface 
$\nAryFunction{n}$, that accepts the receiver and the $n$ arguments for the 
desired method call. The arity of a method includes both the regular value 
arguments as well as the dictionary arguments. 
A simulated type implements the $\typeMetadataLit$ interface by providing 
an assertion simulation method (\trycast), which panics if 
the assertion is invalid. 
The $\fnMeta{n}$ and 
$\paramTypeMeta_i$ structs are created as required by 
the $\arities{\multi{D}}$ and $\maxFormals{D}$ 
functions, respectively.
Each declaration is translated to multiple declarations; we use $\mathcal{D}$ 
to indicate this.

\myparagraph{Interface and dictionary construction}
The translation of 
interfaces produces a number 
of \gls{fg} declarations (\rulename{\dinterface}). 
They are (1) an \gls{fg} interface, and 
(2) a dictionary for that interface.

The interface $\iType{t}\typeFormalType$ becomes the erased type 
$\iType{t}$ (1). 
For each method specification $S$ defined by the source interface, we produce 
two specifications in the target; the first is defined by \rulename{\dspec} 
and replaces types formal with appropriate dictionaries while erasing all other 
types, and the second defines a method producing the simulated 
\gls{fgg} method specification. 
Since \rulename{\dmeth} produces such a simulated specification method for 
each method, it is guaranteed that any type which 
implements the former will implement the latter.

The dictionary (2) for an interface $\iType{t}$ is given by a new struct 
$\dictName{\iType{t}}$, which contains a method pointer (abstractor) for 
each specified method and the simulated type ($\typeField$) 
for the type parameter's specialising type.
Type simulation is also defined here.
For type $\iType{t}\typeFormalType$, the simulation  struct 
($\metadataName{\iType{t}}$) contains a field for each type parameter in $\Phi$. 
The \trycast\ method 
checks that each specified method is implemented correctly by the target of the 
assertion (See \S~\ref{subsubsec:typecollision}).
For clarity of presentation, we assume a number of extra language features that can 
be easily implemented in \gls{fg}, including; if-statement, struct inequality, 
explicit panic, and sequencing \cite{griesemer2020featherweight}.

\myparagraph{Struct declaration}
To translate $\sType{t}\typeFormalType$, 
we erase all field types 
and add a new dictionary field for each 
type parameter in $\Phi$. The simulated type $\metadataName{\sType{t}}$
is constructed with a
variable for each type parameter, and $\trycast$ checks that the 
target value is exactly the assertion type.

\myparagraph{Method declaration}
Judgement on method $m\typeFormalMethod(\multi{x~\tau})~\tau$ (\rulename{\dmeth})
produces a primary method,  a method returning the simulated method type, 
and an abstractor/applicator pair.
The primary method and simulation method's types match those from \rulename{\dinterface}.
The body of the implementing method is translated in the 
$\Delta;\eta;\Gamma$ environments,
where $\Delta$ and $\Gamma$ are built 
according to the typing system.  
There are two locations for type variables -- and thus dictionaries -- to be 
passed into a method, namely in the receiver or as an argument;
consequently, $\eta$ may map into either a dictionary argument ($\dictlit_i$) or 
a receiver's dictionary field ($\this.\dictlit_i$).

\myparagraph{Expressions}
The struct literal ($\sType{t}\typeActualReceive\{\multi{e}\}$) is translated by 
first translating each field assignment and then 
building an appropriate dictionary for each type in $\phi$ using 
\textit{makeDict} (\rulename{\dvalue}). 

Method calls are translated in one of two ways. The first (\rulename{\dcall})
is the immediate structural translation of sub terms and creation of appropriate dictionaries; this
translation is only possible if the type of the receiver is not a type variable, although 
it does not need to be a closed type. 
The second (\rulename{\ddictcall}) translates arguments and creates dictionaries in the 
same way as the former, but needs to resolve the 
method implementation using a dictionary lookup.

%% file: figs/dict/aux.tex
\begin{figure}
    \scriptsize
    \begin{equation*}
        \begin{gathered}
            \arities{\multi{D}} = \bigcup \multi{\arities{D}}
            \quad          \arities{\typeFormalMethod(\multi{x~\tau})~\tau} = |\Psi| + | \multi{x} |\
            \quad
            \arities{\type \iType{\tau} \interface{\multi{mM}}} =
            \multi{\arities{M}}
            \\
            \arities{\type \sType{\tau} \struct{\multi{f~\tau}}} = \emptyset 
            \quad 
            \arities{\func~(\this~\sType{\tau})~mM~\sytxBrace{\return e}}
            = \{\arities{M}\}
            \\
            \maxFormals{\funcDelc{\sType{t}\typeFormalReceive}{m\typeFormalMethod}{\multi{x~\tau}}{\tau}{\return e}} = |\Psi|
            \\
            \maxFormals{\type \iType{t}\typeFormalType \interface{\multi{m\typeFormalMethod(\multi{x~\tau})~\tau}}} = \max(\multi{|\Psi|}, |\Phi|)
            \quad
            \maxFormals{\type \sType{t}\typeFormalType \struct{\multi{f~\tau}}} = |\Phi|
            \\
            \typemeta{\alpha} = \zeta(\alpha)
            \quad
            \typemeta{t[\multi{\tau}]} = \metadataName{t}\sytxBrace{\multi{\typemeta{\tau}}}
            \quad
            \asParameters{\typeFormal[\alpha~\iType{t}\typeActualReceive]}
            = \multi{\dictlit~\dictName{\iType{t}}}\\
            \infer{
                \specMetadata{mM} = \lex{m}()~\fnMeta{n}
            }{
                n = \arities{M}
                & \lex{m} = \specName{m}
            }
            \\
            \infer{
            \signatureMeta{[\typeFormal[\beta~\iType{\sigma}]](\multi{x~\tau})~\tau} =
            \fnMeta{n}\{
            \multi{\typemeta[\zeta']{\iType{\sigma}}},
            \multi{\typemeta[\zeta']{\tau}},
            \typemeta[\zeta']{\tau}\}
            }{
            \zeta' = \zeta, \{\beta_i\mapsto\paramTypeMeta_i\{\}\}_{i}
            & n = \arities{[\typeFormal[\beta~\iType{\sigma}]](\multi{x~\tau})~\tau}
            }
            \\
            \infer{
                \makeDictMeth{t, m[\Psi](\multi{x~\tau})~\tau} =
                \left\{
                \begin{array}{l}
                    \type \methName{t, m} \struct{}; \\
                    \func (x~\methName{t, m})~\apply(
                    \rec~\any,
                    \multi{\dictlit~\any},
                    \multi{x~\any}
                    )~\any \{                        
                    \return \rec.(t).m(
                    \multi{\dictlit.(u)},~
                    \multi{x\vphantom{(u)}}
                    )~
                    \}
                \end{array}
                \right\}
            }{
                \multi{\dictlit~u} = \asParameters{\Psi}
            }
            \\
            \makeDict{\multi{\tau}, \typeFormal[\multi{\alpha~\sigma}]} = \multi{\makeDict{\tau, \sigma}}
            \,\
            \infer{
                \makeDict{\alpha, \iType{t}\typeActualReceive}
                = \sTypeInit{t}{\multi{\eta(\alpha).m},~
                    \eta(\alpha).\typeField}
            }{
                \subtype{\iType{t}\typeActualReceive}{\alpha}
                & \multi{mM} = \methods_\Delta(t\typeActualReceive)
                & \sType{t} = \dictName{\iType{t}}
            }
            \,\
            \infer{
                \makeDict{\alpha, \iType{\tau}}
                = \eta(\alpha)
            }{
                \Delta(\alpha) = \iType{\tau}
            }
            \\
            \infer{
                \makeDict{t[\phi], \iType{u}[\psi]}
                = \sType{t}\{\multi{
                    \methName{t, m}
                },~ \mathit{meta}\}
            }{
                \begin{array}{c}
                    \multi{m\typeFormalMethod(\multi{x~\tau})~\sigma} = \methods_\Delta(\iType{u}[\psi])
                    \quad \sType{t} = \dictName{\iType{u}}\quad
                    \zeta = (-.\typeField) \circ \eta
                    \quad  \mathit{meta} = \typemeta{t[\phi]}
                \end{array}
            }
        \end{gathered}
    \end{equation*}
    \vspace{-3mm}
    \caption{Dictionary-passing auxiliary function}
    \label{fig:dict:aux}
    \vspace{-6mm}
\end{figure}

%% file: figs/dict/trans.tex
\begin{figure}[hbtp]
    \scriptsize
    \begin{equation*}
        \begin{gathered}
            \namedRuleTwo{\dprogram}{
            \mathit{fns} = \{\type \nAryFunction{n} \interface{\apply(\rec~\any, \{x_i~\any\}_{i<n})~\any}\}_{n\in\multi{n}}
            }{
            \infer{
                \dict[]{\program{e}}{\program[\lex{D}]{\lex{e}}}
            }{
                \begin{gathered}
                    \multi{n} = \arities{\multi{D}}
                    \quad
                    \mathit{metas} = \{\type \fnMeta{n} \struct{\{\typeField_i~~\typeMetadataLit\}_{i\leq n}}\}_{n\in\multi{n}}
                    \quad m = \max~~\bigcup\multi{\maxFormals{D}}
                    \\\mathit{params} = \{\type \paramTypeMeta_i \struct{}\}_{i<m}
                    \quad
                    \mathit{typeMeta} = \{\type \typeMetadataLit \interface{
                        \trycast (x~\any)~\any} \}
                    \\
                    \dictMulti[]{D}{\mathcal{D}}
                    \quad \multi{\lex{D}} = \{\type \any \interface{}\}
                    \cup \mathit{params}
                    \cup \mathit{typeMeta}
                    \cup \mathit{metas}
                    \cup \mathit{fns} \cup \bigcup \multi{\mathcal{D}}
                    \quad \dict[\emptyset; \emptyset; \emptyset]{e}{\lex{e}}
                    \\
                \end{gathered}
            }
            }\\
            \namedRuleTwo{\dinterface}{
            \dictMulti[]{S}{\lex S}
            \quad \dictMultiDict{S}{S_{\text{dict}}}
            \quad \zeta = \{\multi{\alpha\mapsto\this.\typeField}\}
            }{
            \infer{
                \begin{gathered}
                    \dict[]{
                        \type \iType{t}[\typeFormal] \interface{\multi{S}}
                    }{
                        \left\{
                        \begin{aligned}
                             & \type \iType{t} \interface{\multi{\lex{S}},~\multi{\specMetadata{S}}}                     \\
                             & \type \dictName{\iType{t}} \struct{\multi{S_{\text{dict}}},~\typeField~~\typeMetadataLit} \\
                             & \type \metadataName{\iType{t}} \struct{
                            \{\typeField_i~~\typeMetadataLit\}_{i<|\alpha| }
                            }                                                                                            \\
                             & \func~(\this \metadataName{\iType{t}})~\trycast(x~\any)~\any \{
                            \multi{\mathit{assertions}}~;~\return x
                            \}                                                                                           \\
                             & \multi{\makeDictMeth{t, S}}
                        \end{aligned}
                        \right\}
                    }
                \end{gathered}
            }{
                &
                \multi{\mathit{assertions}} =
                \{
                \lit{if} (x.(\iType{t}).\specName{m}() \noteq
                \signatureMeta{M}
                )~ \sytxBrace{ \lit{panic}
                }
                ~|~
                \begin{matrix}
                    mM \in \multi{S}
                \end{matrix}
                \}
                }
            }
            \\[-.1cm]
            \namedRuleTwo{\dmeth}{
                (\type~\sType{t}[\typeFormal[\multi{\alpha~\iType{\tau}}]]~T) \in \multi{D}
                \quad 
                \eta = \multi{\alpha \mapsto \this.\dictlit}, \multi{\beta \mapsto \dictlit}
                \quad
                \lex{\Psi} = \asParameters{\typeFormal[\beta ~ \iType{\sigma}]}
            }{
                \infer{
                    \begin{gathered}
                        \dict[]{
                            \funcDelc{\sType{t}[\multi{\alpha}]}
                            {m[\typeFormal[\beta~\iType{\sigma}]]}
                            {\multi{x~\tau}}
                            {\tau}
                            {\return e}
                        }{
                            \left\{
                            \begin{aligned}
                                 & \funcDelc{\sType{t}}
                                {m}
                                {\lex{\Psi}, \multi{x~\any}}
                                {\any}
                                {\return \lex{e}}                                                              \\
                                 & \func~(\this~\sType{t})~\lex{m}\lex{M} \{
                                \return\signatureMeta{[\typeFormal[\beta~\iType{\sigma}]](\multi{x~\tau})~\tau}
                                \\
                                 & \makeDictMeth{t, m[\typeFormal[\beta~\iType{\sigma}]] (\multi{x~\tau})\tau}
                            \end{aligned}
                            \right\}
                        }
                    \end{gathered}
                }{
                    \begin{gathered}
                         \dict[\multi{\alpha~\iType{\tau}\vphantom{\beta}},
                            \multi{\beta~\iType{\sigma}};
                            \eta;
                            \this : \sType{t} {[\multi{\alpha}]},
                            \multi{x : \tau}
                        ]{e}{\lex{e}} 
                        \quad \zeta = \{\alpha_i \mapsto \this.\dictlit_i.\typeField\}_i
                        \quad \lex{m}\lex{M} =  \specMetadata{m[\typeFormal[\beta~\iType{\sigma}]](\multi{x~\tau})~\tau}
                    \end{gathered}
                }
            }
            \\[-.25cm]
            \namedRule{\dstruct}{
            \infer{
            \begin{gathered}
                \dict[]{
                    \type \sType{t}\typeFormalType \struct{\multi{f~\tau}}
                }{
                    \left\{
                    \begin{aligned}
                         & \type \sType{t} \struct{\multi{f~\any}, \multi{\dictlit~u}}     \\
                         & \type~\metadataName{\sType{t}} \struct{
                        \{\typeField_i~~\typeMetadataLit\}_{i<n}
                        }                                                                  \\
                         & \func~(\this~\metadataName{\sType{t}})~\trycast(x~\any)~\any \{
                        x.(\sType{t})~;~\multi{\mathit{assertions}}~;~\return x\}
                    \end{aligned}
                    \right\}
                }
            \end{gathered}
            }{
            \multi{\dictlit~u} = \asParameters{\Phi}
            & n = |\Phi|
            & \multi{\mathit{assertions}} = \{\lit{if} ~ \this.\typeField_i \noteq x.(\sType{t}).\dictlit_i.\typeField~\sytxBrace{\lit{panic}}\}_{i<n}
            }
            }
            \\
            \namedRule{\dfield}{
                \infer{
                    \dict{e.f}{\lex{e}.(\sType{t}).f}
                }{
                    \dict{e}{\lex{e}}
                    & \wellTyped[\Delta;\Gamma]{e}{\sType{t}\typeActualReceive}
                }
            } \!\!\!
            \namedRuleTwo{\dvalue}{
                (\type \sType{t}\typeFormalType~T)\in \multi{D}
            }{
                \infer{
                    \dict{
                        \sType{t}\typeActualReceive\sytxBrace{\multi{e}}
                    }{
                        \sType{t}\{\multi{\lex{e}}, \lex{\phi}\}
                    }
                }{
                    \quad \lex{\phi} = \makeDict{\phi, \Phi}
                    & \dictMulti{e}{\lex{e}}
                }
            } \!\!\!
            \namedRule{\dassert}{
                \infer{
                    \dict{e.(\tau)}{\typemeta[\zeta]{\tau}.\trycast(\lex{e})}
                }{
                    \dict{e}{\lex{e}}
                    & \zeta = (-.\typeField) \circ \eta
                }
            }\\
            \namedRuleTwo{\ddictcall}{
                \wellTyped[\Delta; \Gamma]{e}{\alpha}
                \quad (m\typeFormalMethod(\multi{x~\tau})~\tau) \in \methods[\Delta](\alpha)
            } {
                \infer{
                    \dict{
                        e.m\typeActualMethod(\multi{e})
                    }{
                        \eta(\alpha).m.\apply(\lex{e}, \lex{\psi}, \multi{\lex{e}})
                    }
                }{
                    \lex{\psi} = \makeDict{\psi, \Psi}
                    & \dict{e}{\lex{e}}
                    & \dictMulti{e}{\lex{e}}
                }
            } 
            \namedRuleTwo{\dcall}
            {
                \wellTyped[\Delta; \Gamma]{e}{t\typeActualReceive}
                \quad (m[\Psi](\multi{x~\tau})~\tau) \in \methods[\Delta](t\typeActualReceive)
            } {
                \infer{
                    \dict{
                        e.m\typeActualMethod(\multi{e})
                    }{
                        \lex{e}.(t).m(\lex{\psi}, \multi{\lex{e}})
                    }
                }{
                    \lex{\psi} = \makeDict{\psi, \Psi}
                    & \dict{e}{\lex{e}}
                    & \dictMulti{e}{\lex{e}}
                }
            }\\ 
                \namedRule{\dspec}{
                    \infer{
                            \dict[]{
                                m[\Psi](\multi{x~\tau})~\tau
                            }{} 
                            \quad m(\lex{\Psi}, \multi{x~\any})~\any
                    }{
                        \lex{\Psi} = \asParameters{\Psi}
                    }
                } 
                \
            \namedRule{\ddict}{
                \infer{
                        \dictDict{
                            m[\Psi](\multi{x~\tau})~\tau
                        }{}
                        m~\nAryFunction{n}
                }{
                    n = |\Phi| + |\multi{x}|
                }
            }\ 
            \namedRule{\dvar}{
                \axiom{\dict{x}{x}}
            }
        \end{gathered}
    \end{equation*}
    \vspace{-2mm}
    \caption{Dictionary-passing translation}
    \vspace{-0.4cm}
    \label{fig:dict:prog}
    \label{fig:dict:interface}
    \label{fig:dict:expr}
\end{figure}

%% file: sections/properties.tex
\sectSpace
\section{Correctness of Dictionary-Passing Translation}
\label{section:properties}

In this section, we define, justify, and prove 
the correctness of our dictionary-passing translation 
using a behavioural equivalence.
We first introduce a general 
\emph{correctness criteria} 
which good translations should satisfy. 
We then propose a novel \emph{bisimulation up to} technique
to prove that translated programs are behaviourally equivalent to their source program. 
We use this result to prove the correctness of our dictionary-passing translation. 
Full proofs can be found in
\ifnotsplit{Appendix~\ref{app:proofs}}{the full version of this paper~\cite{fullversion}}.

\subsection{Correctness Criteria}
\label{section:properties:correctness}
The correctness criteria is defined 
using a number of preliminary
predicates provided below.  

\begin{definition}[Type assertion errors]
  \label{def:panics}
  We say expression $e$ in \gls{fg} is a \emph{type assertion error}
  (\emph{panic} in \cite{griesemer2020featherweight})
  if there exists an evaluation
  context $E$, value $v$, and
  type $t$ such that
  $e=E[v.(t)]$
  and $\vtype(v)\not <: t$.
  We say expression $e$ gets
  a \emph{type assertion error} (denoted by
  $e\Downarrow_\panic$)
  if it reduces to an expression that contains a type assertion error,
  \ie{} $e\red^\ast e'$ and $e'$
  is a type assertion error. 
We write $\prog\Downarrow_\panic$ when  
$\prog = \program{e}$ and $e\Downarrow_\panic$.
Similarly, we define 
$e\Downarrow_\panic$ and $\prog\Downarrow_\panic$ for \gls{fgg}.
\end{definition}

We write $e\Downarrow v$ if there exists $v$ such that 
$e\red^\ast v$ and extend this predicate to $\prog$. 
We abbreviate $\dict[\emptyset;\emptyset;\emptyset]{e}{\lex{e}}$ to
$\dict[]{e}{\lex{e}}$. 

We define the following general correctness 
criteria 
related to typability, error correctness, 
and preservation of a program's final result.

\begin{definition}[Preservation properties]
\label{def:type:preservation}
Let $\prog \ok $ in \gls{fgg},
and let there exist $\lex{\prog}$ such that $\dict[]{\prog}{\lex{\prog}}$.
A translation is: 
\begin{enumerate}
   \item \textbf{\emph{type preserving}}: if
         $\prog \ok$, then $\lex\prog\ok$.
   \item \textbf{\emph{type assertion error preserving}}:
         $\prog\Downarrow_\panic$ iff $\lex{\prog}\Downarrow_\panic$.
   \item \textbf{\emph{value preserving}}:
         $\prog\Downarrow v$ iff $\lex{\prog}\Downarrow \lex{v}$ 
with 
$\dict[]{v}{\lex{v}}$.
\end{enumerate}
\end{definition}
We only require the left-to-right direction for 
type preservation, as due to type erasure
(\S~\ref{paragraph:structure}), 
we cannot obtain the right-to-left 
direction for dictionary-passing. 
Our type preservation criteria matches that defined in 
\citet[Theorem~5.3]{griesemer2020featherweight}.
We can, however, show
that type assertions are precisely 
simulated (\S~\ref{subsubsec:typecollision}).

\subsection{Behavioural Equivalence -- Bisimulation up to Dictionary Resolution}
\label{subsec:prop:beh}

\citet[Theorem 5.4]{griesemer2020featherweight} prove the correctness
of the monomorphism translation using 
a simple (strong) bisimulation:
the binary relation $\Re$ is a \emph{bisimulation} iff 
for every pair of $\ENCan{e,d}$ in $\Re$, where $e$ is a \gls{fgg} expression 
and $d$ is a \gls{fg} expression: 
\begin{enumerate*}
\item if $e \red e'$, then $d \red d'$ such that
  $\ENCan{e',d'}\in \Re$; and  
\item if $d \red d'$, then $e \red e'$ such that
  $\ENCan{e',d'}\in \Re$. 
\end{enumerate*}
This strong bisimulation suffices for translations that 
preserve a simple one-to-one reduction-step correspondence.

Unlike monomorphisation,  
dictionary-passing relies
on runtime computation, which prevents such a simple 
correspondence. 
We can, however, distinguish between reductions 
introduced by dictionary-passing and those 
inherited from the source program. This distinction allows 
us to construct
a one-to-one correspondence relation 
\emph{up to} dictionary resolution.  
The formulation is non-trivial since, in \gls{fg}, 
dictionary resolution can occur at any point in a subterm. 
\begin{wrapfigure}{r}{.67\linewidth}

\begin{minipage}[t]{0.37\linewidth }
\vspace{-4mm}
\begin{lstfcgg}
func foo[$\alpha$ Num](a $\alpha$) $\alpha$ {
  return a.Add(bar($\cdots$))
}
func main() {
  foo[Int](Zero{})
}
\end{lstfcgg}
\end{minipage}
\begin{minipage}[t]{0.62\linewidth }
\vspace{-4mm}
\lstset{firstnumber=1}
\begin{lstfcgg}
func foo(dict NumDict, a Any) Any {
  return dict.Add.Apply(a, bar($\cdots$)) }
type Int_Add struct {} // method pointer
func (i Int_Add) Apply(this Any, a Any) Any {
  return this.(Int).Add(a) }
func main() {
  foo(NumDict{Int_Add{}}, Zero{}) }
\end{lstfcgg}
\end{minipage}
\vspace{-3mm}
\caption{Non-trivial dictionary example. Source (Left). Translation (Right)}
\label{fig:example:nontriv}
\vspace{-3mm}
\end{wrapfigure}
We demonstrate this issue by evaluating the example in Figure~\ref{fig:example:nontriv}. 
Importantly, the translated
function \inlinelstfcgg{foo}
cannot resolve the generic \inlinelstfcgg{Add} method 
from dictionary \inlinelstfcgg{dict}
until \emph{after} expression \inlinelstfcgg{bar($\cdots$)}
is fully evaluated.

After one step, the \gls{fgg} program (left) is
\inlinelstfcgg{Zero<<>>.Add(bar($\cdots$))}.  
If we translate the afore reduced term, we get
\inlinelstfcgg{Zero<<>>.(Zero).Add(bar($\cdots$))} ($\lex{Q_0}$).
But reducing the translated \gls{fg} program (right), we obtain the 
\inlinelstfcgg{NumDict<<Int_Add<<>>>>.Add.Apply(Zero<<>>, bar($\cdots$))} ($\lex{Q_1}$).

To show $\lex{Q_0}$ equivalent to $\lex{Q_1}$ using 
the standard \gls{fg} reduction, 
we would first have to fully resolve 
\inlinelstfcgg{bar($\cdots$)} before we could start 
to the resolve dictionary 
access in $\lex{Q_1}$. 
We might attempt to show that the translation in Figure~\ref{fig:example:nontriv} 
is correct using a many-to-many reduction-step relation, \ie 
some binary relation $\Re$ where 
for every pair of $\ENCan{e,d}$ in $\Re$ it holds that
\begin{enumerate*}
\item if $e \red^\ast e'$, then $d \red^\ast d'$ such that
  $\ENCan{e',d'}\in \Re$; and  
\item if $d \red^\ast d'$, then $e \red^\ast e'$ such that
  $\ENCan{e',d'}\in \Re$. 
\end{enumerate*}
This approach is both complicated by the presence of non-termination, \eg
if \inlinelstfcgg{bar($\cdots$)} does not return a value,
then we could never show that
$\lex{Q_0}$ and $\lex{Q_1}$ are related. 
And more importantly, many-to-many relationships give less information 
about the nature of a translation 
than one-to-one relationships.

Were we to consider just the 
\inlinelstfcgg{NumDict<<Int_Add<<>>>>.Add.Apply($\cdots$)}
portion of $\lex{Q_1}$ we observe that 
using a pre-congruence reduction
$\lex{Q_1}$ resolves to \inlinelstfcgg{Zero<<>>.(Int).Add(bar($\cdots$))}. 
We may then safely increase the accuracy of the 
assertion \inlinelstfcgg{Zero<<>>.(Int)} to \inlinelstfcgg{Zero<<>>.(Zero)}
without altering the semantics of the term. 
The later step is required because while the 
dictionary stored the information 
that \inlinelstfcgg{Zero<<>>} was passed to \inlinelstfcgg{foo} 
as type \inlinelstfcgg{Int}, the reduction of the \gls{fgg} term 
forgot this 
information. 
We call these two steps \emph{dictionary resolution},
as they resolve only those computations introduced by the use 
of dictionaries for method resolution.  
$\lex{Q_0}$ is equivalent to $\lex{Q_1}$ 
\emph{up to dictionary resolution}.
Our translation also adds type simulation computations
and type assertions.
Unlike dictionary resolution,
these extra computation steps are subsumed by the 
standard \gls{fg} reduction.

\begin{definition}[Dictionary resolution] 
  \label{def:dictreso}
  We define three pattern sets in \gls{fg}:
$\erasepattern$   
  (type assertions as a result of erasure),
  $\assertpattern$ (type assertion simulation),
  and $\dictpattern$ (dictionary resolution):
\vspace{-2.2mm}
{\small
\begin{flalign*}
    \erasepattern ::=& \left\{\begin{array}{l}
      v.(t)
    \end{array}\right\}\\[-1.5mm]
    \assertpattern ::=& \left\{\begin{array}{l}
      v.\typeField_i,\
      v.\typeField,\
      v.(t),\
      v.\specName{m}(), \
      v.\dictlit_i, \
      \lit{if} ~ v \noteq v~\sytxBrace{\lit{panic}}, \
      \return v
    \end{array}\right\}&\\[-1.5mm]
\dictpattern  ::=& \left\{\begin{array}{l}
      \dictName{t}\{\multi{v}\}.f, \
      v.\dictlit_i, \
      v.\typeField_i,\                                 
      v.\typeField,\\
      \methName{t,m}\{\}.\apply(\multi{e}),\
      \dictName{t}\{\multi{v}\}.(\dictName{t})
    \end{array}\right\}
  \end{flalign*}
}
From these patterns, we define a number of reductions. 
We define the first of these as
$E[e]\prepre E[e']$
if $e\red e'$ with $e\in \erasepattern$; and   
$E[e]\prepostsim E[e']$
if $e\red e'$ with $e\in \assertpattern$.  
We write $d\dicttrans d'$ if 
$d\prepre^\ast \red \prepostsim^\ast d'\not\prepostsim$.

Let $\congEval$ be the context:
{$\congEval::=\hole \bnfsep  
\congEval.f \bnfsep
\congEval.(t) \bnfsep
\sTypeInit{t}{\multi{e}, \congEval, \multi{e}'} \bnfsep
\congEval.m(\multi{e}) \bnfsep 
e.m(\multi{e},\congEval,\multi{e}')$}.
We define the dictionary resolution reduction $\precongdict$ as
\begin{enumerate*}
  \item $\congEval[e]\precongdict \congEval[e']$ 
  if $\reduction{e}{e'}$ where $e\in\dictpattern$; and  
  \item $\congEval[e.(t)]\precongdict \congEval[e.(u)]$ if  
  $\wellTyped[]{e}{u}$ and $u<: t$.
\end{enumerate*}
\end{definition}
 
Notice that if $e \dicttrans e'$, then $e \red^+ e'$; and that   
$\dicttrans$ can be viewed as 
a one-step reduction which corresponds
to a one-step of the source language.
Reduction $\prepostsim$ only occurs following a
call to $\trycast$, and simulates whether or not the source \gls{fgg} assertion
is a type assertion error (See \S~\ref{subsubsec:typecollision}). 
The reduction $\prepre$ resolves only the assertions
introduced during the type erasure step (See \S~\ref{paragraph:structure}).
The dictionary resolution reduction $\dictpattern$ 
will occur following a method call
\rulename{r-call} and simulates the type parameter specialisation.
As demonstrated in the above example, 
the $\precongdict$ reduction may reduce any subterm matching
$\dictpattern$ or refine any type assertion.

\begin{restatable}{lemrest}{lemprec}
\label{lem:rec}
  Let $e$ be an \gls{fg} expression. 
  Assume $\wellTyped[\emptyset]{e}{u}$.
  \begin{enumerate}
      \item $\dicttrans$ is deterministic, \ie if $e \dicttrans e_1$ and
            $e \dicttrans e_2$, then $e_1=e_2$.
      \item $\precongdict$ is confluent, \ie if $e \precongdict e_1$ and
            $e \precongdict e_2$, then there exists $e'$ such that $e_1 \precongdict
                e'$ and $e_2 \precongdict e'$.
  \end{enumerate}
\end{restatable}

We now extend the bisimulation relation to 
bisimulation up to dictionary resolution.

\begin{definition}[Bisimulation up to dictionary resolution] \label{def:sb:upto}
\ \\ 
\begin{tabular}{ll}
\begin{tabular}{l}
The relation $\Re$ is 
a \emph{bisimulation up to dictionary}
\\ 
\emph{resolution} if 
$\Re\cdot (\dictredleft)^\ast$ is a bisimulation,
\\  
\ie 
if $\prog \ok $ in \gls{fgg}
  and $\dict[]{\prog}{\lex{\prog}}$ 
\\   
  where
  $\prog = \program{e}$ and
  $\lex{\prog} = \program[\lex{D}]{\lex{e}}$
\\  
then the diagram (right)
commutes. 
\end{tabular}
&
\hspace{-0.5cm}
\begin{tabular}{l}
\small
\input{figs/fig-bisim.tex}
\end{tabular}
\end{tabular}
\end{definition}

\hspace{.15cm}

Intuitively, our translation forms a bisimulation 
up to dictionary resolution
if 
\begin{enumerate*}
  \item each step that the source program takes can be mimicked
  by the translated program; and  
  \item conversely, that if the translated program 
  reduces,
  then the source program must have been able to make an equivalent step
\end{enumerate*}
-- albeit with the translated program still needing 
to evaluate the added dictionary resolution computations 
  at some future point during computation.

By considering the observable behaviour of a program 
to be non-dictionary resolution reduction steps, 
type assertion errors, and 
termination (value production), we ensure that the translated program 
is behaviourally equivalent to that of the source 
program. 
Note that  
this formulation may be extended to a concurrent or effectful
fragment of Go
with the standard addition of \emph{barbs}~\cite{MiSa92} or 
transition labels.

Finally, we arrive at our main theorem --- that the translation
satisfies the correctness criteria.  

\begin{restatable}[Correctness of dictionary-passing]{thmrest}{thmcorrect}
\label{thm:main:correctness}
Let $\prog \ok $ in \gls{fgg}
and $\dict[]{\prog}{\lex{\prog}}$ with 
$\prog = \program{e}$ and 
$\lex{\prog} = \program[\lex{D}]{\lex{e}}$.%
\begin{enumerate*}
  \item Dictionary-passing translation 
$\lex{(-)}$ is type preserving;
  \item $e$ and $\lex{e}$ are bisimilar up to dictionary resolution; 
  \item $\lex{(-)}$ is 
type assertion error
    preserving; and  
  \item $\lex{(-)}$ is value preserving. 
\end{enumerate*}
\end{restatable}

Theorem~\ref{thm:main:correctness}
states that 
our translation is correct, as 
translated programs behave exactly as the source program would have behaved,
and that any extra computations are accounted for by 
machinery introduced for dictionary-passing. 

It is worth stressing that
our statement
is \emph{stronger} than the various definitions of 
dictionary-passing translation correctness considered in the
literature (see \S~\ref{section:related}), which limit themselves 
to non-termination preserving versions of value preservation. 
By providing an account of intermediate state equivalence,
Theorem~\ref{thm:main:correctness}(2) not only gives a 
meaningful equivalence for non-terminating programs, but 
may also be extended to languages with non-determinism or concurrency.

\subsection{Proof of Theorem \ref{thm:main:correctness}}

We provide the key lemmata, theorems, and corollaries used in the proof 
of Theorem \ref{thm:main:correctness}. All omitted proofs 
may be found in
\ifnotsplit{Appendix~\ref{app:proofs}}{the full version of this paper~\cite{fullversion}}.

\myparagraph{Type preservation} 
The type preservation criteria given in Definition~\ref{def:type:preservation}
only considers whole programs. 
We must first show that the dictionary-passing translation 
is type preserving for expressions. Note that the translation of 
structure literals is 
the only non-\any\ typed expression.

\begin{restatable}[Type preservation of expressions]{lemrest}{lemtypepres}
  \label{lem:type:pres:exp}
  Let $\dict{e}{\lex{e}}$ and $\map{\Gamma}$ be the \gls{fg} environment 
  where all variables in $\Gamma$ are erased (\any) and each dictionary 
  in $\eta$ is appropriately typed according to the bound in $\Delta$.
    If $\wellTyped[\Delta;\Gamma]{e}{\tau}$ then 
    \begin{enumerate*}
        \item If $\tau = \alpha$ or $\iType{\tau}$,
              then $\wellTyped[\map{\Gamma}]{\lex{e}}{\any}$.
        \item If $\tau = \sType{t}\typeActualReceive$, then
              either $\wellTyped[\map{\Gamma}]{\lex{e}}{\any}$
              or  $\wellTyped[\map{\Gamma}]{\lex{e}}{\sType{t}}$.
    \end{enumerate*}
\end{restatable}

\begin{restatable}[Type preservation (Theorem \ref{thm:main:correctness} (1)]{correst}{corprogtypepres}
  \label{lem:type:pres:prog}
  If $\prog \ok$, then $\lex\prog\ok$.
\end{restatable}
\begin{proof}
  By the assumption that name constant functions are distinct and 
  Lemma~\ref{lem:type:pres:exp}.
\end{proof}

\myparagraph{Bisimulation and error preservation} 
The operational correspondence theorem described the behaviour 
of a source program and its translation as four non-overlapping 
cases. Note that $\lex{e} \dicttrans e'$ 
is the maximum reduction without another 
type assertion simulation reduction 
($e'\not\prepostsim$).

\begin{restatable}[Operational correspondence]{thmrest}{thmopcorrespond}
\label{thm:operational:correspondence}
Let $\prog \ok$ where $\prog = \program{e}$
and let $\dict[]{\program{e}}{\program[\lex{D}]{\lex{e}}}$.
\begin{enumerate}[(a)]
  \item If $\reduction{e}{d}$, then there
        exists $\lex{d}$ such that
        $\dict[\emptyset; \emptyset; \emptyset]{d}{\lex{d}}$ and
        $\lex{e} \dicttrans{\precongdict^\ast} \lex{d}$.
  \item If $\lex{e} \dicttrans e'$ where $e$ is not a type assertion error,
        then there exists $d$ such that
        $\reduction{e}{d}$ and there exists $\lex{d}$ such that
        $\dict[\emptyset; \emptyset; \emptyset]{d}{\lex{d}}$ and
        $e' \precongdict^* \lex{d}$.
  \item
        If $\lex{e} \dicttrans e'$ where $e$ is a type assertion error,
        then $e'$ is a type assertion error.
  \item
        If $e$ is a type assertion error, then there exists an $e'$
        such that $\lex{e} \dicttrans e'$ and $e'$ is a type assertion error.
\end{enumerate}
\end{restatable}
\begin{proof}
  By induction over the assumed reduction. Full proof is provided in
  \ifnotsplit{Appendix~\ref{app:proofs}}{the full version of this paper~\cite{fullversion}}.
\end{proof}

\begin{restatable}[Bisimulation up to dictionary resolution (Theorem \ref{thm:main:correctness} (2))]{correst}{corbisim}
  \label{cor:bisim}
  Let $\prog \ok $ 
  and $\dict[]{\prog}{\lex{\prog}}$ with
  $\prog = \program{e}$ and
  $\lex{\prog} = \program[\lex{D}]{\lex{e}}$.
  Then $e$ and $\lex{e}$ are bisimilar up to dictionary resolution.
\end{restatable}
\begin{proof}
  By Theorem \ref{thm:operational:correspondence}. 
  Let $\Re$ be the least relation such that all source expressions 
  are paired with their translation. 
  $\Re$ is a bisimulation up to dictionary resolution. 
  Namely, for each element $\ENCan{e,\lex{e}}\in \Re$, we have that:
  \begin{enumerate}
      \item If $e\red e'$, then by Theorem~\ref{thm:operational:correspondence} (a) 
      there exists a $\ENCan{e',d}\in \Re$ such 
      that $\lex{e} \dicttrans{\precongdict^\ast} d$. 
      \item If $\lex{e} \dicttrans{\precongdict^\ast} d$, then by 
      Theorem~\ref{thm:operational:correspondence} (b)  
      there exists a $\ENCan{e',d}\in \Re$ such 
      that $e\red e'$. 
  \end{enumerate}%
\end{proof}

\begin{restatable}[Type error preservation (Theorem \ref{thm:main:correctness} (1))]{correst}{corerrorpres}
  \label{cor:error:preservation}
  Let $\program \ok$ and $\dict[]{P}{\lex{P}}$.
  $P\Downarrow_\panic$ iff $\lex{P}\Downarrow_\panic$.
\end{restatable}
\begin{proof}
  For this proof, we define $\lex{P}$ as resolving into a type assertion error
  if $\lex{P}\dicttrans P'$ and $P'$ is a type assertion error. This happens
  when $P$ is a type assertion error, as in Theorem~\ref{thm:operational:correspondence} (c) and (d).
  By induction on the reductions in $\Downarrow$.
  \begin{itemize}
      \item[] \caseStd Left to right (base):
            By Theorem~\ref{thm:operational:correspondence} (d).
      \item[] \caseStd Right to left (base): 
            By Theorem~\ref{thm:operational:correspondence} (c). 
      \item[] \caseStd Left to right (induction): \\
            If $P$ is not a type assertion
            error, then it reduces to $Q$ where $Q\Downarrow_\panic$.
            By Theorem~\ref{thm:operational:correspondence} (a) $\lex{P}\dicttrans\prepostdict\lex{Q}$
            where $\dict[]{Q}{\lex{Q}}$.
            Apply induction on if $Q\Downarrow_\panic$ then $\lex{Q}\Downarrow_\panic$.

      \item[] \caseStd Left to Right (induction): \\
            We assume that $\lex{P}$ does not resolve into a type assertion error,
            \ie $\lex{P}\dicttrans Q'$ where $Q'$ is not a type assertion error.
            Since $\prepostdict$ cannot cause a type assertion
            error, we also get that $Q' \prepostdict^* \lex{Q}$ where $\lex{Q}$
            is not a type assertion error.
            By  Theorem~\ref{thm:operational:correspondence} (b) $P\red Q$.
            Apply induction on if $\lex{Q}\Downarrow_\panic$ then $Q\Downarrow_\panic$.
  \end{itemize}
\end{proof}

\myparagraph{Value preservation}
Finally, the value preservation property follows dictionary-passing 
being a bisimulation up to dictionary resolution, as the  
dictionary resolution steps are eager reductions that can equivalently 
be delayed until they become standard reductions.

\begin{restatable}[Reduction rewrite]{lemrest}{lemredrewrite}
  \label{lem:red:rewrite}
  Let $e_1\dictred e_2 \red e_3$ where $e_1=C[d_1]$,  $e_2=C[d_2]$, and $d_1\red d_2$.
  \begin{enumerate}
      \item If there exists an $E$ such that $C=E$ then $e_1 \red^2 e_3$
      \item If there does not exists an $E$ such that $C=E$ then $e_1 \red \dictred e_3$
  \end{enumerate}
\end{restatable}

\begin{restatable}[Resolution to value]{lemrest}{lemredvalue}
  \label{lem:red:val}
  If $e\dictred v$ then $e\red v$.
\end{restatable}

\begin{restatable}[Value preservation (Theorem \ref{thm:main:correctness} (4))]
  {correst}{corvalue}
  \label{cor:valpres}
  Let $\program \ok$ and $\dict[]{P}{\lex{P}}$.
  $P\Downarrow v$ iff $\lex{P}\Downarrow \lex{v}$ where
  $\dict[]{v}{\lex{v}}$.
\end{restatable}

\tikzset{|/.tip={Bar[width=.8ex,round]}}
\begin{proof}
  By Corollary~\ref{cor:bisim} we have
  the following diagram (where $\Re$ is created by $\Mapsto$)
  \[
      \begin{tikzpicture}[line width=rule_thickness,
              arrowlabel/.style={inner sep=.5,fill=white},
          ]
          \node (dagone) [] {$\lex{e}_1$} ;
          \node (dagtwo) [right=1 of dagone] {$\lex{e}_2$} ;
          \node (dagthree) [right=1 of dagtwo] {$\lex{e}_3$} ;
          \node (dagdots) [right=1 of dagthree] {$\cdots$\vphantom{$\lex{e}_3$}} ;
          \node (dagvee) [right=1 of dagdots] {$\lex{v}$\vphantom{$\lex{e}_3$}} ;

          \node (eone) [above=.4 of dagone] {$e_1$} ;
          \node (etwo) [above=.4 of dagtwo] {$e_2$} ;
          \node (ethree) [above=.4 of dagthree] {$e_3$} ;
          \node (edots) [above=.4 of dagdots] {$\cdots$\vphantom{$e_3$}} ;
          \node (evee) [above=.4 of dagvee] {$v$\vphantom{$e_3$}} ;

          \draw[->] (eone) to (etwo);
          \draw[->] (etwo) to (ethree);
          \draw[->] (ethree) to (edots);
          \draw[->] (edots) to (evee);
          \coordinate (onetwo) at ($ (dagone) !.5! (dagtwo) $);
          \draw[-{Implies},double] (dagone) to (onetwo);
          \draw[-{Latex[open]}] (onetwo) to node[very near end, yshift=.8mm] {${}^*$} (dagtwo);
          \coordinate (twothree) at ($ (dagtwo) !.5! (dagthree) $);
          \draw[-{Implies},double] (dagtwo) to (twothree);
          \draw[-{Latex[open]}] (twothree) to node[very near end, yshift=.8mm] {${}^*$} (dagthree);
          \coordinate (threedots) at ($ (dagthree) !.5! (dagdots) $);
          \draw[-{Implies},double] (dagthree) to (threedots);
          \draw[-{Latex[open]}] (threedots) to node[very near end, yshift=.8mm] {${}^*$} (dagdots);
          \coordinate (dotsvee) at ($ (dagdots) !.5! (dagvee) $);
          \draw[-{Implies},double] (dagdots) to (dotsvee);
          \draw[-{Latex[open]}] (dotsvee) to node[very near end, yshift=.8mm] {${}^*$} (dagvee);

          \draw[|-{Implies},double] (eone) to (dagone);
          \draw[|-{Implies},double] (etwo) to (dagtwo);
          \draw[|-{Implies},double] (ethree) to (dagthree);
          \draw[|-{Implies},double] (evee) to (dagvee);
      \end{tikzpicture}
  \]
  By Lemma~\ref{lem:red:rewrite} and \ref{lem:red:val}
  each dictionary resolution reduction $\dictred$ is
  either subsumed by $\red$ or may be delayed using
  reduction rewriting until
  it becomes a $\red$ reduction.
  In other words, since 
  $e_1 \red e_2 \red \cdots \red v$ iff 
$\lex{e_1} \dicttrans\dictred \lex{e_2} \dicttrans\dictred \cdots
\dicttrans\dictred \lex{v}$.   
We use that 
$\dictred$ can be 
delayed ($d \dictred\dicttrans d'$ implies $d \dicttrans\dictred d$ 
or $d \red\dicttrans d$),
hence
$\lex{e_1} \dicttrans^+\dictred^+ \lex{v}$.   
Finally, from  $e\dictred^+ v$ implies $e\red^+ v$, we have that 
$e_1\Downarrow v$ iff 
$\lex{e_1}\Downarrow \lex{v}$.

\end{proof}

Proof of Theorem~\ref{thm:main:correctness} is given by 
Corollary~\ref{lem:type:pres:prog},
\ref{cor:bisim},
\ref{cor:error:preservation}, and 
\ref{cor:valpres}.

%% file: figs/fig-bisim.tex
\newcommand{\vheightmax}{\vphantom{$\ENCan{d,\lex{d}}\in\Re$}}
  \begin{tikzpicture}[line width=rule_thickness,
      arrowlabel/.style={inner sep=.5,fill=white},
    ]
      \node (a) [] {$e$ \vheightmax};
      \node (b) [below=.5 of a] {$d$ \vheightmax};
      \node (bsone) [right=.7 of a] {$\ENCan{e,\lex{e}}\in\Re$ \vheightmax};
      \node (bstwo) [below=.5 of bsone] {$\ENCan{d,\lex{d}}\in\Re$ \vheightmax};
      \node (d) [right=.7 of bstwo] {$\lex{d}$ \vheightmax};
      \node (e) [right=.7 of d] {$d'$ \vheightmax};
      \node (c) [above=.5 of e] {$\lex{e}$ \vheightmax};
      \draw[->] (a) to (b);
      \draw[->,dotted] (bsone) to node [arrowlabel,xshift=.5mm,yshift=1.2mm] {\footnotesize $\pi_1$} (a);
      \draw[->,dotted] (bstwo) to node [arrowlabel,xshift=.5mm,yshift=1.2mm] {\footnotesize $\pi_1$} (b);
      \draw[->,dotted] (bsone) to node [arrowlabel,xshift=-.4mm,yshift=1.2mm] {\footnotesize $\pi_2$} (c);
      \draw[->,dotted] (bstwo) to node [arrowlabel,xshift=-.4mm,yshift=1.2mm] {\footnotesize $\pi_2$} (d);
      \draw[-{Latex[open]}] (e) to node[very near end, yshift=.8mm] {${}^*$} (d);
      \draw[-{Implies},double] (c) to (e);
  \end{tikzpicture}

%% file: sections/experiment-0.tex
\sectSpace

\section{Implementation and Evaluation}
\label{sec:exp}
Beside the \gls{dict},
we also implement an \gls{erasure}. 
We compare the two implementations with three existing translators: \gls{mono}~\cite{griesemer2020featherweight},  \gls{gotogo} 
(the initial prototype based on a source-to-source monomorphisation), 
and \gomacro~\cite{go118} (the official generic type 
implementation released on 15th
March 2022).
This section first discusses the two implementations, 
then describes the evaluation methodology, 
before finally presenting the evaluation results.

\subsection{Implementation of Dictionary-Passing Translation}
\label{subsec:imple}

\input{figs/eval/fig-venn.tex}

We implement the dictionary-passing translator 
(\glsentryshort{dict}) 
and the erasure-based translator (\glsentryshort{erasure}) based on 
the \gls{fgg} artifact~\cite{fgg-artifact} in Go 1.16.
We have fully tested the implementations 
using designed unit tests. 
Figure~\ref{fig:venn} shows the code coverage difference across the five translators. 

\gls{fgg} is the calculus presented in
            \cite{griesemer2020featherweight}; 
            \glsentryshort{dict} does not
            cover receiver type formal subtyping; 
            \glsentryshort{erasure} does not cover \gls{fgg} type assertions; 
            \glsentryshort{mono} does not cover a class 
            of recursive (\nomono) programs
            \cite{griesemer2020featherweight}; 
            \glsentryshort{gotogo} is a 
            source-to-source monomorphisation translator implemented by the Go Team, 
and does not cover \textit{F}-bounded polymorphism, method parametrisation,
            receiver type formal subtyping,
            or recursive (\nomono) programs; and
            \gomacro is the
            official release with generics and has the same limitations
          as \glsentryshort{gotogo}. Both \gomacro and \gls{gotogo}
          target the full Go language, including features
          not considered by \gls{fgg}.

We implement \glsentryshort{dict} following the rules in 
\S~\ref{section:dictionary}. 
Rather than strictly follow the formalisations of \gls{fg} and \gls{dict} 
translation, we leverage the first-class functions support in Go and use 
function types~\cite{go-function-types} as dictionary fields, similar 
to using function pointers in C/C++. 
We also ignore unnecessary type assertions in \rulename{\dfield}
and \rulename{\dcall} when the translation is not on an interface.
We memorise expression typing results to accelerate compilation.
We exclude type simulation (\S~\ref{subsubsec:typecollision}) 
of non-generic types (\ie{} the size of 
the type formal is zero), and directly use type assertion for \rulename{\dassert} 
for better runtime performance. We also find if those type metadata are used, and remove them when possible. Users can also disable all 
type metadata copies if there are no type assertions in the input program.
In total, \glsentryshort{dict} contains 1160 lines of Go code.

\glsentryshort{erasure} is an alternative homogeneous translation implementation from \gls{fgg}. 
This implementation erases generic type information and uses the underlying interface type, similar to the erasure implementations for Java~\cite{odersky2000two, Igarashi99FJ}.
When calling a method,
the erased object is directly used as the receiver (If 
$\wellTyped[\Delta, \alpha{:}\iType{t}\typeActualReceive; \Gamma]{e}{\alpha}$ then $\dict[\Delta, \alpha{:}\iType{t}\typeActualReceive; \Gamma]{
    e.m\typeActualMethod(\multi{e})
}{
    \lex{e}.(\iType{t}).m(\multi{\lex{e}})
}$), in contrast to \glsentryshort{dict}'s dictionary lookup (\rulename{\ddictcall}). 
For example, \inlinelstfcgg{func f[a Foo](x a) \{x.Bar()\}} translates to \inlinelstfcgg{func f(x Any) \{x.(Foo).Bar()\}}, while \glsentryshort{dict} calls the corresponding function in a dictionary field. 
As in \S~\ref{paragraph:structure}, naively erasing type parameters breaks type assertion preservation (Definition~\ref{def:type:preservation}).
An example of \glsentryshort{erasure} is provided in
\ifnotsplit{Appendix~\ref{sec:erasure-example}}{the full version of this paper~\cite{fullversion}}.
Compared with \glsentryshort{dict}, \glsentryshort{erasure} 
provides a concise translation of generics that is fully based 
on Go's existing dynamic dispatch mechanism.  
When calling a method of a generic object as though 
it were an interface, the Go runtime looks up 
the actual method to call from a list of methods~\cite{go-interface-slow-1,go-interface-slow-2}, 
while \glsentryshort{dict} finds the actual method from the dictionary.
The implementation of \glsentryshort{erasure} contains 765 lines of Go
code.

%% file: figs/eval/fig-venn.tex
\usetikzlibrary{math} 

\begin{wrapfigure}{r}{.42\linewidth }
    \vspace*{-.3cm}
    \centering
    \tikzset{every node/.style={rounded corners}}
    \scalebox{0.8}{
    \small
    \begin{tikzpicture}
        \node[draw, fill=white] at (1.5,1.5) (fgg)  {
            \begin{minipage}[t][3cm]{3cm}
                \glsentryshort{fgg} 
        \end{minipage}
        };
        \node[draw, fill=yellow, fill opacity=0.1, text opacity=1] at (1.1,1.2) (fgg)  {
            \begin{minipage}[t][2.4cm]{2.2cm}
                \glsentryshort{dict} 
                (\S\ref{section:dictionary})
\ {\color{red}{\cmark}}            \end{minipage}
        };        

        \node[draw, fill=blue, fill opacity=0.1, text opacity=1] at (1.3, 0.9) (fgg)  {
            \begin{minipage}[t][1.8cm]{2.6cm}
                \glsentryshort{erasure} 
                \ \xmark

            \end{minipage}
        };        

        \node[draw, fill=green, fill opacity=0.1, text opacity=1] at (1.5, 0.6) (fgg)  {
            \begin{minipage}[t][1.2cm]{3cm}
                \glsentryshort{mono}
\ {\color{red}{\cmark}}
            \end{minipage}
        };        

      \node[draw,  fill=red, fill opacity=0.1, text opacity=1] at (1, 0.15) (fgg)  {
            \begin{minipage}[t][0.9cm]{2cm}
                \glsentryshort{gotogo} 
                \ \xmark\\
                \gomacro
                \ \xmark
            \end{minipage}
        };

\end{tikzpicture}
}
\vspace*{-.3cm}
    \caption{Relationship of implementations. {\color{red}\cmark} denotes 
    a translation proven correct (Theorem~\ref{thm:main:correctness}); 
    \xmark\ denotes a translation that has not.}
    \vspace*{-.3cm}
    \label{fig:venn}
\end{wrapfigure}
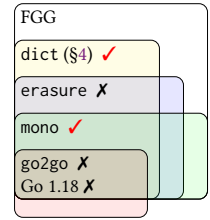

%% file: sections/experiment-1.tex
\subsectSpace
\subsection{Evaluation Methodology}
\label{subsec:evaluation}

\myparagraph{Benchmarks}
We build two benchmark 
suites to conduct head-to-head comparisons for the 
five translators.
\textit{1) Micro Benchmarks:} 
we design five micro benchmark sets. 
Each has a configuration parameter
to demonstrate how the translated code scales with a particular 
aspect of \gls{fgg}/Go programs. 
\textit{2) Real-World Benchmarks:}  
we reimplement all benchmarks
in previous papers about generics in
Java and Scala~\cite{odersky2000two, ureche2013miniboxing}. 
\gomacro officially released generics on March 15th, 2022, and 
it is infeasible for us to find usage of generics in real Go programs.
The second benchmark suite is a reasonable substitute to reveal how the five 
translators behave in reality.

\input{figs/eval/fig-toy2-sep}

\myparagraph{Micro Benchmarks} 
The five sets of micro benchmarks, \benchname{Program}~\mycircledtext{a}-\mycircledtext{e}, are all derived 
from a base program. Figure~\ref{fig:benchmark-prog-pdf} shows the base program and how the five 
benchmark sets are derived from it.
In the base program, 
lines 29--32 enumerate all possible combinations of types actual for \gocode{f$_1$()}.
Function \gocode{f$_1$()} takes two parameters and uses them 
to call \gocode{f$_2$()} on line~20, which in turn calls \gocode{CallBase()} 
on line~14. 
Function \gocode{CallBase()} calls \gocode{Ops()} on line 5, which further calls \gocode{Op()} 
twice to represent two non-generic operations. All methods of interface \gocode{Base} 
(\gocode{g$_1$()} and \gocode{g$_2$()}) are implemented by struct \gocode{Derived}, 
and called on line 17, from receiver variable \gocode{x} 
with generic type \gocode{base}. 
Function \gocode{main()} calls \gocode{DoIt()} 10,000 times (line~36)
to provide stable performance results. 

The set of \benchname{Program}~\mycircledtext{a} extends the number of methods of \gocode{Base} (lines 39--42) 
and \inlinelstfcgg{Derived} (lines 45--47) in the base program, from 2 to $n$.
\benchname{Program}~\mycircledtext{b} repeats the non-generic operation $c$ times on line 56, instead of two. 
In \benchname{Program}~\mycircledtext{c}, we increase the number of type parameters from 2 to $m$ (lines 59, 60, 63, and 65), 
and enumerate all $2^m$ type actual combinations (lines 67--70).
\benchname{Program}~\mycircledtext{d} increases the length of the call chain between \gocode{Doit()} 
and \inlinelstfcgg{CallBase()} from 2 to $p$ (lines 49--55).
\benchname{Program}~\mycircledtext{e} is particularly designed to expose the exponential complexity 
of monomorphisation (lines 72--88).
Its configuration parameter $m$ controls both
the type parameter number of \gocode{Base} (and \gocode{Derived}) and 
the number of functions called in between \gocode{DoIt()} 
and \gocode{BaseCall()} along the call chain. 
For the $m$ functions in between \gocode{DoIt()} 
and \gocode{BaseCall()}, we further configure each caller to call its callee twice,
and each callee to have one more parameter than its caller (e.g., function body of \inlinelstfcgg{f$_1$} and \inlinelstfcgg{f$_2$} on lines 77--84).

\myparagraph{Real-World Benchmarks}
We reimplement the Java and Scala programs using 
\glsentryshort{gotogo}, \gomacro, and \glsentryshort{fgg} for our evaluation. 
Since \glsentryshort{fgg} does not support all syntax  in the programs, 
we first use \glsentryshort{fgg}
to reimplement as many functionalities as possible. Then, 
we translate the \glsentryshort{fgg} code to Go
and manually insert the missed non-generic functionalities. 
On the other hand, 
\glsentryshort{gotogo} and \gomacro support all required syntax, so 
we use them to reimplement each whole program.
We manually test the reimplementations with designed testing inputs
and compare their outputs with the original versions in Java or Scala. 
Our tests achieve 100\% code coverage.

The benchmarks' functionalities are explained as follows. 
\benchname{List} \cite{ureche2013miniboxing} 
is an implementation of a linked list. It supports insert and search operations
on the linked list. 
\benchname{ResizableArray} \cite{ureche2013miniboxing} implements a resizable array.
It inserts elements into the array, 
reverses the array, and searches elements in the array. 
\benchname{ListReverse} \cite{odersky2000two} constructs a linked list and reverses it. 
It contains two reversing implementations.
\benchname{VectorReverse} \cite{odersky2000two} is to reverse an array. 
Similarly, it implements the reversing
functionality in two different ways.
\benchname{Cell} \cite{odersky2000two}
implements a generic container.
\benchname{Hashtable} \cite{odersky2000two} accesses elements in a hash table.

\myparagraph{Metrics}
We consider \emph{code size}, \emph{execution time}, and \emph{compilation time} as our metrics.
For code size, we compile each translated benchmark program into a binary executable and
disassemble the executable using objdump~\cite{objdump120:online}.
Next, we count the number of assembly instructions compiled from the benchmark program as its code size, while excluding the assembly instructions 
of linked libraries.   
To measure execution time, 
we compile each translated \gls{fg} program using
the Go compiler and compute the average execution time over {\em ten} runs.
We consider the time spent on the source-to-source translation and the compilation 
from a \gls{fg} program to an executable as the compilation time for the four source-to-source translators. For \gomacro, we measure its compilation time directly. 
We compile each benchmark program with each translator {\em ten} times 
and report the 
average compilation time.

\myparagraph{Platform \& Configurations}
All our experiments are conducted on a desktop machine, 
with AMD Ryzen 5 2600 CPU, 32GB RAM, and Ubuntu-18.04. 
To focus more on the impact of different translations for generics,
we disable garbage collection and compiler optimisations for all translators.
No benchmark requires type simulation.
Thus, we disable this option in \gls{dict}, 
allowing us to better understand the impact of method translation and dispatch.

%% file: figs/eval/fig-toy2-sep.tex
\begin{figure}[t]
\includegraphics[width=\textwidth]{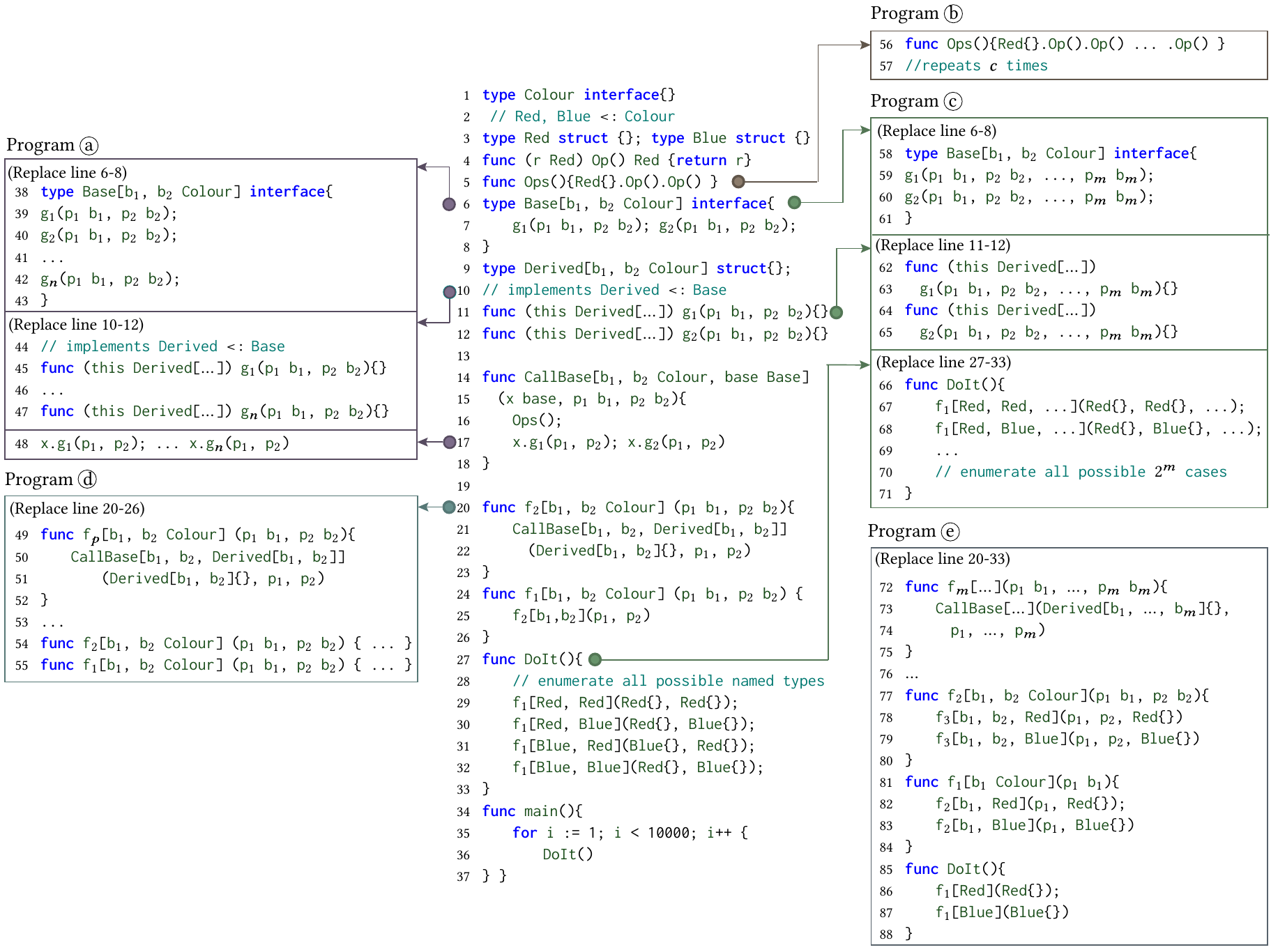}
\vspace*{-.5cm}
\caption{
    The base program and its five variations in the micro benchmarks.}
\label{fig:benchmark-prog-pdf}
\vspace*{-.5cm}
\end{figure}

%% file: sections/experiment-2.tex
\subsectSpace
\subsection{Evaluation Results}
\label{subsec:benchmark}
\subsubsectSpace
\subsubsection{Micro benchmarks}
\input{figs/eval/fig-prog-a-result.tex}

\myparagraph{Program \mycircledtext{a}}
We change $n$ from $2$ to $40$ to 
analyse how the method number of a generic interface impacts
the five translators. 
As shown in Figure~\ref{fig:ssa-n}, the code size (number of assembly instructions) 
of translated \gls{fg} programs has a linear relationship with $n$ for all five translators. 
However, different translators have different coefficients. 
The coefficients of \glsentryshort{mono} ($328.8$), \glsentryshort{gotogo} ($300.8$), and \gomacro ($297.8$) are much larger 
than the coefficients of \glsentryshort{dict} ($117.9$) and \glsentryshort{erasure} ($103.8$).

Figure~\ref{fig:time-n} shows
the execution time of translated programs. 
The programs translated by 
\glsentryshort{dict} and \gomacro 
have a similar performance. 
They are slower 
than the corresponding programs translated by \glsentryshort{mono} and \glsentryshort{gotogo}. 
This is largely due to the usage of dictionaries. 
The programs generated by  \glsentryshort{erasure} have the worst performance,
since the structural typing conducted by \glsentryshort{erasure} when it translates generic method calls to 
polymorphic method calls is very slow~\cite{go-interface-slow-1,go-interface-slow-2}.

Figure~\ref{fig:compile-time-n} shows the compilation
time. 
\glsentryshort{mono} is significantly slower than the other four translators, 
and its compilation time is even not in a linear relationship with $n$. 
The compilation times of the other four translators are similar to each other.

\myparagraph{Programs \mycircledtext{b} and \mycircledtext{d}}
How the number of non-generic operations and the length of the call chain impact the three metrics is quite similar to the method number of 
generic interface \gocode{Base} in \mycircledtext{a}.
In particular, the code size, execution time, 
and compilation time are all in a linear relationship with
the two configuration parameters, except for the compilation time of \glsentryshort{mono}. 
Comparing \mycircledtext{b} with \mycircledtext{a}, 
one important difference to note is that 
for \mycircledtext{b}, 
the programs translated by \glsentryshort{dict} spend a similar execution time 
to that of the corresponding programs translated by \glsentryshort{erasure}, and the execution time
is larger than the execution time of the programs translated by \gomacro. 
However, in Figure~\ref{fig:time-n} for \mycircledtext{a}, 
the line of \glsentryshort{dict} is almost identical to the line 
of \gomacro, indicating that their execution times are similar, 
and the line of \glsentryshort{dict} is lower 
than the line of \glsentryshort{erasure}.
The reason is that when \glsentryshort{dict} translates \glsentryshort{fgg} to \glsentryshort{fg}, 
it also synthesises type assertions for the non-generic 
operations in \glsentryshort{fgg} (line 56 in Figure~\ref{fig:benchmark-prog-pdf}).
The type assertions slow down the translated \glsentryshort{fg}  programs.

\myparagraph{Program \mycircledtext{c}}
The code size, execution time, and compilation time all scale 
exponentially with $m$ for the five translators.
The underlying reason is that function \gocode{DoIt()} 
calls \gocode{f$_1$()} $2^m$ times
in each input \glsentryshort{fgg} program. 
After normalising the three metrics
with the number of characters in the \glsentryshort{fgg} programs, 
we find that the three metrics are in a linear relationship with $m$. 
Among the five translators,  \glsentryshort{erasure}'s translated programs 
have the longest execution time. \glsentryshort{dict} and \glsentryshort{erasure}
spend a similar compilation time, which is much shorter than \glsentryshort{mono}, \glsentryshort{gotogo}, and \gomacro.
\glsentryshort{dict}'s translated programs are similar in size to
 \glsentryshort{erasure}'s translated programs, but they are smaller 
compared with the programs translated by \glsentryshort{mono}, \glsentryshort{gotogo}, and \gomacro.

\input{figs/eval/fig-prog-b-and-real-result.tex}

\myparagraph{Program \mycircledtext{e}}
As shown in Figures~\ref{fig:ssa-m}~and~\ref{fig:compile-m}, 
both the code size of the translated programs 
and the compilation time scale exponentially with $m$
for \glsentryshort{mono}, 
\glsentryshort{gotogo}, and \gomacro.
The reason is that \gocode{f$_m$()} essentially calls \gocode{CallBase()} $2^m$ times with $2^m$
distinct parameter combinations, because for $i\in[2,m), $ \gocode{f$_i$()} calls \gocode{f$_{i+1}$()}
twice, with its input parameters plus \gocode{Red} for the first time and its parameters
plus \gocode{Blue} for the second time, leading the three translators to copy  
\gocode{CallBase()} $2^m$ times. 
However, neither \glsentryshort{dict} nor  \glsentryshort{erasure}
makes any copy of \gocode{CallBase()}, 
and the code size of their translated programs is in a polynomial relationship with $m$
(e.g., for \glsentryshort{dict}'s translated programs, $\textit{size} = 12.8m^2 + 34.5m + 381$, $p<0.001$).

Contrary to the intuition, as shown in Figure~\ref{fig:time-m}, 
the programs translated by \glsentryshort{mono} 
have a worse execution performance compared with the corresponding programs translated by \glsentryshort{dict}, 
when $m$ is larger than 7. 
The reason is that when $m$ is large, a program synthesised by \glsentryshort{mono}
has a large code size, and thus many cache misses occur during its execution. 
For example, when $m$ is 9, the size of the executable file translated by \glsentryshort{mono} is 6.3MB,
and the executable triggers 6,058,156 cache misses in one run, 
while the program translated by \glsentryshort{dict} only
causes 93,695 cache misses. 

\myparagraph{Type simulation} 
As we discussed earlier, we disable the metadata copy of type simulation. 
If we enable the copy, then the translated programs
become slower (e.g., 10\% slower for \mycircledtext{a} when configuring $n$ equal to $2$). The slowdown becomes negligible when $n$ is equal to $40$.%

%% file: figs/eval/fig-prog-a-result.tex
\begin{figure}[t]
    \input{figs/eval/legend.tex}
    \centering
    \begin{subfigure}{0.32\linewidth}
        \includegraphics[width=\linewidth]{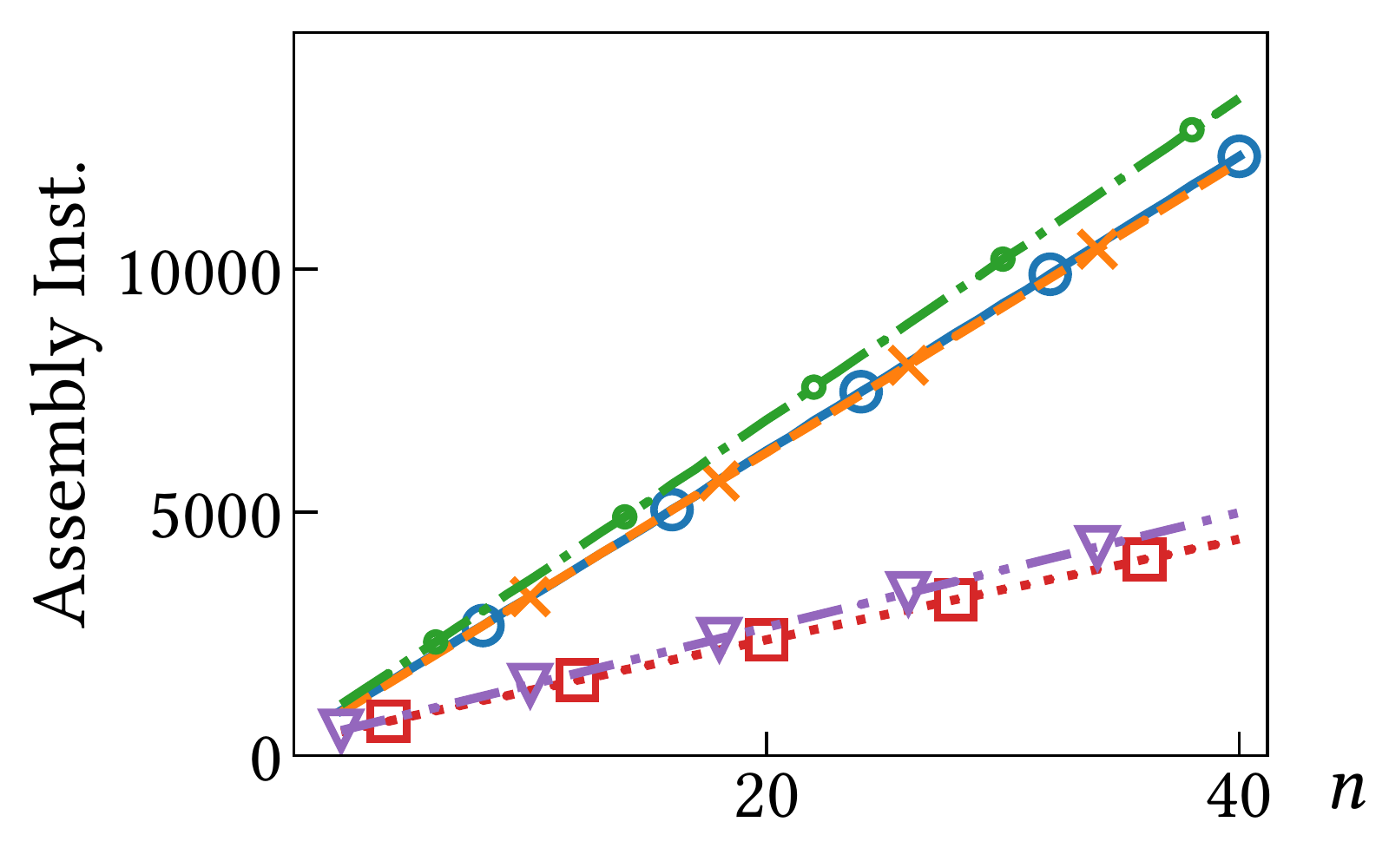}
        \vspace*{-.8cm}
        \caption{Lines of code in assembly}\label{fig:ssa-n}
    \end{subfigure}
    \begin{subfigure}{0.32\linewidth}
        \centering
        \includegraphics[width=\linewidth]{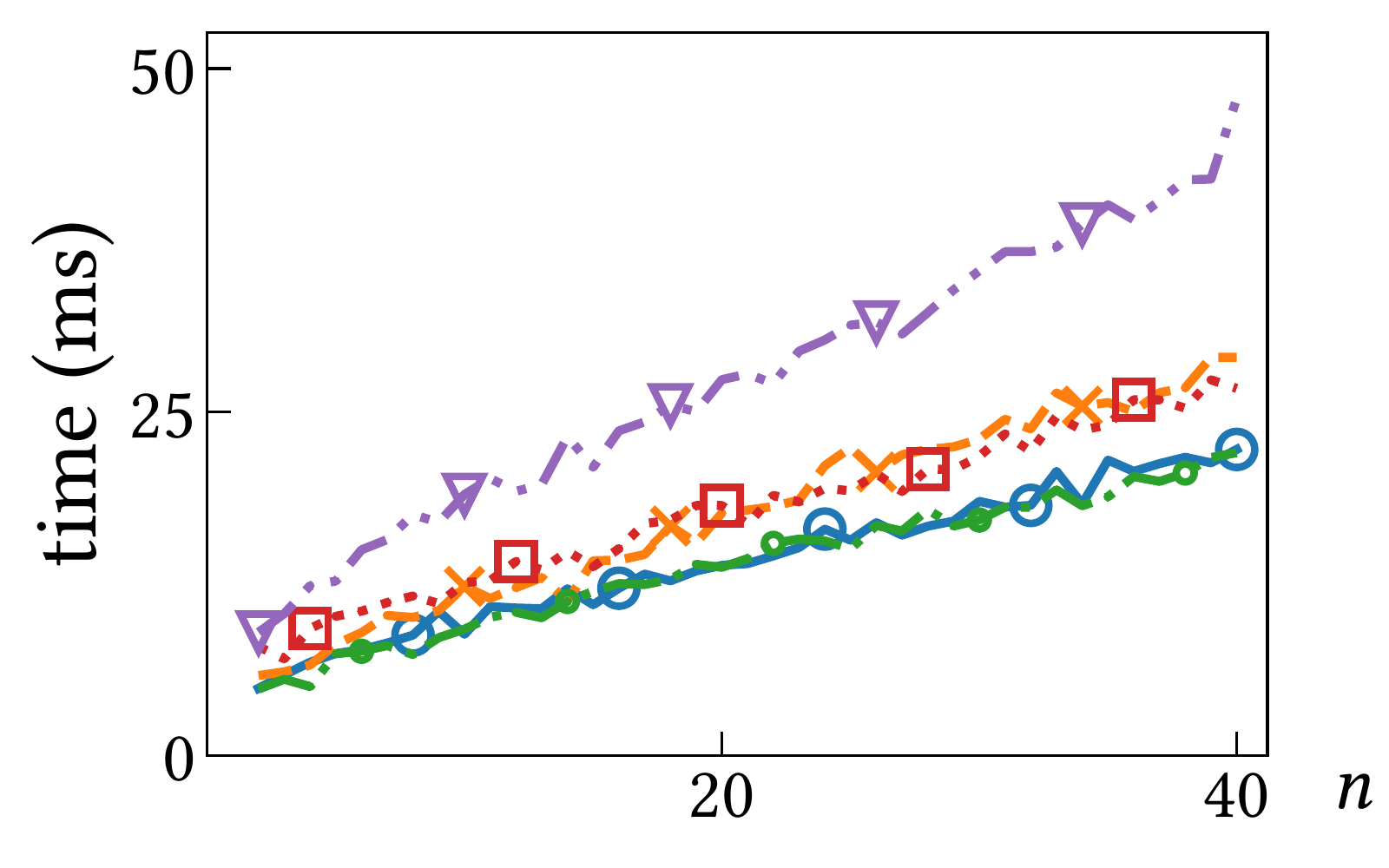}
        \vspace*{-.8cm}
        \caption{Execution time}\label{fig:time-n}
        \Description{}
    \end{subfigure}
    \begin{subfigure}{0.32\linewidth}
        \centering
        \includegraphics[width=\linewidth]{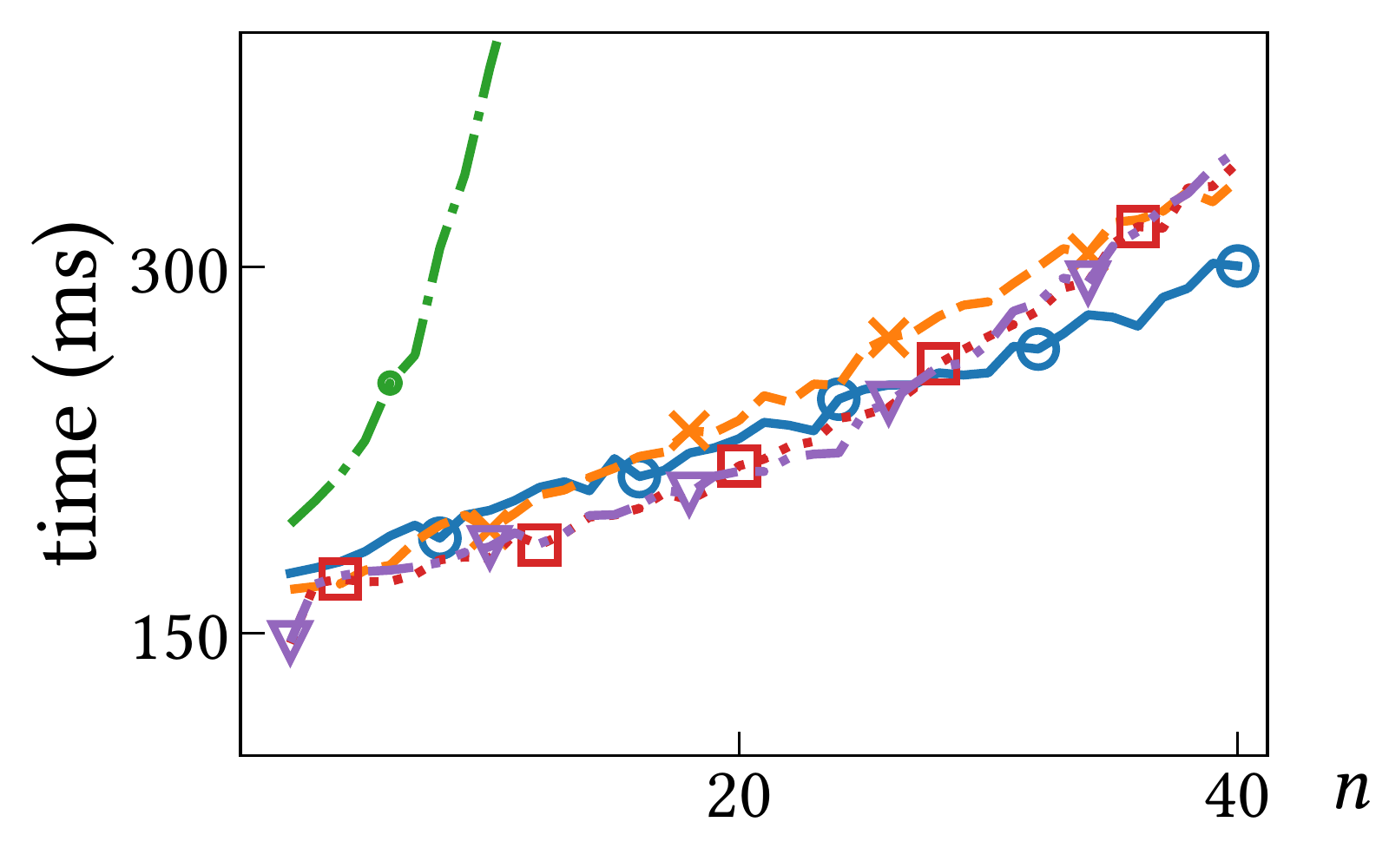}
        \vspace*{-.8cm}
        \caption{Compilation time}\label{fig:compile-time-n}
        \Description{}
    \end{subfigure}
\vspace*{-.3cm}      
\caption{Evaluation results of \benchname{Program}~\mycircledtext{a}}
\vspace*{-0.3cm}      

\label{fig:prog-a-results}
\end{figure}

%% file: figs/eval/legend.tex
{\small\includegraphics[width=2em]{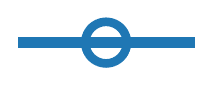}\glsentryshort{gotogo}\hspace{1.5em}
    \includegraphics[width=2em]{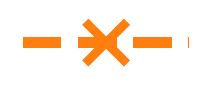}\gomacro\hspace{1.5em}
    \includegraphics[width=2em]{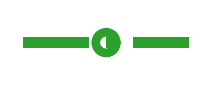}\glsentryshort{mono}\hspace{1.5em}
    \includegraphics[width=2em]{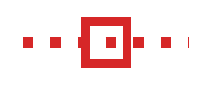}\glsentryshort{dict}\hspace{1.5em}
    \includegraphics[width=2em]{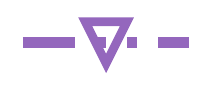}\glsentryshort{erasure}\\
    }

%% file: figs/eval/fig-prog-b-and-real-result.tex
\begin{figure}[t]
    \centering
    \input{figs/eval/legend.tex}
    \begin{subfigure}{0.32\linewidth}
        \includegraphics[width=\linewidth]{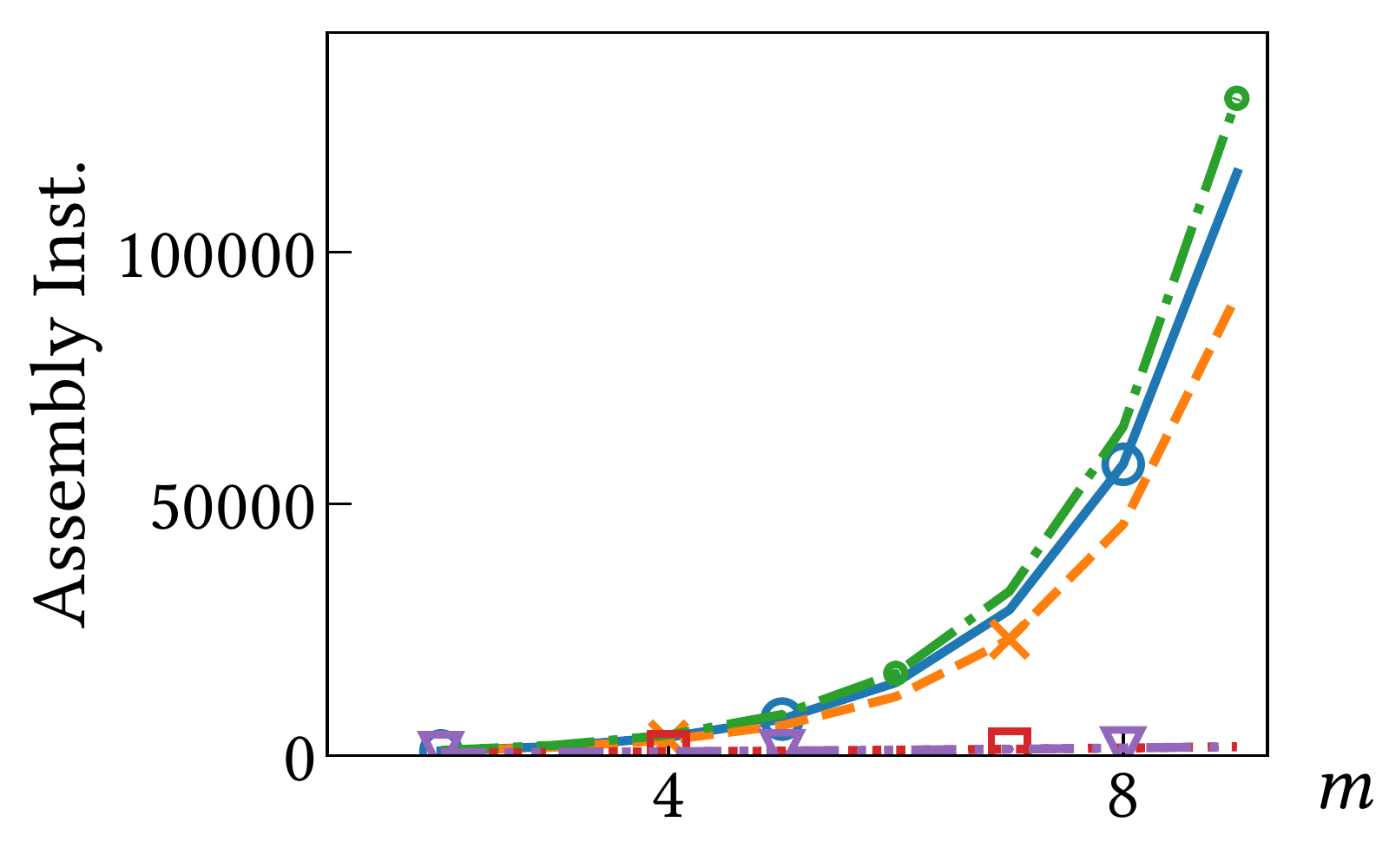}
        \vspace*{-.7cm}
        \caption{Lines of code in assembly}\label{fig:ssa-m}
    \end{subfigure}
    \begin{subfigure}{0.32\linewidth}
        \centering
        \includegraphics[width=\linewidth]{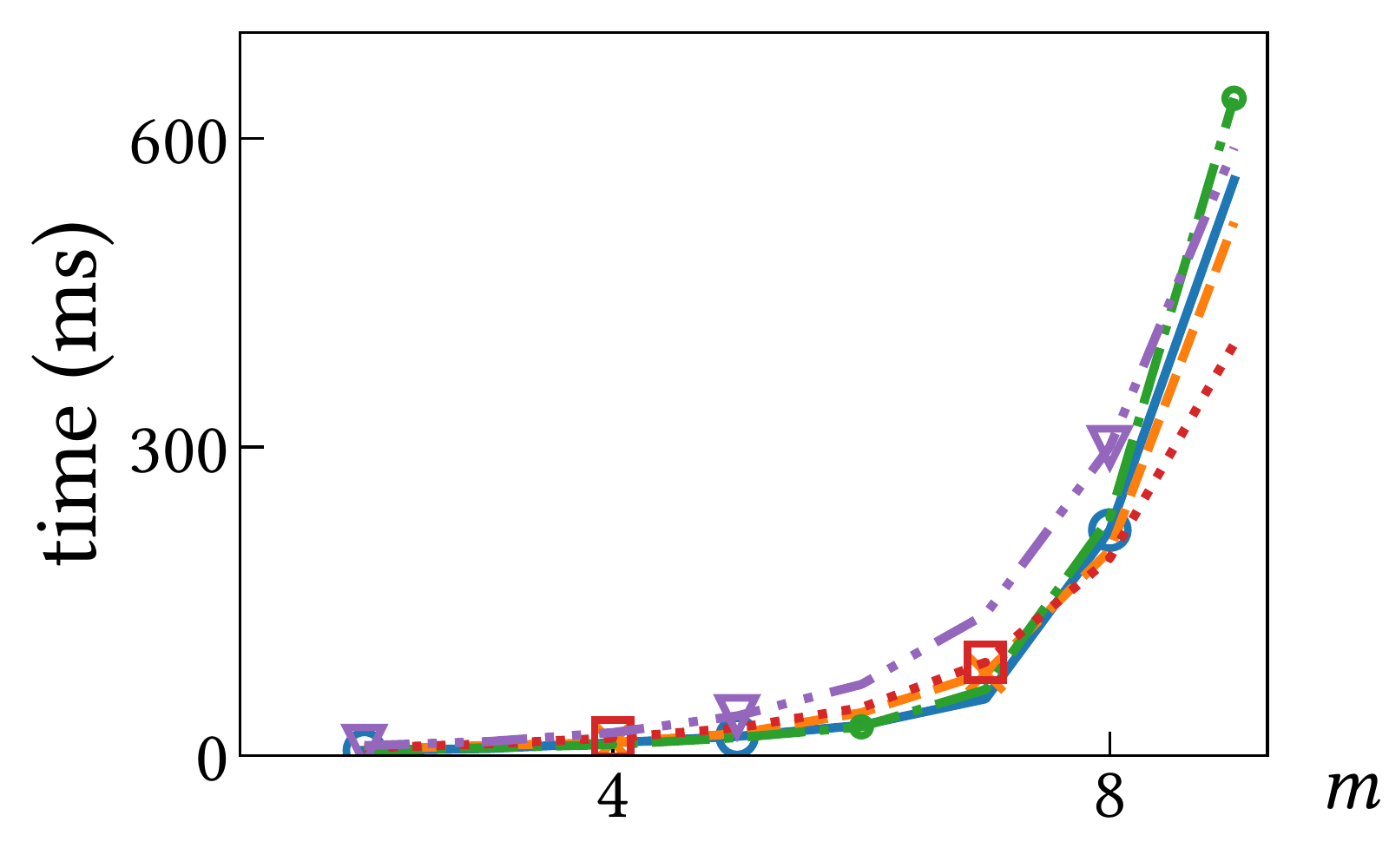}
        \vspace*{-.7cm}
        \caption{Execution time}\label{fig:time-m}
        \Description{}
    \end{subfigure}
    \begin{subfigure}{0.32\linewidth}
        \centering
        \includegraphics[width=\linewidth]{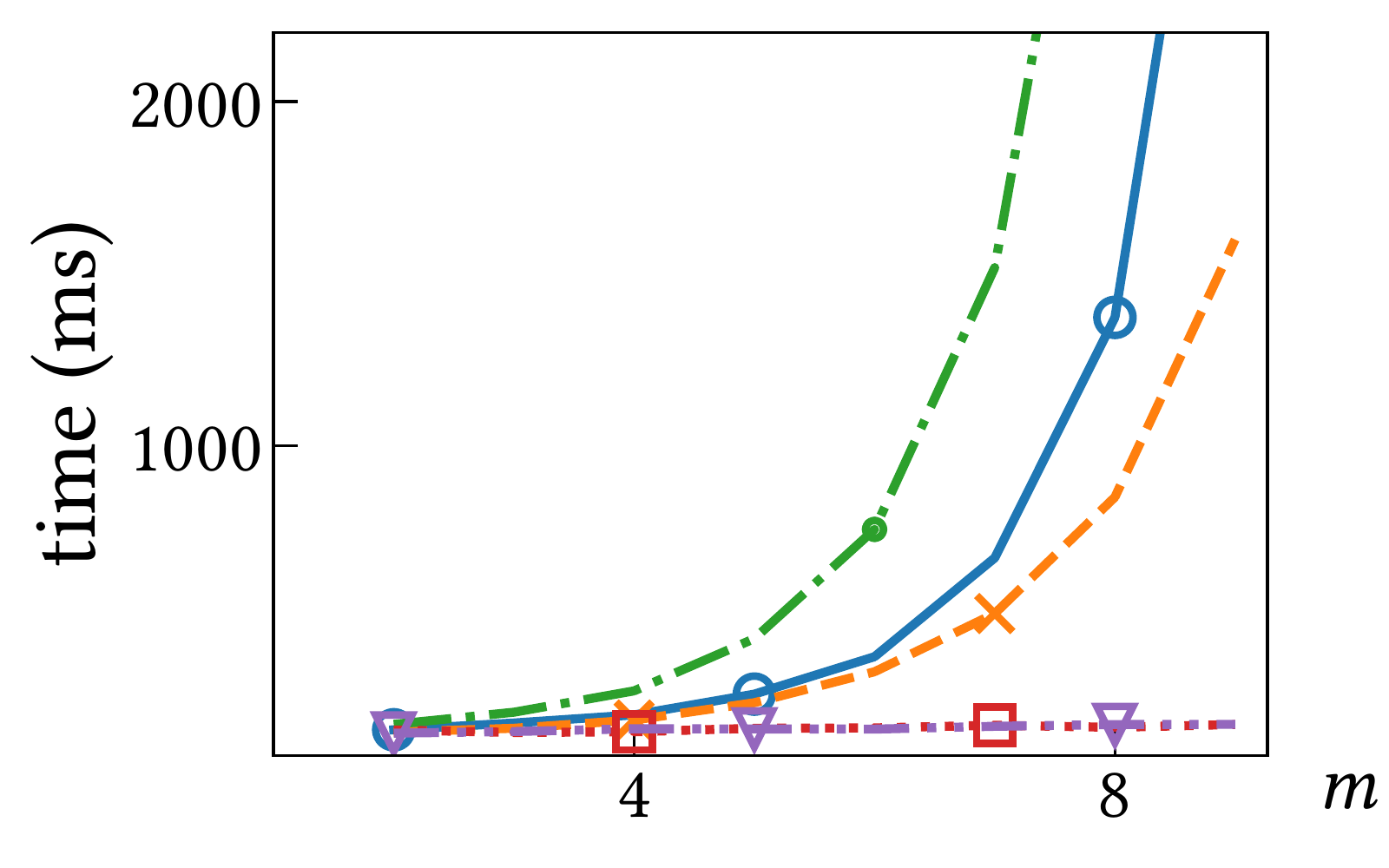}
        \vspace*{-.7cm}
        \caption{Compilation time}\label{fig:compile-m}
        \Description{}
    \end{subfigure}
    \vspace*{-.3cm}      
\caption{Evaluation results of \benchname{Program}~\mycircledtext{e}}
\vspace*{-0.5cm}      

\label{fig:prog-b-results}
\end{figure}

%% file: sections/experiment-3.tex
\subsubsectSpace
\subsubsection{Real-world benchmarks}

The evaluation results of real-world benchmarks are shown in Table~\ref{tab:real-benchmark-results}.
Overall, the translated programs of \glsentryshort{dict} and \glsentryshort{erasure} 
have a smaller code size, but a longer execution time, compared with the corresponding programs translated by 
\glsentryshort{gotogo}, \glsentryshort{mono}, and \gomacro, which is consistent 
with the results on the micro benchmarks. 
However, the compilation time does not change significantly across different translators, 
because all real-world benchmarks are small and do not have many usages of generics.

\input{figs/eval/exp-result-table.tex}

\subsectSpace
\subsection{Discussion and Limitations}
\label{section:discussion}

Our experimental results largely reflect the common intuition that 
monomorphisation translators (\gls{mono}, \gls{gotogo}, and \gomacro) generate programs 
with a better runtime performance, 
while non-specialising translators (\gls{dict} and \gls{erasure}) synthesise programs in a smaller code size.
However, our evaluation also pinpoints cases where monomorphisation generates 
programs in an extremely large size. 
The programs trigger excessive cache misses during execution and have a very bad runtime performance.  
On the other hand, 
our experimental results motivate the introduction and usage of Go generics, 
since without generics, 
Go programmers have to implement polymorphism
using interfaces, which is exactly the same as the programs translated by \gls{erasure},
and our experimental results show that those programs are slow.

In practice, our dictionary-passing translator (\gls{dict}) constantly 
generates programs in a smaller size and takes a smaller (or comparable) compilation time 
than all existing translators (including \gomacro, 
the official generic type implementation).
Thus, it provides an alternative for real-world users of Go generics to 
strike their desired tradeoff. 
Moreover, our implementation and evaluation experience show 
that type simulation is an important component of \gls{dict},
and that type metadata incurs extra runtime overhead. 
Thus, corresponding data structures and algorithms need to be carefully designed 
for better translated programs.
 For instance, link-time optimisation can be applied to remove unused type metadata.

\myparagraph{Possible improvements for \gomacro}
First, \gomacro is very conservative in its support for GC shapes -- 
only considering pointers to have the same GC shape. 
In our experiments, we do not observe the reuse of method implementations, 
or synthesis and use of dictionaries.
Thus, to make full use of dictionaries and GC shape stenciling~\cite{go118}, 
it is necessary for the Go team to improve the current implementation and support
more GC shapes. 
Second, the Go team can 
consider dictionary-passing-based homogeneous compilation, as proposed
in this paper, since it supports polymorphic recursion, provides a faster compilation speed,
generates programs with a smaller code size, and enables separate compilation.

\myparagraph{Limitations}
Since the official generic type implementation released on
March  15th, 2022, 
there does not yet exist generic Go code from 
large, production-run Go software (e.g.,~Docker, Kubernetes, etcd).
We build the two benchmark suites to explore the translators'
asymptotic behaviours 
and inspect how they perform on representative generic programs in other languages,
which is our best effort in conducting the evaluation. 

We formalise \gls{dict} as a source-to-source translator 
to clarify design choices for future implementations and 
aid our proof of correctness (Theorem~\ref{thm:main:correctness}). 
However, this choice limits the performance of our implementation, and the evaluation results
may not reflect the true capability of dictionary-passing 
translation for two reasons: 
first, we erase all types to \gocode{Any} to ensure type preservation, 
which is slow at runtime; 
and second, Go does not allow the creation of global constant dictionaries in source code, 
but those dictionaries can potentially be created by the Go compiler
and leveraged by translated programs for a better runtime performance.

%% file: figs/eval/exp-result-table.tex
\begin{table}[t]\centering

        \resizebox{\columnwidth}{!}{
            \begin{tabular}{@{\hspace{0pt}}l@{\hspace{3pt}}|@{\hspace{3pt}}c@{\hspace{2pt}}c@{\hspace{2pt}}c@{\hspace{2pt}}c@{\hspace{2pt}}c|c@{\hspace{2pt}}c@{\hspace{2pt}}c@{\hspace{3pt}}c@{\hspace{2pt}}c|c@{\hspace{2pt}}c@{\hspace{2pt}}c@{\hspace{3pt}}c@{\hspace{2pt}}c@{\hspace{2pt}}}	
                \toprule
                & \multicolumn{5}{@{\hspace{0pt}}c|@{\hspace{0pt}}}{Instruction Count}           & \multicolumn{5}{c|}{Execution Time (s)}  &  \multicolumn{5}{c}{Compilation Time (s)}\\

                Name           & \smallgls{gotogo}              &{\small\gomacro}                & \smallgls{mono} & \smallgls{dict}                             & \smallgls{erasure} &  \smallgls{gotogo} & {\small\gomacro} & \smallgls{mono} & \smallgls{dict} & \smallgls{erasure}& \smallgls{gotogo} & {\small\gomacro} & \smallgls{mono} & \smallgls{dict} & \smallgls{erasure} \\\midrule
                \benchname{List}           & 1201  &1419  & 1736  & $862$   & 840           &  23.2   &22.4     & 24.4      & $37.69$      & 30.89 & 0.24 & 0.22 & 0.23 & 0.22 & 0.22         \\
                \benchname{ResizableArray} & 1976  &2281  & 1882  & $867$   & 841  &  7.0    &6.8  & 6.9  & $14.90$ & 15.00 & 0.24 & 0.22 & 0.23 & 0.23 & 0.22 \\
                \benchname{ListReverse}    & 1546  &1818  & 1753  & $1115$  & 1204 & 37.0   &35.6  & 36.0 & $42.89$ & 41.76 & 0.26 & 0.25 & 0.26 & 0.24 & 0.25 \\
                \benchname{VectorReverse}  & 985   &1072  & 1047  & $921$   & 914  & 2.99   &2.9   & 3.12 & $2.66$  & 2.68  & 0.25 & 0.24 & 0.26 & 0.24 & 0.25 \\
                \benchname{Cell}           & 112   &104   & 75    & $196$   & 151  &  .006  &.009  & .006 & $0.007$ & .007  & 0.24 & 0.23 & 0.16 & 0.24 & 0.24 \\
                \benchname{Hashtable}      & 188   &184   & 209   & $245$   & 249  &  0.24  &0.21  & 0.46 & $0.46$  & 0.45  & 0.19 & 0.17 & 0.19 & 0.19 & 0.19 \\
                \hline
                Geometric mean & 651.0 & 703.0 & 674.3 & 576.9 & 555.0 & 1.71 & 1.72 & 1.91 & 2.45 & 2.35 & 0.24 & 0.22 & 0.22 & 0.23 & 0.23 \\
                \bottomrule
            \end{tabular}
        }
    \caption{ Results of real-world benchmarks. 
    }\label{tab:real-benchmark-results}
    \vspace*{-1.1cm}
\end{table}

%% file: sections/related.tex
\section{Related Work}
\label{section:related}

\textbf{\emph{Implementation and benchmarks of generics}.}

\begin{wrapfigure}{r}{0.70\linewidth}
        \footnotesize
        \begin{tabular}{@{\hspace{0pt}}c@{\hspace{3pt}}|@{\hspace{3pt}}l@{\hspace{2pt}}l@{\hspace{2pt}}l@{\hspace{2pt}}l@{\hspace{2pt}}l}
    &Language &Translation(s) \ \ & Optimal & Optimal \\
    & & & Exec.~Time \ & Code Size \\\midrule
             Our work
             & \gls{fgg} (Go) 
             & Dict/Mono/ 
             & Mono (1st)
             & Erasure$^\dagger$ (1st)\\    
              &  
             & Erasure$^\dagger$ 
             & Dict (2nd)
             & Dict (2nd)\\
    \midrule
    Go team & Go & Mono/Hybrid & Mono & Mono \\
    \midrule
    
            \cite{ureche2013miniboxing}
                    & Scala (JVM)\ 
                    & Hybrid
                    & Hybrid
                    & Hybrid
                    \\
            \cite{odersky2000two}
                    & Pizza (Java)\ 
                    & Mono/Erasure
                    & Mono
                    & Erasure
                    \\
            \cite{kennedy2001design} 
                    & .NET CLR 
                    & Hybrid 
                    & Hybrid 
                    & N/A
                    \\ 
            \cite{Jones93} 
                    & Haskell 
                    & Dict/Mono 
                    & Mono 
                    & Mono
                    \\
        \end{tabular}
    \\[1mm]
    ($\dagger$) \gls{fgg} Erasure is not type preserving.
    \vspace*{-2mm}
        \caption{Implementations and benchmarks}
        \label{fig:tab-benchmark-works}
    \vspace*{-3mm}
    \end{wrapfigure}

To the best of our knowledge,
there is no existing work comparing implementations 
of generics in Go.
The closest ones target JVM languages \cite{odersky2000two,ureche2013miniboxing},
.NET common language runtime (CLR) \cite{kennedy2001design},
and Haskell \cite{jones1995dictionary}.
\citet{odersky2000two} 
benchmark a homogeneous
(similar to \acrshort{erasure}) and a
heterogeneous (similar to \acrshort{mono}) translation
for Pizza (an extension of Java with generics). 
They find that heterogeneity
reduces execution time, but also increases code size. 
\citet{jones1995dictionary} gave a similar comparison for Haskell, reporting
that monomorphisation produces a smaller code size;  
our work shows the opposite result.
One major reason is that unnecessary dictionary 
fields and manipulation of dictionary parameters require 
more assembly instructions in Haskell than Go, as 
Go targets 
low-level efficiency.

\citet{kennedy2001design} apply a hybrid 
dictionary and monomorphisation approach targeting the 
Just-In-Time (JIT) .NET CLR compiler. 
Object instantiation 
is conducted lazily at runtime according to an object's code 
and structure (\eg memory layout and garbage 
collection shape). Each object contains 
a pointer to a dictionary (vtable), which provides 
method entry points and type information.
With the help of lazy instantiation during runtime, 
.NET CLR supports abundant language features, 
including but not limited to $F$-bounded polymorphism 
and polymorphic recursion.
They compare their design with equivalent non-generic 
implementations using \texttt{Object}s and 
hand-specialised code. Their execution speed is 
close to that of the hand-specialised versions.
The \gomacro approach is similar to .NET CLR, 
but unlike .NET CLR, its instantiation happens at 
compile time. Due to structural typing dictionaries are instantiated 
through an approach similar to instance discovery 
in monomorphisation. 
Hence, \gomacro suffers from an inability to support polymorphic 
recursion (\ie constrained by \textit{nomono}, \S~\ref{section:nomono}) and 
the large code size of monomorphisation (\S~\ref{sec:exp}). 

\citet{ureche2013miniboxing} propose an optimised
monomorphisation approach called miniboxing using 
one monomorphised instance on
types with different sizes to reduce code size.
Methods of different types are specialised at runtime 
using a custom classloader.
They benchmark seven different settings, one achieving
at most a 22 times speedup over the default 
generics translation in Scala.
The main design goal of their benchmarks is the 
performance of reading and writing miniboxed objects 
allocated on heap by the JVM.
They test the different combinations of concrete 
types for generics (``Multi Context''), which is 
similar to the scenario of \benchname{Program}~\mycircledtext{c} (in \S~\ref{subsec:evaluation}), 
but their goal is to test the historical 
paths executed in the HotSpot JVM.
They also test the speed of one method call 
\texttt{hashCode} from generics types. 
In comparison, our benchmarks test how various 
factors impact the performance (\eg the method 
number in an interface).

\myparagraph{Formal translations of generics}
Formal translations of generics
can be split into three main techniques;
\emph{Erasure},
\emph{dictionary-passing}, and 
\emph{monomorphisation}. 
We consider the most relevant work, 
a breakdown of which
is provided in Figure~\ref{table:trans:theory}.
Where these works formally prove the correctness of their translation,
we observe that they can be grouped as 
\emph{behavioural equivalence}~\cite{griesemer2020featherweight, Igarashi99FJ}
and \emph{value preservation}~\cite{yu2004formalization}.
The former demands that during evaluation the source and target
programs are still related, whereas the latter merely requires that
the result of a productive program be preserved.
In general behavioural equivalence is a more fine-grained equivalence, as
it can be used to show value preservation. 
In this paper, we formalised and then 
proved our dictionary-passing translation 
correct using bisimulation up to dictionary-resolution, which is 
categorised as  
a behavioural equivalence.

\citet{yu2004formalization}
formalise a hybrid dictionary and
monomorphisation translation for the .NET~CLR.
\begin{wrapfigure}{r}{0.66\linewidth}
\footnotesize
\vspace{-3mm}
 \begin{tabular}{@{\hspace{0pt}}c@{\hspace{3pt}}|@{\hspace{3pt}}lllc}
        & Language   & \hspace{-2mm}Approach & \hspace{-2mm}Translation(s)\ &
        \hspace{-5mm}Formalised   \\ \midrule
        Our work
                & \gls{fgg} (Go)
                & S-to-S
                & Dict
                & \tick                     \\
\midrule
        \cite{griesemer2020featherweight}
                & \gls{fgg} (Go)
                & S-to-S
                & Mono
                & \tick                     \\
        \cite{Igarashi99FJ}
                & Java
                & S-to-S
                & Erasure
                & \tick                     \\
        \cite{yu2004formalization} 
                & .NET CLR 
                & IR-to-IR
                & Hybrid 
                & \tick                     \\ 
        \cite{bottu2019coherence} 
                & Haskell 
                & S-to-IR
                & Dict
                & \tick                     \\
        \cite{OW97}
                & Pizza 
                & S-to-S
                & Mono/Erasure
                & \cross
\end{tabular}
\\
\begin{center}
        S-to-S$=$Source to Source; IR$=$Intermediate representation 
\end{center}
\vspace{-3mm}
\caption{Related Work: Theory}
\label{table:trans:theory}
\vspace{-5mm}
\end{wrapfigure}
They mostly follow the design of \cite{kennedy2001design}.
They consider a target language which can, using an object's type, request the
specific dictionary from an assumed infinite map.
This is justified for the .NET~CLR as method dictionaries
are created on-demand using an object's type.
Compare this to our translation in which we must eagerly
construct dictionaries and pass
them in addition to the objects that they describe.
\citet[Theorem 5]{yu2004formalization} show that their 
translation is value preserving; 
for expression $e$, and value $v$, 
if $e$ evaluates to $v$ ($e \Downarrow v$)
then there is a reduction
such that $\mapother{e} \red^* \mapother{v}$ 
(where $\mapother{-}$ is their translation).

\citet{bottu2019coherence} formalise
dictionary-passing in Haskell. 
Their work focuses on proving a \emph{coherency theorem}.  
They motivate this work as nominally typed languages 
featuring multiple inheritance (\ie Haskell) 
suffer from an ambiguity in dictionary-resolution such that 
the translation of 
a single source program may
\emph{non-deterministically}
produce different terms in the target language.
A translation is coherent when these target terms 
are contextually equivalent. 
We need not consider this issue, as Go's structural typing 
system does not support the multiplicity of superclass implementations 
that causes incoherence. 
\citet{bottu2019coherence} do not prove the correctness 
of their dictionary-passing translation using an equivalence 
between the source and target language.

\citet{griesemer2020featherweight} formalised the \gls{fg} and \gls{fgg} 
languages, as well as the \gls{mono} translation used in \S~\ref{sec:exp}. 
This work defines a class of \gls{fgg} 
programs that can be monomorphised, 
and proves that class membership is decidable. 
Finally, they prove that their translation forms 
a one-to-one bisimulation. 
Their behavioural equivalence is 
straightforward and does not require any up to 
techniques, as monomorphisation does not 
introduce runtime computations.

\citet{OW97} describe, but do not formalise, two 
alternative approaches -- erasure and monomorphisation -- 
to implementing 
generics in the Pizza language, a generic variant of Java. 
\citet{Igarashi99FJ} build on the erasure technique 
developed in \cite{OW97}. 
Their work formalises 
Featherweight Generic Java and  
proves a formal erasure translation 
to Featherweight Java. 
They prove the correctness of their erasure 
translation using a behavioural equivalence, 
although their translation introduces
\emph{synthetic casts} (assertions), which complicates 
the correctness theorems. 
To resolve this issue, they introduce a reduction 
for their proofs which freely adds, 
removes, or safely alters any required synthetic casts.
Correctness of their translation 
is split 
into two directions, called 
\emph{weak completeness} and 
\emph{soundness} \cite[Theorem~4.5.4 and Theorem~4.5.5]{Igarashi99FJ}, 
which uses a behavioural equivalence up to 
the cast reduction. 
As with our paper, they use these theorems to show a 
value preservation corollary. 
\citet[Corollary~4.5.6]{Igarashi99FJ} also prove  
that their erasure translation is type assertion error preserving
-- in contrast to our \gls{erasure} translation, since ours does 
not preserve type assertions. This disparity is due to 
a limitation on the expressivity of assertion  in Generic Java. 
The inclusion of this limitation has been an area of contention, with other 
authors suggesting that it could be overcome with the 
use of type-reps~\cite{allen02thecase,agesen1997adding,solorzano98reflection,Viroli00reflectiveGJ,crary1998intensional}.

\myparagraph{Formal non-generics dictionary translation}
\citet{sulzmann2021dictionary} propose a
dictionary-passing translation from the non-generic \gls{fg} to an
\emph{untyped} variant of the $\lambda$-calculus
with pattern matching.
They use a dictionary-passing approach to investigate
Go's resolution mechanism for overloaded methods and
structural subtyping. 
\citet{sulzmann2021dictionary}
prove that their translation is value preserving 
using a step-indexed logical relation.

Intuitively, \citet{sulzmann2021dictionary} use an inductive 
proof technique that, using two related values $v$ and $v'$ at type $t$, 
relates any terms ($e$ and $\mapother{e}$) 
that can reduce to 
$v$ and $v'$ (\emph{resp.}) within $k$~reduction-steps. 
Step-indexed logical relations are a sophisticated extension to 
logical relations (\eg \cite{bottu2019coherence}),
and are applicable for languages with recursion. 
\citet{sulzmann2021dictionary} left 
a type-preserving translation from \gls{fg} and
a translation from \gls{fgg}
as their future work. 
No implementation or evaluation of their translation is provided.

\myparagraph{Alternatives to bisimulation up to} 
In our motivating example for \emph{up to dictionary resolution} 
(Figure~\ref{fig:example:nontriv}), 
we briefly discuss potential alternate many-to-many bisimulation approaches.
One such approach is the 
\emph{stuttering bisimulation}~\cite{browne1988charachterizing},
which has been studied extensively in the domain of model checking~\cite{baier2008principles}. 
The stutter bisimulation relates two terms when they 
both reduce to related terms in an unbounded, but finite, number of steps. 
Formally,  $e$ and $\mapother{e}$ are related by a 
\emph{stutter bisimulation} 
when \begin{enumerate*}
        \item $e\red e'$ implies that there exists a finite reduction 
                $\mapother{e}\red d_0 \red \cdots \red d_n$ ($n\ge 0$) 
                where each intermediate state $d_i$ is related to $e'$; and symmetrically,
        \item $\mapother{e}\red d$ implies that there is a finite reduction from $e$ 
                with each element being related to $d$. 
\end{enumerate*}
This approach works well for finite models, but becomes \emph{undecidable}
when applied to Turing complete languages such as \gls{fgg}. 
To overcome this issue, the works in \cite{hur2014logical,leroy2009formally}
consider restricted, decidable, variants of 
the stutter bisimulation to show the correctness of their translations. 
\citet{leroy2009formally} formulates the non-symmetric 
\emph{``star''-simulation}, which requires a well-founded ordering on reducing terms 
to ensure that either \begin{enumerate*}
        \item both source and target terms reduce infinitely; or 
        \item the source cannot reduce infinitely while the target is stuck. 
\end{enumerate*}
In practice, the well-founded ordering used in (2) 
is approximated using fixed parametric bounds. 
\citet{hur2014logical} formulate this idea 
using \emph{stuttering parametric bisimulation}, 
which bounds the number of steps that two related 
terms can take before their reductions are related. 
Such restricted variants of the stutter bisimulation 
cannot provide a sound and complete correctness proof for \gls{dict}.
More generally, our use of a fine-grained up to bisimulation
not only develops on existing correctness theorems 
for the translation generics
\cite{Igarashi99FJ,griesemer2020featherweight}, but it
can also be readily extended to include advanced language features 
such as concurrency and side effects in Go.

%% file: sections/conclusion.tex
\section{Conclusion}
\label{section:conclusion}
In this paper, we design and formalise a new source-to-source, 
non-specalised call-site 
dictionary-passing translation of Go, and prove
essential correctness properties 
introducing a novel and general \emph{bisimulation up to} technique. 
The theory guides a correct implementation of 
the translation, 
which we empirically compare along with the recently released \gomacro, 
an erasure translator, and two existing monomorphisation 
translators~\cite{gotogo,griesemer2020featherweight},  
with micro and real-world benchmarks. 
We demonstrate that our dictionary-passing translator handles 
an important class of Go programs (\textit{F}-bounded polymorphism and \nomono{}
programs) beyond the capability of \gomacro 
and existing translations \cite{gotogo,griesemer2020featherweight},  
and provide several crucial findings and implications 
for future compiler developers to refer to. 
For instance,  \gomacro requires more improvements on GC shapes in order to
effectively generate small binary code
(See \ref{section:discussion}
for a more detailed discussion).

Beyond Go language, 
many dynamically typed languages (such as Python, JavaScript, and
Erlang) type-check at runtime, and their engines cannot 
easily decide an object's 
implemented methods nominally, similarly to Go. 
Consequently,  
many
of their implementations~\cite{salib2004starkiller,castanos2012benefits, gal2009trace}
apply similar approaches to monomorphisation to optimise
execution speed. 
Rust also supports generic via monomorphisation,
yet this is considered a major reason for slow compilation.
Our work can help in choosing alternative
optimisations for these languages to reduce  
code size and compilation time. 

In the future, we plan to inspect how other important Go language features 
(e.g., \emph{reflection}, 
\emph{packages}, \emph{first-class}, \emph{anonymous} functions) interact with generics
by proving the correctness and examining the trade-offs among runtime performance, code sizes, and compilation times.

%% file: ack.tex
\begin{acks} 
The authors wish to thank Ziheng Liu 
and Zi Yang for benchmark collection and initial paper discussion, 
and the anonymous reviewers for their invaluable comments and suggestions on the paper.
This work was partially supported by 
EPSRC (EP/T006544/1, EP/K011715/1, EP/K034413/1, EP/L00058X/1,
EP/N027833/1, EP/N028201/1, EP/T006544/1, EP/T014709/1, EP/V000462/1
and EP/X015955/1), NCSS/EPSRC VeTSS, 
a Mozilla Research Award, and an Ethereum Grant. 

\end{acks}

%% file: appendix.tex
\appendix
\onecolumn
\input{sections/appendix-pl-table.tex}
\input{sections/appendix-fg}

\input{sections/appendix-fgg}

\setlength{\proofrightwidth}{0.3\linewidth}
\input{proof/dict/correctness}
\input{sections/appendix-dyn-ex.tex}

%% file: sections/appendix-pl-table.tex
\section{Appendix: Generic Implementations of Top 16 Statically Typed Generic Programming Languages}
\label{app:implementations}

\begin{table}[!htp]\centering
\scriptsize
\begin{tabular}{lllll}\toprule
Programming Language &Mainstream Implementation &Memory Management &Runtime Environment \\\cmidrule{1-4}
Java &Erasure &Garbage Collection &JVM \\
Kotlin &Erasure &Garbage Collection &JVM \\
Scala &Erasure &Garbage Collection &JVM \\
C\# &Just-In-Time Specialisation + Non-Specialised Dictionary &Garbage Collection &.NET CLR \\
Visual Basic &Just-In-Time Specialisation + Non-Specialised Dictionary &Garbage Collection &.NET CLR \\
Dart &Erasure &Garbage Collection &Virtual Machine \\
Swift &Non-Specialised Dictionary/Monomorphisation* &Reference Counting &Native \\
Objective-C &Non-Specialised Dictionary/Monomorphisation* &Reference Counting &Native \\
Haskell &Non-Specialised Dictionary &Garbage Collection &Native \\
Go & Monomorphisation + Specialised Dictionary &Garbage Collection &Native \\
D & Monomorphisation &Garbage Collection &Native \\
C++ & Monomorphisation &Manual &Native \\
Rust & Monomorphisation &Manual &Native \\
Delphi & Monomorphisation &Manual &Native \\
Ada & Monomorphisation &Manual &Native \\
Fortran & Monomorphisation &Manual &Native \\
\bottomrule
\end{tabular}
\caption{Generic implementations of top 16 statically typed programming languages with generics. Languages are selected from the top 40 languages by IEEE Spectrum in 2021~\cite{TopProgr17:online}.
(*when source code available or specified by users.)
}\label{tab:app-pl-table }
\end{table}

%% file: sections/appendix-fg.tex
\section{Appendix: \glsentrylong{fg}}
\label{appendix:fg}
\label{app:fg}
For the reviewer's convenience, this section provides 
more explanations of the syntax and the full definitions of the typing
system from those in \cite{griesemer2020featherweight}. 
 
\subsection{\glsentrylong{fg} Syntax}
\label{app:fg:syntax}
We explain 
the syntax of \gls{fg} in 
Figure~\ref{fig:fg:syntax}.
The meta variables for field ($f$), method ($m$), variable
($x$), structure type names ($t_S, u_S$), and interface type names
($t_I, u_I$) range over their respective namespaces. Types ($t, u$)
range over both structures and interfaces.
A program ($\prog$) is given by a sequence of declarations ($\multi{D}$)
along with a {\bf main} function which acts as the top-level expression. 
We often shorten this as $\prog = \program{e}$.

Expressions in \gls{fg} are  
variables ($x$), method calls ($e.m(\multi{e})$), structure literals
($t_S\{\multi{e}\}$), field selection ($e.f$), and type assertion
($e.(t)$). 

Declarations ($D$) can take three forms; 
\emph{structure}, \emph{interface}, or \emph{method declaration}. The
structure declaration ($\struct{\multi{f\ t}}$) gives a sequence of
typed fields whereas the interface declaration ($\interface{\multi{S}}$)
gives the method specifications which instances of that interface
should implement. A method specification ($m(\multi{x~t})~t$)
prescribes the name and type for implementing methods. 

A method declaration ($\func (x\ t_S)\ m(\multi{x\ t})\ t_r\ \{b\}$)
defines a method $m$ on the structure $t_S$. This method accepts the
arguments $\multi{x\ t}$, which along with the receiver $x$ are passed
to the method body $b$. On a successful computation this method will
return a result of type $t_r$. The special $\lit{main}$ function acts
as the entrance function, and thus has no receiver, arguments or
return value. 

\subsection{\glsentrylong{fg} Typing}
\label{app:fg:types}
\input{figs/fg/fg-typing}
For the reviewer's convinience,
we reproduce the typing system with the minimum explainations.  
Figure~\ref{fig:fg:types} gives the \gls{fg} typing rules
and auxiliary functions.
Environment $\Gamma$ is a sequence
of typed variable names ($\multi{x : t}$). 
We assume 
all variables in $\Gamma$ are distinct, and 
write $\Gamma,x : t$ if $x\not\in \dom{\Gamma}$.

\myparagraph{Implements} 
    The \emph{implements relation} ($t <: u$) holds if type $t$ is a subtype of type $u$,
    understood as a relationship in which a variable of type $u$ can be
    substituted by any variable of type $t$. 
    A structure can only only be
    implemented by itself (\rulename{<:s}). An interface $t_I$ can be
    implemented by any type that possesses at least the same methods as
    $t_I$ (\rulename{<:i}).

\myparagraph{Well-formedness} 
A well-formed term is one that is not only syntactically correct, but
one that also has semantic meaning. 
$x \ok$ holds if $x$ is well
formed according to the typing rules, with the extension
$\wellFormed{x}$ if the term is well-formed in the environment
$\Gamma$.
A type declaration is
well-formed when the type it declares is well-formed (\rulename{t-type})
which happens when it is either a structure with distinct and well
formed fields (\rulename{t-struct}), or an interface with unique and
well-formed method specifications (\rulename{t-interface}). Method
specifications are well-formed when all argument types and the return
type are well-formed (\rulename{t-specification}).

\myparagraph{Method body and statement type checking}
The typing judgement $\wellTyped{x}{t}$ holds if the term $x$ has type $t$
in the environment $\Gamma$.  A method
($\textbf{func}~(x~t_S)~m(\multi{x~t})~u~\{b\}$) is well-formed if the
type of the receiver, the return, and all arguments are well-formed,
with all names being distinct from one
another. 

    A structure literal ($\sType{t}\{\multi{e}\}$) is well-typed when each 
    field instantiation ($\multi{e}$) subtypes the field's 
    declared type (\rulename{t-literal}).
Field 
assignment and access follow the order of declaration. 

\myparagraph{Expression type checking}
Given an expression $e$ of type $u$ the type assertion $e.(t)$ casts the expression to 
type $t$. There are three non-overlapping type assertion rules. 
The Go specification only permits type assertions from an interface type, which 
informs rules \rulename{t-assert$_I$}. 
An assertion between two interface types (\rulename{t-assert$_I$}) does not statically 
check the assertion since the expression $e$ could evaluate to a term that 
implements the target type $t$.
Assertion from an interface type~$u$ to a non-interface type~$t$
is allowed only if $t$~implements~$u$ (\rulename{t-assert$_S$}). 

Not part of the Go specification and not need for compile time
checking, 
the rule \rulename{tr-stupid} 
is only used for the type assertion $e.(t)$ where $e$ has evaluated to a concrete non-interface type.  
This assertion provides no utility at compile time as an assertion from a non-interface type is either a no-op
or it unnecessarily erases type information -- yet without this rule a term may become ill-typed during evaluation. 

More detailed explanations can be found in
\cite[\S~3.3]{griesemer2020featherweight}. 

\begin{theorem}[Preservation] 
\label{fcgSubjtCong}
\label{thm:fgTypePreservation}
{\rm (Theorem 3.3 in \cite[\S~3.3]{griesemer2020featherweight})} 
If $\wellTyped[\emptyset]{e}{u}$ and $\reduction{e}{e'}$ 
then $\wellTyped[\emptyset]{e'}{t}$ for some $t <: u$.
\end{theorem}

\begin{theorem}[Progress] 
\label{thm:fgProgress}
{\rm (Theorem 3.4 in \cite[\S~3.3]{griesemer2020featherweight})}
If $\wellTyped[\emptyset]{e}{u}$ then  
$e$ is either a value, $\reduction{e}{e'}$ for some $e'$ or 
$e$ panics. 
\end{theorem}

%% file: figs/fg/fg-typing.tex
\begin{figure}
\small
    Implements and well-formed types
        \hfill
        \fbox{$t <: u \vphantom{\ok}$} \quad \fbox{$t\ok$}
                {\footnotesize
                    \begin{equation*}
                        \begin{gathered}
                           \namedRule{\subS}{
                                \axiom{\sType{t} <: \sType{t}}
                            }
                            \namedRule{\subI}{
                                \infer{
                                    t <: \iType{t}
                                }{
                                    \methods(t) \supseteq \methods(\iType{t})
                                }   
                            }
                            \namedRule{\tNamed}{
                                \infer{
                                    t \ok
                                }{
                                    t \in \tdecls(\multi{D})
                            }}
                        \end{gathered}
                    \end{equation*}
                }
           Well-formed method specifications and type literals
                \hfill
                \fbox{$S \ok$} \  \fbox{$T \ok$}
            {\footnotesize 
                \begin{equation*}
                    \begin{gathered}
                        \namedRule{\tspecification}{
                            \infer{
                                m(\multi{x~t})~t \ok
                            }{
                                \distinct(\multi{x})
                                & \multi{t \ok}
                                & t\ok
                            }
                        }
                        \ 
                        \namedRule{\tstruct}{
                            \infer{
                               \struct{\multi{f~t}}\ok
                            }{  
                                \distinct(\multi{f})
                                & \multi{t\ok}
                            }
                        }
                        \
                        \namedRule{\tinterface}{
                            \infer{
                                \interface{\multi{S}}\ok
                            }{
                                \unique(\multi{S})
                                & \multi{S\ok}
                            }
                            }                        
                    \end{gathered}
                    \end{equation*}
                }
            Well formed declarations
                \hfill
                \fbox{$D \ok$} 
            {\footnotesize 
                \begin{equation*}
                    \begin{gathered}
                        \namedRule{\ttype}{
                            \infer{
                                \type~t~T \ok
                            }{
                                T \ok
                            }
                        }
                        \
                        \namedRuleTwo{\tfunc}{
                            \sType{t} \ok \quad \multi{t \ok} \quad u \ok
                        }{
                            \infer{
                                \funcDelc{\sType{t}}{m}{\multi{x~t}}{u}{\return e} \ok
                            }{
                                \wellTyped[x : \sType{t}, \multi{x : t}]{e}{u}
                                & \distinct(x, \multi{x})
                            }
                        }
                    \end{gathered}
                \end{equation*}
            }
            Expressions
            \hfill
            \fbox{$\wellTyped{e}{t}$}
            {\footnotesize
                \begin{equation*}
                    \begin{gathered}
                        \namedRule{\tvar}
                            {\infer{
                                \wellTyped{x}{t}
                            }{
                                (x:t)\in\Gamma
                            }
                        }
                        \namedRuleTwo{\tliteral}{\sType{t}\ok}{
                            \infer{
                                \wellTyped{\sTypeInit{t}{\multi{e}}}{\sType{t}}
                            } {
                                & \wellTypedMulti{e}{t}
                                & (\multi{f~u}) = fields(\sType{t})
                                & \multi{t <: u}
                            }
                        }\ 
                        \namedRule{\tcall}{
                            \infer{
                                \wellTyped{e.m(\multi{e})}{u}
                            } {
                                \wellTyped{e}{t}
                            & \wellTypedMulti{e}{t}
                            & (m(\multi{x~u})~u) \in \methods(t)
                            & \multi{t <: u}
                            }
                        }\\ 
                        \namedRule{\tfield}{
                            \infer {
                                \wellTyped{e.f_i}{u_i}
                            } {
                                \wellTyped{e}{\sType{t}}
                            & (\multi{f~u}) =\fields(\sType{t})
                            }
                        }\ 
                        \namedRule{\tasserts}{
                            \infer {
                                \wellTyped{e.(\sType{t})}{\sType{t}}
                            } {
                                \sType{t} \ok
                                & \wellTyped{e}{\iType{u}}
                                & \sType{t} <: \iType{u}
                            }
                        } \ 
                        \namedRule{\tasserti}{
                            \infer {
                                \wellTyped{e.(\iType{t})}{\iType{t}}
                            } {
                                \iType{t} \ok
                                & \wellTyped{e}{\iType{u}}
                            }
                        }\ 
                       \fbox{$
                            \namedRule{\tstupid}{
                                \infer {
                                    \wellTyped{e.(t)}{t}
                                } {
                                    t \ok
                                & \wellTyped{e}{\sType{u}}
                                }
                           }$
                        }
                    \end{gathered}
                \end{equation*}
            }
            Programs 
            \hfill
            \fbox{$\prog\ok$}
            {\footnotesize
            \begin{equation*}
                \begin{gathered}
                        \namedRule{\tprog}{
                            \infer {
                                \textbf{ package main};\ \multi{D}\ \textbf{func main}() \sytxBrace{\_ = e}\ok
                            } {
                                \distinct(\tdecls(\multi{D}))
                            & \distinct(\mdecls(\multi{D}))
                            & \multi{D\ok}
                            & \wellTyped[\emptyset]{e}{t}
                            }
                        }
                \end{gathered}
            \end{equation*}
            }
\rule{\columnwidth}{0.5mm}
{\footnotesize
    \begin{equation*}
            \begin{gathered}    
                \methods(\sType{t}) = \{ mM \mid (\func~(x~\sType{t})~mM~\{\return e\}) \in \multi{D}\}\\
\quad                 \infer{
                    \methods(\iType{t}) = \multi{S}
                } {
                    \type~\iType{t}~\interface{\multi{S}} \in \multi{D}
                } \quad 
                \infer{
                    \unique(\multi{S})
                } {
                    mM_1, mM_2 \in \multi{S}~\text{implies}~M_1 = M_2
                }\quad
                \infer{
                    \fields (t_S) = \multi{f~t}
                } {
                    (\type t_S \struct{\multi{f~t}}) \in \multi{D}
                }\\
                \tdecls(\multi{D}) = [t \mid (\type~t~T) \in
                \multi{D}] \quad  
                \mdecls(\multi{D}) = [t_S.m \mid (\func~(x~t_S)~mM~\{\return e\}) \in \multi{D}]
            \end{gathered}
        \end{equation*}
}
        \caption{Typing and auxiliary function.} 
        \label{fig:fg:types}
\end{figure}

%% file: sections/appendix-fgg.tex
\section{Appendix: \glsentrylong{fgg}}
\label{app:fgg}
\label{appendix:fgg}

\input{figs/fgg/fgg-typing}
For the reviewer's convenience, this section provides 
the definitions and more explanations of the typing
system from those in \cite{griesemer2020featherweight}. 
\subsection{\glsentrylong{fgg} Typing Rules}
Judgements are extended with the type environment $\Delta$ which
relates types than type names. 
The subtyping $\subtype{\tau}{\sigma}$ uses $\Delta$ where both $\tau$
and $\sigma$ may have type parameters in $\Delta$; 
judgement $\wellFormed[\Delta]{\tau}$ says 
the argument of channel type ($\tau$) is well-formed 
w.r.t. all type parameters declared in $\Delta$;  
a method declaration is well-formed if 
$\Phi$ and $\Psi$ are well-formed types formal of the receiver and the
method, yielding $\Delta$ ($\Phi;~\Psi\ok~\Delta$) and the receiver's type is
declared by $\Phi'$ such that $\Phi <: \Phi'$. 
Judgements for expressions,  
method calls and processes are extended w.r.t. $\Delta$, accordingly.

This is so that interface subtyping (\rulename{<:i}) may ensure that
type parameters still implement all methods that an interface
requires. 
The type formal subtyping rule ($\Phi <: \Psi$) ensures that if the
type substitution $\Phi :=_\Delta \multi{\tau}$ is well-defined, 
then $\Psi :=_\Delta \multi{\tau}$ is well-defined.  

We deviate from \cite{griesemer2020featherweight} 
in our typing for \rulename{\tfunc}. We require that 
receiver types formal are identical to those in the structures 
declaration. This more closely follows the official Go proposal \cite{type-parameters-proposal}.
Rather than require the developer to write a full type formal 
which must exactly match the structure's declaration they instead provide 
a receiver type parameter list which is converted to a type formal by 
looking up the structure type formal. 

When looking at method typing it becomes necessary to consider two
types formal, with the methods type formal ($\Psi$) depending on the
receiver's ($\Phi$).  A methods type environment $\Delta$ is
constructed by the well formed composition of the receiver and
method's types formal ($\Phi;~\Psi\ok~\Delta$). This environment is
well formed when $\Phi$ is well formed in an empty environment while
method's type formal $\Psi$ is well formed under $\Phi$.  Type formal
well formedness ($\wellFormed[\Phi]{\Psi}$) holds when there is no
repetition in the type parameters between $\Phi$ and $\Psi$ and all
bounds in $\Psi$ are well formed in the $\Phi, \Psi$ environment. This
definition allows mutually recursive bounds in $\Psi$.

A type declaration includes a type formal, this type formal is the
environment that either the structure or interface must be well formed
under. A structure is well formed in an environment $\Phi$ when each
of its fields is well formed under $\Phi$. An interface is well formed
in $\Phi$ when each method is specifies is also well formed under
$\Phi$. 

A method specifications is well-formed
($\wellFormed[\Phi]{m\typeFormalMethod(\multi{x~\tau})~\tau}$) when
its composite type environment is well-formed ($\Phi;~\Psi\ok~\Delta$)
and if its argument and return types are well-formed under that
composite type environment.

\begin{theorem}[Preservation]
{\rm (Theorem 4.3 in \cite[\S~4]{griesemer2020featherweight})} 
\label{lemma:fcgg:subjReduction:Expression}
\label{thm:fggTypePreservation}
    If $\wellTyped[\emptyset; \emptyset]{e}{\tau}$
    and $\reduction{e}{e'}$
    then $\wellTyped[\emptyset; \emptyset]{e'}{\tau'}$,
    for some $\tau'$ such that $\subtype[\emptyset; \emptyset]{\tau'}{\tau}$.
\end{theorem}

\begin{theorem}[Progress] 
\label{thm:fggProgress}
{\rm (Theorem 3.4 in \cite[\S~3.3]{griesemer2020featherweight})}
If $\wellTyped[\emptyset;\emptyset]{e}{\tau}$ then  
$e$ is either a value, $\reduction{e}{e'}$ for some $e'$ or 
$e$ panics. 
\end{theorem}

\pagebreak

%% file: figs/fgg/fgg-typing.tex
{
\begin{figure}
    \footnotesize
\onecolumn
    Implements and well-formed types
    \hfill
    \fbox{$\subtype{\tau}{\sigma} \vphantom{\ok}$} \quad \fbox{$\Phi<:\Psi$}
    {\footnotesize
        \begin{equation*}
            \begin{gathered}
                \namedRule{\subParam}{
                    \axiom{\subtype{\alpha}{\alpha}}
                }
                \namedRule{\subS}{
                    \axiom{\subtype{\sType{\tau}}{\sType{\tau}}}
                }
                \namedRule{\subI}{
                    \infer{
                        \subtype{\tau}{\iType{\sigma}}
                    }{
                        \methods[\Delta](\tau) \supseteq \methods[\Delta](\iType{\sigma})
                    }
                }
                \namedRule{\subFormal}{
                    \infer{
                        (\typeFormal) <: (\typeFormal[\multi{\beta~\iType{\sigma}}])
                    } {
                        \subtypeMulti[]{\iType{\tau}}{\iType{\sigma}}
                    }
                }
            \end{gathered}
        \end{equation*}
    }
    Well-formed types and actuals
    \hfill
    \fbox{$\wellFormed[\Delta]{\tau}$} \quad
    \fbox{$\wellFormed[\Delta]{\phi}$}
    {\footnotesize
        \begin{equation*}
            \begin{gathered}
                \namedRule{\tParam}{
                    \infer{
                        \wellFormed[\Delta]{\alpha}
                    }{
                        (\alpha : \iType{\tau}) \in \Delta
                    }
                }
                \namedRule{\tNamed}{
                \infer{
                \wellFormed[\Delta]{t[\phi]}
                }{
                \wellFormed[\Delta]{\phi}
                & \type~t[\Phi]~T \in \multi{D}
                & \eta = (\Phi \by_\Delta \phi)
                }}
                \namedRule{\tActual}{
                    \infer{
                        \wellFormed[\Delta]{\multi{\tau}}
                    }{
                        \wellFormedMulti[\Delta]{\tau}
                    }
                }
            \end{gathered}
        \end{equation*}
    }
    Well-formed types formal
    \hfill
    \fbox{$\wellFormed[\Phi]{\Psi}$} \quad
    \fbox{$\Phi; \Psi \ok \Delta$}
    {\footnotesize
        \begin{equation*}
            \begin{gathered}
                \namedRule{\tformal}{
                    \infer{
                        \wellFormed[{\typeFormal[\multi{\beta~\iType{\sigma}}]}]{\typeFormal}
                    }{
                        \distinct(\multi{\beta}, \multi{\alpha})
                        & \wellFormedMulti[{\typeFormal[\multi{\alpha~\iType{\tau}},\multi{\beta~\iType{\sigma}}]}]{\iType{\tau}}
                    }
                }
                \namedRule{\tnested}{
                    \infer{
                        \Phi; \Psi \ok \Delta
                    } {
                        \wellFormed[\emptyset]{\Phi}
                        & \wellFormed[\Phi]{\Psi}
                        & \Delta = \Phi, \Psi
                    }
                }
            \end{gathered}
        \end{equation*}
    }
    Well-formed method specifications,  type literals, and declarations
    \hfill
    \fbox{$\wellFormed[\Phi]{S}$} \  \fbox{$\wellFormed[\Phi]{T}$}\ 
    \fbox{$D \ok$}
    {\footnotesize
        \begin{equation*}
            \begin{gathered}
                \namedRuleTwo{\tspecification}{
                    \Phi;\Psi\ok\Delta
                }{
                    \infer{
                        \wellFormed[\Phi]{m[\Psi](\multi{x~\tau})~\tau}
                    }{
                        \distinct(\multi{x})
                        & \wellFormedMulti[\Delta]{\tau}
                        & \wellFormed[\Delta]{\tau}
                    }
                }
                \
                \namedRule{\tstruct}{
                    \infer{
                        \wellFormed[\Phi]{\struct{\multi{f~\tau}}}
                    }{
                        \distinct(\multi{f})
                        & \wellFormedMulti[\Phi]{\tau}
                    }
                }
                \
                \namedRule{\tinterface}{
                    \infer{
                        \wellFormed[\Phi]{\interface{\multi{S}}}
                    }{
                        \unique(\multi{S})
                        & \wellFormedMulti[\Phi]{S}
                    }
                }\\
                \namedRule{\ttype}{
                \infer{
                \type~t[\Phi]~T \ok
                }{
                \wellFormed[\emptyset]{\Phi}
                &\wellFormed[\Phi]{T}
                }
                }
                \namedRuleTwo{t-func}{
                    \distinct(x, \multi{x})
                    \quad (\type~\sType{t}\typeFormalType~T) \in \multi{D}
                }
                {
                    \infer
                    {
                        \funcDelc{\sType{t}[\multi{\alpha}]}{m[\Psi]}{\multi{x~\tau}}{\sigma}{\return e} \ok
                    }
                    {
                        \Phi; \Psi\ok \Delta
                        & \wellFormedMulti[\Delta]{\tau}
                        & \wellFormed[\Delta]{\sigma}
                        & \wellTyped[\Delta ; x : \sType{t}{[\multi{\alpha}]},
                            \multi{x : \tau}
                        ]{e}{\tau}
                        & \subtype{\tau}{\sigma}
                    }
                }
            \end{gathered}
        \end{equation*}
    }
    Expressions and 
    Programs 
    \hfill
    \fbox{$\wellTyped{e}{t}$} \ 
    \fbox{$\prog\ok$}
    {\footnotesize
        \let\oldwt\wellTyped
        \renewcommand{\wellTyped}[3][\Delta;\Gamma]{#1 \vdash #2 : #3}
        \renewcommand{\wellTypedMulti}[3][\Delta;\Gamma]{#1 \vdash \multi{#2 : #3}}
        \begin{equation*}
            \begin{gathered}
                \namedRule{\tvar}{
                    \infer{
                        \wellTyped{x}{\tau}
                    }{
                        (x:\tau)\in\Gamma
                    }
                }\!
                \namedRuleTwo{\tliteral}{
                    \wellFormed[\Delta]{\sType{\tau}}
                }{
                    \infer{
                        \wellTyped{\sTypeInit{\tau}{\multi{e}}}{\sType{\tau}}
                    } {
                        & \wellTypedMulti{e}{\tau}
                        & (\multi{f~\sigma}) = \fields(\sType{\tau})
                        & \multi{\tau <: \sigma}
                    }
                }\!
                \namedRuleTwo{\tcall}{
                    (m\typeFormalMethod(\multi{x~\sigma})~\sigma) \in \methods[\Delta](\tau)
                }{
                    \infer{
                        \wellTyped{e.m\typeActualMethod(\multi{e})}{\sigma[\eta]}
                    } {
                        \wellTyped{e}{\tau}
                        & \wellTypedMulti{e}{\tau}
                        & \eta = (\Psi \by_\Delta \psi)
                        & \subtypeMulti{\tau}{\sigma}[\eta]
                    }
                }\\
                \namedRule{\tfield}{
                    \infer {
                        \wellTyped{e.f_i}{\sigma_i}
                    } {
                        \wellTyped{e}{\sType{\tau}}
                        & (\multi{f~\sigma}) =\fields(\sType{\tau})
                    }
                }
                \namedRule{\tasserts}{
                    \infer {
                        \wellTyped{e.(\sType{\tau})}{\sType{\tau}}
                    } {
                        \wellFormed[\Delta]{\sType{\tau}}
                        & \wellTyped{e}{\jType{\sigma}}
                        & \subtype{\sType{\tau}}{\bounds(\jType{\sigma})}
                    }
                }
                \namedRule{\tasserti}{
                    \infer {
                        \wellTyped{e.(\jType{\tau})}{\jType{\tau}}
                    } {
                        \jType{\tau} \ok
                        & \wellTyped{e}{\jType{\sigma}}
                    }
                }\\
                \fbox{$
                        \namedRule{\tstupid}{
                            \infer {
                                \wellTyped{e.(\tau)}{\tau}
                            } {
                                \wellFormed[\Delta]{\tau}
                                & \wellTyped{e}{\sType{\sigma}}
                            }
                        }$
                }
                \namedRule{\tprog}{
                    \infer {
                        \textbf{ package main};\ \multi{D}\ \textbf{func main}() \sytxBrace{\_ = e}\ok
                    } {
                        \distinct(\tdecls(\multi{D}))
                        & \distinct(\mdecls(\multi{D}))
                        & \multi{D\ok}
                        & \wellTyped[\emptyset; \emptyset]{e}{t}
                    }
                }
            \end{gathered}
        \end{equation*}
    }
    \rule{\columnwidth}{0.5mm}
    {\footnotesize
        \begin{equation*}
            \begin{gathered}
                \infer{
                    \eta = (\typeFormal \by \multi{\tau})
                }{
                    \eta = (\multi{\alpha \by \tau})
                }
                \qquad
                \infer{
                    \eta = (\typeFormal \by_\Delta \phi)
                }{
                    \eta = (\typeFormal \by \phi)
                    & \Delta \vdash \multi{\alpha <: \iType{\tau}}[\eta]
                }
                \qquad
                \infer{
                    \fields (\sType{t}\typeActualReceive) = \multi{f~\tau}[\eta]
                } {
                    (\type \sType{t}\typeFormalType \struct{\multi{f~\tau}}) \in \multi{D}
                    & \eta = (\Phi \by \phi)
                }
                \\
                \methods[\Delta](\sType{t}\typeActualReceive) = \{
                mM[\eta] \mid (\func~(x~\sType{t}\typeFormalReceive)~mM~\{\return e\}) \in \multi{D},
                (\type~\sType{t}\typeFormalType~T) \in \multi{D},
                \eta = (\Phi \by_\Delta \phi)
                \}
                \\
                \infer{
                    \methods[\Delta](\iType{t}\typeActualReceive) = \multi{S}[\eta]
                } {
                    \type~\iType{t}\typeFormalType~\interface{\multi{S}} \in \multi{D}
                    & \eta = (\Phi \by \phi)
                }
                \quad
                \infer{
                    \methods[\Delta](\alpha) = \methods[\Delta](\iType{\tau})
                }{
                    (\alpha : \iType{\tau}) \in \Delta
                }
                \quad
                \infer{
                    \unique(\multi{S})
                } {
                    mM_1, mM_2 \in \multi{S}~\text{implies}~M_1 = M_2
                }
                \\
                \tdecls(\multi{D}) = [t \mid (\type~t\typeFormalType~T) \in \multi{D}]
                \quad
                \mdecls(\multi{D}) = [\sType{t}.m \mid (\func~(x~\sType{t}\typeFormalReceive)~mM~\{\return e\}) \in \multi{D}]
                \\
                \infer{
                    \bounds(\alpha) = \iType{\tau}
                }{
                    (\alpha : \iType{\tau}) \in \Delta
                }
                \quad
                \bounds(\sType{\tau}) = \sType{\tau}
                \quad
                \bounds(\iType{\tau}) = \iType{\tau}
                \\
                \infer{
                    \body(\sType{t}\typeActualReceive.m\typeActualMethod) = (x:\sType{t}\typeActualReceive, \multi{x:\tau}).e[\theta]
                }{
                    (\funcDelc{\sType{t}\typeFormalReceive}{m\typeFormalMethod}{\multi{x~\tau}}{\tau}{\return e}) \in \multi{D}
                    \quad \theta = (\multi{\alpha},\Psi \by \phi,\psi)
                }
            \end{gathered}
        \end{equation*}
    }
    \caption{Typing and auxiliary function.}
    \label{fig:fcg:types}
\end{figure}
}

%% file: proof/dict/correctness.tex
\section{Appendix of Section~\ref{section:properties}}
\label{app:proofs}

This section provides the detailed proofs for the main theorem.

\begin{definition}[Dictionary map]
    To simplify definitions we sometimes use a functional notation ($\map[]{-}$)
    for the dictionary-passing translation, defined as
    \begin{align*}
        \map{e}       & = \lex e~\text{\it where }~\dict{e}{\lex{e}}                                     \\
        \map[]{\prog} & = \lex \prog~\text{\it where }~\dict[]{\prog}{\lex{\prog}}                       \\
        \map{\Gamma}  & = \dom{\Gamma} : \any, \multi{\dictlit : \dictName{\Delta(\eta^{-1}(\dictlit))}}
    \end{align*}
    While $\map{\Gamma}$ seems complex at first glance it simply
    erases all variables already in $\Gamma$ while adding
    appropriate dictionary variables. This is done by finding
    the type parameter name for $\dictlit_i$ from $\eta^{-1}$
    and then getting the type bound of the type parameter from $\Delta$.
    This type bound is used to decide the type of the $\dictlit_i$ variable.
    Note that $\eta$ is bijective since we assume all type parameter names and
    dictionary variable names are unique, as such $\eta^{-1}$ exists.
  \end{definition}

\lemprec*
\begin{proof}
    \begin{enumerate}
        \item This case is shown by the deterministic nature of $\red$ and the fact
              that $\dicttrans \subseteq \red^+$.
        \item This case is immediate for most $\dictpattern$ cases. Only the
              $\methName{t,m}\sytxBrace{}.\apply(\multi{e})$ case provides
              a complexity. Once we realise, however, that each $e_i$ is
              used linearly in the method $\methName{t,m}\sytxBrace{}.\apply$,
              as created by $\makeDictMeth{}$,
              it becomes clear that this case does not interact with any
              reductions in $\multi{e}$.
    \end{enumerate}
\end{proof}

\begin{restatable}[Type specialisation is resolved by $\precongdict$ ]
    {lemrest}{lempreposttype}
  \label{lem:preposttype}
  Let $\Delta = \multi{\alpha : \iType{\tau}}$,
  expression $e$ be of type
  $\wellTyped[\Delta; \Gamma]{e}{\tau}$,
  and $\multi{\sigma}$ be a type actual such that $\subtypeMulti[\Delta]{\sigma}{\iType{\tau}}$.
  Let the map $\eta = \{\multi{\alpha \mapsto \dictlit}\}$.
  then
  $\map{e}[
          \multi{\dictlit \by \makeDict[\emptyset; \emptyset]{\sigma, \iType{\tau}}}
      ] \precongdict^*
      \map[\emptyset; \emptyset; \Gamma]{e[\multi{\alpha \by \sigma}]}$
\end{restatable}

\begin{proof}
    By induction on $e$, we apply the substitution of each $\alpha_i$ in turn.
    Note that dictionary $\eta(\alpha_i)$ is either of the form $v.\dictlit_i$ when
    $\alpha_i$ comes from the receiver or $\dictlit_i$ when it is a method parameter.
    We can limit our considerations to the latter case as
    if the former holds we can transform it to the latter as
    $\dictpattern$ resolves the desired receiver field access.

    \begin{itemize}
        \item[] \caseof{\ddictcall}\\
              We assume that $\wellTyped[\alpha:\iType{\tau}; \Gamma]{e}{\alpha}$
              such that the translation of
              $e.m\typeActualMethod(\multi{e})$ is
              \[\dict{e.m\typeActualMethod(\multi{e})}
                  {\eta(\alpha_i).m.\apply(\map{e}, \map{\psi}, \multi{\map{e}})}\]
              $\eta(\alpha_i)$ is either of the form $v.\dictlit_i$ when
              $\alpha_i$ comes from the receiver or $\dictlit_i$ when it is a method parameter.
              In the former case $\dictpattern$ resolves the desired receiver field access,
              so we may limit our considerations the latter.

              Let $\sigma = u[\phi]$ and $\iType{\tau} = \iType{t}[\phi']$ then
              applying the substitution $\theta = [\dictlit \by \makeDict[\emptyset; \emptyset]{\sigma, \iType{\tau}}]$
              we produce the term
              \[\makeDict[\emptyset; \emptyset]{u[\phi], \iType{t}[\phi']}.m.\apply(\map{e}[\theta], \map{\psi}[\theta], \multi{\map{e_2}[\theta]})\]
              where method $m$ is in the method abstractors of the given dictionary, such that
              \[\makeDict[\emptyset; \emptyset]{u[\phi], \iType{t}[\phi']} = \dictName{\iType{t}}\{\multi{v}, \methName{u,m}, \multi{v}\}\]

              $\dictpattern$ resolves as follows
              \begin{flalign*}
                  \qquad    \dictName{\iType{t}}\{\cdots\}.&m.\apply(\map{e}[\theta], \map{\psi}[\theta], \multi{\map{e}[\theta]}) &\\
                  & \precongdict \methName{u,m}.\apply(\map{e}[\theta], \map{\psi}[\theta], \multi{\map{e}[\theta]}) &\\
                  & \precongdict \map{e}[\theta].(u).m(\map{\psi}[\theta], \multi{\map{e}[\theta]})
              \end{flalign*}
              By the induction hypothesis
              \begin{flalign*}
                \qquad    \map{e}[\theta].(u)&.m(\map{\psi}[\theta], \multi{\map{e}[\theta]}) \precongdict^* \\
                & \map[\emptyset; \emptyset; \Gamma]{e[\multi{\alpha \by \sigma}]}.(u).m(\map[\emptyset; \emptyset; \Gamma]{\psi[\multi{\alpha \by \sigma}]}, \multi{\map{e[\emptyset; \emptyset; \Gamma][\multi{\alpha \by \sigma}]}})
            \end{flalign*}

              Note that we can only apply induction on the arguments to $m$
              because $\dictpattern$ is defined on the pre-congruence evaluation
              context $\congEval$. We explicitly do not use $\assertpattern$
              as part of this induction.

              To resolve $\map[\emptyset; \emptyset; \Gamma]{e.m\typeActualMethod(\multi{e})[\multi{\alpha \by \sigma}]}$
              we need first observe that the type substitution specifies all type variables
              by our assumptions. This means that the dictionary-passing translation uses
              the homomorphic rule \rulename{\dcall}. We also know that the type of $e$ ($\alpha$)
              is mapped to $\sigma$ ($u[\phi]$).
              \begin{flalign*}
                  \qquad & \map[\emptyset; \emptyset; \Gamma]{e.m\typeActualMethod(\multi{e})[\multi{\alpha \by \sigma}]} &\\
                  &\qquad = \map[\emptyset; \emptyset; \Gamma]{e[\multi{\alpha \by \sigma}].m[\psi[\multi{\alpha \by \sigma}]](\multi{e[\multi{\alpha \by \sigma}]})} &\\
                  &\qquad = \map[\emptyset; \emptyset; \Gamma]{e[\multi{\alpha \by \sigma}]}.(u).m(\map[\emptyset; \emptyset; \Gamma]{\psi[\multi{\alpha \by \sigma}]}, \multi{\map[\emptyset; \emptyset; \Gamma]{e[\multi{\alpha \by \sigma}]}})
              \end{flalign*}
        \item[] \caseof{\dassert} $\wellTyped[\Delta, \Gamma]{e.(\alpha)}{\alpha}$\\
              We start by considering the $e.(\alpha)[\alpha \by \sigma]$ side.
              \begin{flalign*}
                  \qquad & \map[\emptyset; \emptyset; \Gamma]{e.(\alpha)[\alpha \by \sigma]}
                  = \map[\emptyset; \emptyset; \Gamma]{e[\alpha \by \sigma].(\sigma)}
                  = \typemeta[\emptyset]{\sigma}.\trycast(\map[\emptyset; \emptyset; \Gamma]{e[\alpha \by \sigma]}) &
              \end{flalign*}
              We now look at the lhs. Let $\zeta = (-.\typeField) \circ \eta$
              \begin{flalign*}
                  \qquad  \map{e.(\alpha)} & = \typemeta{\alpha}.\trycast(\map{e}) &\\
                  & = \zeta(\alpha).\trycast(\map{e}) \\
                  & = (\eta(\alpha).\typeField).\trycast(\map{e})\\
                  & = \dictlit.\typeField.\trycast(\map{e})
              \end{flalign*}
              If $\iType{\tau} = \iType{t}[\phi]$ then we have that
              $\makeDict[\emptyset; \emptyset]{\sigma, \iType{\tau}} =
                  \dictName{\iType{t}}\sytxBrace{\multi{v}, \typemeta{\sigma}}$.
              We also know that $\sigma$ is well typed in the empty environment ($\subtype[]{\sigma}{\iType{\tau}}$)
              so $\typemeta{\sigma} = \typemeta[\emptyset]{\sigma}$.
              With the previously derived $\map{e.(\alpha)} = \dictlit.\typeField.\trycast(\map{e})$
              {\small
              \begin{flalign*}
                  \quad \!& \dictlit.\typeField.\trycast(\map{e})[\dictlit \by \dictName{\iType{t}}\sytxBrace{\multi{v}, \typemeta[\emptyset]{\sigma}}] &\\
                  & \quad = \dictName{\iType{t}}\sytxBrace{\multi{v}, \typemeta[\emptyset]{\sigma}}.\typeField.\trycast(\map{e}[\dictlit \by \dictName{\iType{t}}\sytxBrace{\multi{v}, \typemeta[\emptyset]{\sigma}}])\\
                  & \quad \precongdict \typemeta[\emptyset]{\sigma}.\trycast(\map{e}[\dictlit \by \dictName{\iType{t}}\sytxBrace{\multi{v}, \typemeta[\emptyset]{\sigma}}])
              \end{flalign*}
              }
              We can now apply the induction hypothesis
              \begin{flalign*}
                  \qquad & \typemeta[\emptyset]{\sigma}.\trycast(\map{e}[\dictlit \by \dictName{\iType{t}}\sytxBrace{\multi{v}, \typemeta[\emptyset]{\sigma}}]) &\\
                  & \quad \precongdict^* \typemeta[\emptyset]{\sigma}.\trycast(\map[\emptyset; \emptyset; \Gamma]{e[\alpha \by \sigma]})
              \end{flalign*}
        \item[] \caseof{\dassert} $\wellTyped[\Delta, \Gamma]{e.(\tau)}{\tau}$
              If $\alpha \not\in \fv{\tau}$ then this proof is immediate by induction.
              If we instead assume that $\alpha \in \fv{\tau}$ then
              with $\zeta = (-.\typeField) \circ \eta$
              \begin{flalign*}
                  \qquad & \map{e.(\tau)}[\dictlit \by \makeDict[\emptyset; \emptyset]{\sigma, \iType{\tau}}] &\\
                  & \quad = \typemeta{\tau}.\trycast(\map{e})[\dictlit \by \makeDict[\emptyset; \emptyset]{\sigma, \iType{\tau}}]
              \end{flalign*}
              We further assume that $\tau = t[\alpha, \multi{\sigma}]$. While it may be that
              $\alpha$ is a type specialisation of a type specialisation ($t[u[\alpha]]$),
              this case does not significantly alter the proof, so we assume $\tau$
              is as given. Naturally this also applies if $\alpha$ is used more than
              once in $\tau$ (we assume $\alpha \not\in \fv{\multi{\sigma}}$).

              Computing $\typemeta{\tau}$ we get $\metadataName{t}\sytxBrace{\zeta(\alpha), \multi{\typemeta{\sigma}}}$,\\
              with $\zeta(\alpha) = \dictlit.\typeField$.
              {\small
              \begin{flalign*}
                  \qquad & = \typemeta{\tau}.\trycast(\map{e})[\dictlit \by \makeDict[\emptyset; \emptyset]{\sigma, \iType{\tau}}]&\\
                  & \quad = \metadataName{t}\sytxBrace{\dictlit.\typeField, \multi{\typemeta{\sigma}}}.\trycast(\map{e})[\dictlit \by \makeDict[\emptyset; \emptyset]{\sigma, \iType{\tau}}]\\
              \end{flalign*}
              }
              Furthermore we know $\makeDict[\emptyset; \emptyset]{\sigma, \iType{\tau}} = \dictName{\iType{t}}\sytxBrace{\multi{v}, \typemeta{\sigma}}$,
              and so the above substitution becomes
              {\small
              \begin{flalign*}
                  \qquad & = \metadataName{t}\sytxBrace{\dictlit.\typeField, \multi{\typemeta{\sigma}}}.\trycast(\map{e})[\dictlit \by \makeDict[\emptyset; \emptyset]{\sigma, \iType{\tau}}]&\\
                  & = \metadataName{t}\sytxBrace{\dictName{\iType{t}}\sytxBrace{\multi{v}, \typemeta{\sigma}} \\
                      & \qquad \qquad .\typeField, \multi{\typemeta{\sigma}}}.\trycast(\map{e}[\dictlit \by \makeDict[\emptyset; \emptyset]{\sigma, \iType{\tau}}])\\
                  & \precongdict \metadataName{t}\sytxBrace{\typemeta{\sigma}, \multi{\typemeta{\sigma}}}.\trycast(\map{e}[\dictlit \by \makeDict[\emptyset; \emptyset]{\sigma, \iType{\tau}}])\\
              \end{flalign*}
              }
              Applying induction
              {\small
              \begin{flalign*}
                  \qquad &\metadataName{t}\sytxBrace{\typemeta{\sigma}, \multi{\typemeta{\sigma}}}.\trycast(\map{e}[\dictlit \by \makeDict[\emptyset; \emptyset]{\sigma, \iType{\tau}}])&\\
                  &\quad \precongdict^*  \metadataName{t}\sytxBrace{\typemeta{\sigma}, \multi{\typemeta{\sigma}}}.\trycast(\map[\emptyset; \emptyset; \Gamma]{e[\alpha \by \sigma]})
              \end{flalign*}
              }

              We now look to $\map[\emptyset; \emptyset; \Gamma]{e.(\tau)[\alpha \by \sigma]}$.
              We again assume that $\tau = t[\alpha, \multi{\sigma}]$ with $\alpha \not\in \fv{\multi{\sigma}}$
              \begin{flalign*}
                  \qquad &\map[\emptyset; \emptyset; \Gamma]{e.(\tau)[\alpha \by \sigma]} &\\
                  & \quad = \map[\emptyset; \emptyset; \Gamma]{e.(t[\alpha, \multi{\sigma}])[\alpha \by \sigma]}\\
                  & \quad = \map[\emptyset; \emptyset; \Gamma]{e[\alpha \by \sigma].(t[\sigma, \multi{\sigma}])}\\
                  & \quad = \typemeta[\emptyset]{t[\sigma, \multi{\sigma}]}.\trycast(\map[\emptyset; \emptyset; \Gamma]{e[\alpha \by \sigma]})\\
                  & \quad = \metadataName{t}\sytxBrace{\typemeta{\sigma}, \multi{\typemeta{\sigma}}}.\trycast(\map[\emptyset; \emptyset; \Gamma]{e[\alpha \by \sigma]})
              \end{flalign*}
    \end{itemize}
\end{proof}

\begin{restatable}[Subtype preservation]{lemrest}{lemsubtypepres}
    \label{lem:sub:pres}
    Let $\dict[]{\prog}{\lex{\prog}}$. 
    If $\subtype[\emptyset]{u[\psi]}{t\typeActualReceive}$ in 
    $\prog$ then 
    $u<:t$ in $\lex{\prog}$.
\end{restatable}
\begin{proof}
    By case analysis on $\subtype[\emptyset]{u[\psi]}{t\typeActualReceive}$. 
  \end{proof}

\begin{restatable}[Value substitution is compositional upto $\precongdict$ ]
    {lemrest}{lemprepostval}
  \label{lem:prepostval}
  Let $\Gamma = \multi{x : \tau}$,
  expression $e$ be of type
  $\wellTyped[\emptyset; \Gamma]{e}{\tau'}$,
  and expressions $\multi{v}$ be typed by
  $\wellTypedMulti[\emptyset; \emptyset]{v}{\sType{\sigma}}$
  such that
  $\subtypeMulti[\emptyset]{\sType{\sigma}}{\tau}$.
  We have that
  $\map[\emptyset, \emptyset, \Gamma]{e}[
          \multi{x\by \map[\emptyset, \emptyset, \emptyset]{v}}
      ] \precongdict^*
      \map[\emptyset; \emptyset; \emptyset]{e[\multi{x\by v}]}$
\end{restatable}

\begin{proof}
    By induction on the translation rule used,
    we apply the substitution of each $x_i$ in turn.

    \begin{itemize}
        \item[] \caseof{d-call}
              \[
                  \namedRule{\dcall}{
                      \infer{
                          \dict[\emptyset; \emptyset; \Gamma]{
                              e.m\typeActualMethod(\multi{e})
                          }{
                              \map[\emptyset;\emptyset;\Gamma]{e}.(t).m(\lex{\psi}, \multi{ \map[\emptyset;\emptyset;\Gamma]{e}})
                          }
                      }{
                          \wellTyped[\emptyset; \Gamma]{e}{t\typeActualReceive}
                          & \lex{\psi} = \makeDict[\emptyset; \emptyset]{\psi, \Psi}
                          & (m[\Psi](\multi{x~\tau})~\tau) \in \methods[\Delta](t\typeActualReceive)
                      }
                  }
              \]
              By the substitution lemma \cite[Lemma~4.2]{griesemer2020featherweight}
              we have
              \begin{pfsteps*}
                  \pf[1]{\wellTyped[\emptyset; \emptyset]{e[\multi{x\by v}]}{u[\psi']}}
                  \pf[2]{\subtype[\emptyset]{u[\psi']}{t\typeActualReceive}}
              \end{pfsteps*}
              and
              \begin{pfsteps*}
                  \pf[3]{\dict[\emptyset; \emptyset; \emptyset]{
                          e[\multi{x\by v}].m\typeActualMethod(\multi{e[\multi{x\by v}]})
                      }{
                          \map[\emptyset;\emptyset;\emptyset]{e[\multi{x\by v}]}.(u).m(\lex{\psi}, \multi{ \map[\emptyset;\emptyset;\emptyset]{e[\multi{x\by v}]}})
                      }}
              \end{pfsteps*}
              By lemma~\ref{lem:sub:pres} \pfref{2} we have that $u<: t$.
              We now have
              \begin{pfsteps*}
                  \pf[4]{\map[\emptyset;\emptyset;\Gamma]{e}.(t).m(\lex{\psi}, \multi{ \map[\emptyset;\emptyset;\Gamma]{e}})
                      \precongdict
                      \map[\emptyset;\emptyset;\Gamma]{e}.(u).m(\lex{\psi}, \multi{ \map[\emptyset;\emptyset;\Gamma]{e}})
                  }
              \end{pfsteps*}
              and by the induction we have
              \begin{pfsteps*}
                  \pf[5]{
                      \map[\emptyset;\emptyset;\Gamma]{e}.(u).m(\lex{\psi}, \multi{ \map[\emptyset;\emptyset;\Gamma]{e}})
                      \precongdict^*
                      \map[\emptyset;\emptyset;\emptyset]{e[\multi{x\by v}]}.(t).m(\lex{\psi}, \multi{ \map[\emptyset;\emptyset;\emptyset]{e[\multi{x\by v}]}})
                  }
              \end{pfsteps*}
    \end{itemize}
\end{proof}

\begin{lemma}[Method specification simulaiton preserves substitution]
    \label{lem:methspec}
    Let $\wellFormedMulti[\alpha:\iType{\tau}]{M}$, and assume $\multi{\sigma}$ such that
    $\subtypeMulti[\emptyset]{\sigma}{\iType{\tau}}$.
    We also assume $\this = \metadataName{\iType{t}}\sytxBrace{\multi{\typemeta[\emptyset]{\sigma}}}$.
    For $\zeta = \{\multi{\alpha \mapsto \this.\typeField}\}$ it
    holds that $\signatureMeta[\zeta]{M}~\prepostsim^*~\signatureMeta[\emptyset]{M[\alpha \by \sigma]}$.
\end{lemma}
\begin{proof}
    Since each $\alpha_i$ is distinct we can consider each separately.
    We begin by noting that $\arities{M[\alpha_i \by \sigma_i]} = \arities{M} = n$.
    We also define a suitable $\zeta'$ for the $\paramTypeMeta\{\}$ map,
    such that $\alpha_i \not \in \dom{\zeta'}$.
    \begin{flalign*}
        \qquad & \signatureMeta[\emptyset]{M[\alpha_i \by \sigma_i]}&\\
        & \quad = \fnMeta{n}\sytxBrace{
            \multi{\typemeta[\zeta']{\tau[\alpha_i \by \sigma_i]}}
        }\\
        & \signatureMeta[\zeta]{M} \\
        &\quad = \fnMeta{n}\sytxBrace{
            \multi{\typemeta[\zeta', \zeta]{\tau}}
        }
    \end{flalign*}
    Where $\multi{\tau} = \multi{\tau_0}, \multi{\tau_1}, \tau_2$ for
    $M = [\multi{\beta ~ \tau_0}](\multi{x~\tau_1})\tau_2$.
    It now suffices to show that for all $\tau$. $\typemeta[\zeta', \zeta]{\tau} \prepostsim^*
        \typemeta[\zeta']{\tau[\alpha_i \by \sigma_i]}$. This is done by induction on $\tau$.
    \begin{itemize}
        \item[] \caseStd $\tau = \alpha_i$\\
              The term $\typemeta[\zeta']{\alpha_i[\alpha_i \by \sigma_i]}$ becomes $\typemeta[\zeta']{\sigma_i}$
              which is equal to $\typemeta[\emptyset]{\sigma_i}$ since $\sigma_i$ is defined
              outside the scope of $\zeta'$.
              The other term $\typemeta[\zeta', \zeta]{\alpha_i}$ is equal to
              $(\zeta', \zeta)(\alpha_i)$ which by the definition of $\zeta$ is\\
              $\metadataName{\iType{t}}\sytxBrace{\multi{\typemeta[\emptyset]{\sigma_i}}}.\typeField_i \prepostsim \typemeta[\emptyset]{\sigma_i}$.
        \item[] \caseStd $\tau = \beta$ where $\beta \noteq \alpha_i$
              Both terms are immediately equal to $\zeta'(\beta)$.
        \item[] \caseStd $\tau = t[\multi{\tau}]$
              By induction on each $\tau_i$.
    \end{itemize}
\end{proof}

\begin{lemma}
    \label{lem:typeformalsmatchasparams}
    If $\Phi \by_\Delta \phi$ with $\eta$ such that 
    $\dom{\eta} = \dom{\Delta}$ then
    $\vtype(\makeDict{\phi, \Phi}) <: \asParameters{\Phi}$.
\end{lemma}
\begin{proof}
    Immediately from the definition of $\makeDict[]{}$ and
    $\asParameters{}$.
\end{proof}

\lemtypepres*
\begin{proof}
    By induction on the type of $e$.
    \begin{itemize}
        \item[] \caseof{\tfield}
              \begin{pfsteps*}
                  \pf[1]{\wellTyped[\Delta; \Gamma]{e.f}{\tau}}
                  \pf[2]{\dict{e.f}{\lex{e}.(\sType{t}).f}}
              \end{pfsteps*}
              For \pfref{1} to hold $e$ must be of type $\sType{t}\typeActualReceive$.
              By the induction hypothesis, either $\wellTyped[\map{\Gamma}]{\map{e}}{\any}$
              or $\wellTyped[\map{\Gamma}]{\map{e}}{\sType{t}}$.
              In either case $\map{e}.(\sType{t})$ is well typed by \rulename{\tasserts}
              or \rulename{\tstupid} (\emph{resp.}). Since $f$ is a field of
              type $\sType{t}\typeActualReceive$ it must also be a field of $\sType{t}$.
              We get the final typing judgement $\wellTyped[\map{\Gamma}]{\map{e}.(\sType{t}).f}{\any}$.
        \item[] \caseof{\tvar}
              Immediate by our definition of $\map{\Gamma}$.
        \item[] \caseof{\tliteral}
              \begin{pfsteps*}
                  \pf[1]{\wellTyped[]{\sType{t}\typeActualReceive\sytxBrace{\multi{e}}}{\sType{t}\typeActualReceive}}
                  \pf[2]{\dict{\sType{t}\typeActualReceive\sytxBrace{\multi{e}}}{
                          \sType{t}\sytxBrace{\multi{\lex{e}}, \lex{\phi}}}}
                  \pf[3]{\lex{\phi} = \makeDict{\phi, \Phi}}
                  \pfstep{\inversion{\tliteral}}
                  {4}{\type \sType{t}\typeFormalType \struct{x~\tau}}
                  \pfstep{\rulename{\dstruct}}
                  {5}{\dict[]{\type \sType{t}\typeFormalType \struct{x~\tau}}{
                          \type \sType{t} \struct{x~\any, \asParameters{\Phi}}
                      }}
              \end{pfsteps*}
              Each $\lex{e_i}$ implements the $\any$ type while by
              Lemma \ref{lem:typeformalsmatchasparams}, $\lex{\phi}$
              implements $\asParameters{\Phi}$. As such
              \begin{pfsteps*}
                  \pf[6]{\wellTyped{\sType{t}\sytxBrace{\multi{\lex{e}}, \lex{\phi}}}{\sType{t}}}
              \end{pfsteps*}
        \item[] \caseof{\tcall}
              \begin{pfsteps*}
                  \pf[1]{\wellTyped{e.m\typeActualMethod(\multi{e})}{\tau}}
                  \pfSubcase{\wellTyped{e}{\alpha}}
                  \pf[2]{\dict{
                          e.m\typeActualMethod(\multi{e})
                      }{
                          \eta(\alpha).m.\apply(\lex{e}, \lex\psi, \multi{\lex{e}})
                      }}
              \end{pfsteps*}
              Since \pfref{1} we know that the bounds of $\alpha$
              ($\iType{t}\typeActualReceive = \Delta(\alpha)$)
              contains the method $m$ and that the type of the
              dictionary $\eta(\alpha)$ is $\dictName{\iType{t}}$
              we know that $\dictName{\iType{t}}$ has a field $m$.
              We further know that the field $m$ has type
              $\nAryFunction{n}$ where $n = |\psi| + |\multi{e}|$.
              Because all arguments to the $\apply$ method are of type
              $\any$ the rhs of \pfref{2} is well typed.
              \pfSubcase{\wellTyped{e}{t\typeActualReceive}}
              \begin{pfsteps*}
                  \pf[3]{\dict{
                          e.m\typeActualMethod(\multi{e})
                      }{
                          \lex{e}.m(\lex\psi, \multi{\lex{e}})
                      }}
              \end{pfsteps*}
              We can combine the cases where $t$ is a structure
              or an interface since \rulename{\dmeth} and
              \rulename{\dspec} both do the same thing.
              If $m\typeFormalMethod(\multi{x~\tau})~\tau \in
                  \methods_{\Delta}(t\typeActualReceive)$ then the translation
              produces $m(\lex\Psi, \multi{x~\any})~\any \in
                  \methods(t)$.
    \end{itemize}
\end{proof}

\thmopcorrespond*

\begin{proof}

    By induction over the assumed reduction.

    \begin{itemize}
        \item[] \caseof{\rfields} --- (a) direction
              \begin{pfsteps*}
                  \pf[1]{\sType{t}\typeActualReceive\sytxBrace{\multi{v}}.f_i \red v_i}
                  \pf[2]{\dict[\emptyset; \emptyset; \emptyset]{
                          \sType{t}\typeActualReceive\sytxBrace{\multi{v}}.f_i
                      }{
                          \sTypeInit{t}{\multi{\lex{v}}, \makeDict[\emptyset; \emptyset]{\phi, \Phi}}.f_i
                      }}
              \end{pfsteps*}
              Inversion on \rulename{\rfields} \pfref{1} and the the definition of
              $\fields$ gives us
              \begin{pfsteps*}
                  \pf[3]{(\multi{f~\tau})[\eta] = \fields(\sType{t}\typeActualReceive)}
                  \pf[4]{\type \sType{t}\typeFormalType \struct{\multi{f~\tau}}\in \multi{D}}
              \end{pfsteps*}
              Applying the dictionary translation rule \rulename{\dstruct} to \pfref{4} we get
              \begin{pfsteps*}
                  \pf[5]{\type \sType{t} \struct{\multi{f~\any}, \multi{\dictlit~u}} \in \multi{\lex{D}}}
                  \pf[6]{(\multi{f~\any}, \multi{\dictlit~u}) = \fields(\sType{t})}
                  \pfstep{\rulename{\rfields} \pfref{6, 2}}
                  {7}{\reduction{
                          \sTypeInit{t}{\multi{\lex{v}}, \makeDict[\emptyset; \emptyset]{\phi, \Phi}}.f_i
                      }{
                          \lex{v_i}
                      }}
                  \pfstep{\inversion{\dvalue} \pfref{2}}
                  {8}{
                      \dict[\emptyset;\emptyset;\emptyset]{v_i}{\lex{v_i}}
                  }
              \end{pfsteps*}
        \item[] \caseof{\rfields} --- (b) direction mostly the same as
              the (a) direction, since there are no $\precongdict$ reductions available to
              $\sTypeInit{t}{\multi{\lex{v}}, \makeDict[\emptyset; \emptyset]{\phi, \Phi}}.f_i$
              as both $\multi{\lex{v}}$ and $\makeDict[\emptyset;\emptyset]{\phi,\Phi}$ are values.
        \item[] \caseof{\rcall} --- (a) direction \\
              We begin by stating our assumptions explicitly
              \begin{pfsteps*}
                  \pf[1]{\reduction{
                          v.m\typeActualMethod(\multi{v})
                      }{
                          e[\theta][\this\by v, \multi{x\by v}]
                      }
                  }
                  \pf[2]{
                      \wellTyped[\emptyset; \emptyset]{v.m\typeActualMethod(\multi{v})}{\tau[\theta]}
                  }
              \end{pfsteps*}
              with $v$ of the form
              \begin{pfsteps*}
                  \pf[vform]{v = \sType{t}\typeActualReceive\sytxBrace{\multi{v_1}}}
                  \pf[7]{\wellTyped[\emptyset; \emptyset]{v}{\sType{t}\typeActualReceive}}
              \end{pfsteps*}
              By analysing the proof tree of \pfref{1} using inversion on \rulename{\rcall}
              and the definition of $\body$ we get
              \begin{pfsteps*}
                  \pf[4]{
                      (\this:\sType{t}\typeActualReceive,~ \multi{x:\tau}).e[\theta]
                      = \body(\vtype(v).m\typeActualMethod)
                  }
                  \pf[5]{\theta = (\Phi, \Psi := \phi, \psi)}
                  \pf[6]{\funcDelc{\sType{t}\typeFormalReceive}{m\typeFormalMethod}{\multi{x~\tau}}{\tau}{\return~e} \in \multi{D}}
              \end{pfsteps*}
              and so $v.m\typeActualMethod(\multi{v})$ is translated using rule \rulename{\dcall}
              \begin{pfsteps*}
                  \pf[8]{\dict[\emptyset;\emptyset; \emptyset]{
                          v.m\typeActualMethod(\multi{v})
                      }{
                          \lex{v}.(\sType{t}).m(\makeDict[\emptyset; \emptyset]{\psi, \Psi}, \multi{\lex{v}})
                      }}
              \end{pfsteps*}
              where $\lex{v}$ is defined using \rulename{\dvalue}
              \begin{pfsteps*}
                  \pf[vddagger]{
                      \dict[\emptyset; \emptyset; \emptyset]{
                          \sType{t}\typeActualReceive\sytxBrace{\multi{v_1}}
                      }{
                          \sType{t}\sytxBrace{\multi{\lex{v_1}}, \makeDict[\emptyset, \emptyset]{\phi, \Phi}}
                      }
                  }
              \end{pfsteps*}
              With $\Phi = (\typeFormal)$ and
              $\Psi = (\typeFormal[\multi{\beta~\iType{u}[\Psi']}])$,
              the method definition \pfref{6} is translated using \rulename{\dmeth}
              \begin{pfsteps*}
                  \pf[9]{\eta = \multi{\alpha\mapsto\this.\dictlit}, \multi{\beta \mapsto \dictlit}}
                  \pf[10]{\dict[\Phi,\Psi; \eta; \this : {\sType{t}[\multi{\alpha}]}, \multi{x:\tau}]{e}{\lex{e}}}
                  \pf[11]{\dict[]{
                          \funcDelc{\sType{t}\typeFormalReceive}{m\typeFormalMethod}{\multi{x~\tau}}{\tau}{\return~e} \pfbreakline
                      }{
                          \funcDelc{\sType{t}}{m}{\multi{\dictlit~\dictName{\iType{u}}},~\multi{x~\any}}{\any}{\return~\lex{e}}
                      }}
              \end{pfsteps*}
              From here on we write $\lex{e}$ using the functional notation
              \[\lex{e} = \map[\Phi,\Psi; \eta; \this : {\sType{t}[\multi{\alpha}]}, \multi{x:\tau}]{e}\]
              Now that we have fleshed out the translation we begin to look at
              the translated term's reductions.
              For our value $v$ of type $\sType{t}\typeFormalReceive$, the translated
              term $\lex{v}$ is both a value and of type $\sType{t}$.
              This is immediately evident by \rulename{\dvalue}.
              As such the assertion is always resolved by $\prepre$.
              \begin{pfsteps*}
                  \pf[12]{\lex{v}.(\sType{t}).m(\makeDict[\emptyset; \emptyset]{\psi, \Psi}, \multi{\lex{v}}) \prepre
                      \lex{v}.m(\makeDict[\emptyset; \emptyset]{\psi, \Psi}, \multi{\lex{v}})}
              \end{pfsteps*}
              resolving the method call to the implementation in \pfref{11}
              \begin{pfsteps*}
                  \pf[13]{
                  \lex{v}.m(\makeDict[\emptyset; \emptyset]{\psi, \Psi}, \multi{\lex{v}})
                  \pfbreakline\red
                  \map[\Phi,\Psi; \eta; \this : {\sType{t}[\multi{\alpha}]}, \multi{x:\tau}]{e}
                  [\this \by \lex{v},
                  \multi{\dictlit} \by \makeDict[\emptyset; \emptyset]{\psi, \Psi},
                  \multi{x \by \lex{v}}]
                  }
              \end{pfsteps*}
              By the definition of $\lex{v}$, we can separate the substitution $\this\by \lex{v}$ into\\
              $\this\by \lex{v}, \multi{\this.\dictlit \by \lex{v}.\dictlit}$ meaning that we can
              rewrite the reduced term and then apply Lemma~\ref{lem:preposttype} and~\ref{lem:prepostval}
              \begin{pfsteps*}
                  \pf[13]{
                  \map[\Phi,\Psi; \eta; \this : {\sType{t}[\multi{\alpha}]}, \multi{x:\tau}]{e}
                  [\this \by \lex{v},
                  \multi{\dictlit} \by \makeDict[\emptyset; \emptyset]{\psi, \Psi},
                  \multi{x \by \lex{v}}]\pfbreakline
                  =
                  \map[\Phi,\Psi; \eta; \this : {\sType{t}[\multi{\alpha}]}, \multi{x:\tau}]{e}
                  [\multi{\this.\dictlit \by \lex{v}.\dictlit},
                  \multi{\dictlit} \by \makeDict[\emptyset; \emptyset]{\psi, \Psi},
                  \this \by \lex{v},
                  \multi{x \by \lex{v}}]\pfbreakline
                  =
                  \map[\Phi,\Psi; \eta; \this : {\sType{t}[\multi{\alpha}]}, \multi{x:\tau}]{e}
                  [\multi{\this.\dictlit} \by \makeDict[\emptyset;\emptyset]{\phi, \Phi},
                  \multi{\dictlit} \by \makeDict[\emptyset; \emptyset]{\psi, \Psi},
                  \this \by \lex{v},
                  \multi{x \by \lex{v}}]\pfbreakline
                  \precongdict^*
                  \map[\emptyset; \emptyset; \this : {\sType{t}[\multi{\alpha}]}, \multi{x:\tau}]{e[\theta]}
                  [ \this \by \lex{v},
                  \multi{x \by \lex{v}}]
                  \pfbreakline
                  \precongdict^*
                  \map[\emptyset; \emptyset; \emptyset]{e[\theta][\this \by v,
                              \multi{x \by v}]}
                  }
              \end{pfsteps*}
        \item[] \caseof{\rcall} --- (b) direction
              \begin{pfsteps*}
                  \pf[1]{\wellTyped[\emptyset; \emptyset]{
                          \sType{t}\typeActualReceive\sytxBrace{\multi{v_1}}.m\typeActualMethod(\multi{v_2})
                      }{\tau[\theta]}}
                  \pf[2]{\dict[\emptyset;\emptyset;\emptyset]{
                          \sType{t}\typeActualReceive\sytxBrace{\multi{v_1}}.m\typeActualMethod(\multi{v_2})
                      }{
                          \sType{t}\sytxBrace{\multi{\map[\emptyset;\emptyset;\emptyset]{v_1}}, \lex{\phi}}.(\sType{t}).m(\lex{\psi},\multi{\map[\emptyset;\emptyset;\emptyset]{v_2}})
                      }
                  }
                  \pf[3]{\lex{\phi} = \makeDict[\emptyset;\emptyset]{\phi, \Phi}}
                  \pf[4]{\lex{\psi} = \makeDict[\emptyset;\emptyset]{\psi, \Psi}}
              \end{pfsteps*}
              By the welltypedness of \pfref{1} we know that
              \begin{pfsteps*}
                  \pf[5]{\funcDelc{\sType{t}\typeFormalReceive}{m\typeFormalMethod}{ \multi{x~\tau}}{\tau}{\return e} \in \multi{D}}
                  \pf[6]{\type \sType{t}\typeFormalType \struct{\multi{y~\sigma}} \in \multi{D}}
              \end{pfsteps*}
              Translating \pfref{5} with $\Delta = \Phi, \Psi$ where $\Phi = \multi{\alpha~\iType{\tau}}$,
              $\Psi = \multi{\beta~\iType{\sigma}}$, and
              $\eta = \multi{\alpha \mapsto \this.\dictlit}, \multi{\beta \mapsto \dictlit}$ we get
              \begin{pfsteps*}
                  \pf[7]{\funcDelc{\sType{t}}{m}{ \multi{x~\any}}{\any}{\return \map{e}} \in \multi{\lex D}}
              \end{pfsteps*}
              The (b) direction assumes a reduction on the translated term.
              We first note that $\makeDict[\emptyset;\emptyset]{\cdots}$ is always
              a value. We then consider the trivial $\prepre$ reduction available
              before taking the \rulename{r-call} step.
              \begin{pfsteps*}
                  \pf[8]{
                      \sType{t}\sytxBrace{\multi{\map[\emptyset;\emptyset;\emptyset]{v_1}}, \lex{\phi}}.(\sType{t}).m(\lex{\psi},\multi{\map[\emptyset;\emptyset;\emptyset]{v_2}})
                      \pfbreakline \prepre
                      \sType{t}\sytxBrace{\multi{\map[\emptyset;\emptyset;\emptyset]{v_1}}, \lex{\phi}}.m(\lex{\psi},\multi{\map[\emptyset;\emptyset;\emptyset]{v_2}})
                      \pfbreakline \red
                      \map{e}[\this \by \sType{t}\sytxBrace{\multi{\map[\emptyset;\emptyset;\emptyset]{v_1}}, \lex{\phi}},
                          \multi{x \by \map[\emptyset;\emptyset;\emptyset]{v_2}},
                          \multi{\dictlit} \by \lex{\psi}]
                      \pfbreakline =
                      \map{e}[\multi{\this.\dictlit \by \sType{t}\sytxBrace{\multi{\map[\emptyset;\emptyset;\emptyset]{v_1}}, \lex{\phi}}.\dictlit},
                          \multi{\dictlit} \by \lex{\psi}]
                      [
                          \this \by \sType{t}\sytxBrace{\multi{\map[\emptyset;\emptyset;\emptyset]{v_1}}, \lex{\phi}},
                          \multi{x \by \map[\emptyset;\emptyset;\emptyset]{v_2}}
                      ]
                  }
              \end{pfsteps*}
              When we consider the $\prepostdict$ reduction
              we can relate $\multi{\this.\dictlit \by \sType{t}\sytxBrace{\multi{\map[\emptyset;\emptyset;\emptyset]{v_1}}, \lex{\phi}}.\dictlit}$
              and $\multi{\this.\dictlit \by  \lex{\phi}}$. This allows us to use
              Lemma \ref{lem:preposttype} and \ref{lem:prepostval}.
              \begin{pfsteps*}
                  \pf[8]{
                      \map{e}[\multi{\this.\dictlit \by \sType{t}\sytxBrace{\multi{\map[\emptyset;\emptyset;\emptyset]{v_1}}, \lex{\phi}}.\dictlit},
                          \multi{\dictlit} \by \lex{\psi}]
                      [
                          \this \by \sType{t}\sytxBrace{\multi{\map[\emptyset;\emptyset;\emptyset]{v_1}}, \lex{\phi}},
                          \multi{x \by \map[\emptyset;\emptyset;\emptyset]{v_2}}
                      ]
                      \pfbreakline \prepostdict^*
                      \map[\emptyset;\emptyset;\Gamma]{e[\Phi \by \phi, \Psi \by \Psi]}
                      [
                          \this \by \sType{t}\sytxBrace{\multi{\map[\emptyset;\emptyset;\emptyset]{v_1}}, \lex{\phi}},
                          \multi{x \by \map[\emptyset;\emptyset;\emptyset]{v_2}}
                      ]
                      \pfbreakline \prepostdict^*
                      \map[\emptyset;\emptyset;\emptyset]{e[\Phi \by \phi, \Psi \by \Psi]
                              [\this \by \sType{t}[\phi]\sytxBrace{\multi{v_1}},
                                  \multi{x \by v_2}]
                      }
                  }
              \end{pfsteps*}
              We now look at the (only) reduction available to the original term
              \begin{pfsteps*}
                  \pf[9]{\sType{t}\typeActualReceive\sytxBrace{\multi{v_1}}.m\typeActualMethod(\multi{v_2})
                      \red e[\Phi \by \phi, \Psi \by \psi][
                              \this \by \sType{t}\typeActualReceive\sytxBrace{\multi{v_1}},
                              \multi{x \by v_2}]}
              \end{pfsteps*}
        \item[] \caseof{\rassert} --- (a) direction
              \begin{pfsteps*}
                  \pf[1]{\sType{t}\typeActualReceive\sytxBrace{\multi{v}}.(\tau)
                      \red \sType{t}\typeActualReceive\sytxBrace{\multi{v}}}
                  \pf[2]{\dict[\emptyset;\emptyset;\emptyset]
                      {\sType{t}\typeActualReceive\sytxBrace{\multi{v}}.(\tau)}
                      {\typemeta[\emptyset]{\tau}.\trycast(\sTypeInit{t}{\multi{\lex{v}}, \lex{\phi}})}}
                  \pf[3]{\type~\sType{t}\typeFormalType~T \in \multi{D}}
                  \pf[4]{\lex{\phi} = \makeDict[\emptyset;\emptyset]{\phi, \Phi}}
                  \pfstep{\inversion{\rassert}~\pfref{1}}
                  {5}{\subtype[\emptyset]{\sType{t}\typeActualReceive}{\tau}}
              \end{pfsteps*}
              As such we know that $\typemeta[\emptyset]{\tau}.\trycast$
              should return (as opposed to panicking) if and only if
              $\subtype[\emptyset]{\sType{t}\typeActualReceive}{\tau}$
              \pfSubcase{\tau = \iType{u}\typeActualMethod}
              For \pfref{5} to hold the following must hold
              \begin{pfsteps*}
                  \pf[6]{\methods_\emptyset(\sType{t}\typeActualReceive) \supseteq
                      \methods_\emptyset(\iType{u}\typeActualMethod)}
                  \pf[7]{\type~\iType{u}\typeFormalMethod~\interface{\multi{S}} \in \multi{D}}
              \end{pfsteps*}
              For all $mM_{u} \in \multi{S}$ there
              must exist a function
              \[\func~(\this~\sType{t}\typeFormalReceive)~mM_t~\sytxBrace{\return e} \in \multi{D}\]
              such that \[M_u[\Psi \by \psi] = M_t[\Phi \by \phi]\]
              To show that this is property is preserved we need first elaborate
              a number of other definitions.
              Let $\Psi = (\typeFormal[\multi{\beta~\iType{\sigma}}])$,
              and the map $\zeta$ be $\{\multi{\beta \mapsto \this.\typeField}\}$. 
              \begin{pfsteps*}
                  \pf[9]{\dict[]{\type~\iType{u}\typeFormalMethod~\interface{\multi{S}}}{ \pfbreakline
                      \type \iType{u} \interface{\multi{\lex{S}},~\multi{\specMetadata{S}}}\pfbreakline
                      \funcDelc{\metadataName{\iType{u}}}{\trycast}{x~\any}{\any}{\pfbreakline
                          \qquad\left\{
                          \lit{if} (x.(\iType{u}).\specName{m}() \noteq
                          \signatureMeta{M_u}
                          )~ \sytxBrace{ \lit{panic}
                          }
                          ~\middle|~
                          \begin{matrix}
                              mM_u \in \multi{S}
                          \end{matrix}
                          \right\} ;\pfbreakline
                          \qquad \return x
                          \pfbreakline
                          }
                      }}
              \end{pfsteps*}
              And for $\Phi = (\typeFormal)$ and $\phi = \multi{\tau}$ the map
              $\zeta' = \{\alpha \mapsto \this.\dictlit_i.\typeField\}$.
              \begin{pfsteps*}
                  \pf[10]{\dict[]{
                          \func~(\this~\sType{t}\typeFormalReceive)~mM_t~\sytxBrace{\return e}\pfbreakline
                      }{
                          \func~(\this~\sType{t})~\specName{m}()~\fnMeta{n}~\sytxBrace{\return \signatureMeta[\zeta']{M_t}}
                      }}
                  \pf[11]{\lex{\phi} = \makeDict[\emptyset;\emptyset]{\multi{\tau}, \typeFormal}
                      = \multi{\dictName{\iType{t}}\sytxBrace{\multi{ptr}, \typemeta[\emptyset]{\tau}}  }}
              \end{pfsteps*}
              We may now consider the reduction of the translated term $\typemeta[\emptyset]{\tau}.\trycast(\sTypeInit{t}{\multi{\lex{v}}, \lex{\phi}})$
              {\small
              \begin{pfsteps*}
                  \pf[12]{
                  \typemeta[\emptyset]{\iType{u}\typeActualMethod}.\trycast(\sTypeInit{t}{\multi{\lex{v}}, \lex{\phi}})
                  \pfbreakline \red \pfbreakline
                  \left\{
                  \lit{if} (\sTypeInit{t}{\multi{\lex{v}}, \lex{\phi}}.(\iType{u}).\specName{m}() \noteq
                  \signatureMeta{M_u}
                  )~ \sytxBrace{ \lit{panic}
                  }
                  ~\middle|~
                  \begin{matrix}
                      mM_u \in \multi{S}
                  \end{matrix}
                  \right\} ;~ \pfbreakline\return \sTypeInit{t}{\multi{\lex{v}}, \lex{\phi}}
                  }
              \end{pfsteps*}
              }
              We can now use Lemma \ref{lem:methspec} to resolve $\zeta$
              {\small
              \begin{pfsteps*}
                  \pf[13]{ \left\{
                  \lit{if} (\sTypeInit{t}{\multi{\lex{v}}, \lex{\phi}}.(\iType{u}).\specName{m}() \noteq
                  \signatureMeta[\emptyset]{M_u[\Psi \by \psi]}
                  )~ \sytxBrace{ \lit{panic}
                  }
                  ~\middle|~
                  \begin{matrix}
                      mM_u \in \multi{S}
                  \end{matrix}
                  \right\} ;~ \pfbreakline \return \sTypeInit{t}{\multi{\lex{v}}, \lex{\phi}}}
              \end{pfsteps*}
              }
              Using the $\assertpattern$ we can further reduce the term.
              While this would happen in a sequential order we simplify the presentation of the proof.
              We begin by looking at $\sTypeInit{t}{\multi{\lex{v}}, \lex{\phi}}.(\iType{u})$.
              Since $\subtype[\emptyset]{\sType{t}\typeActualReceive}{\iType{u}\typeActualMethod}$
              we know that $\sType{t}$ must posses each method defined by $\iType{u}$.
              {\small
              \begin{pfsteps*}
                  \pf[14]{ \prepostsim^* \left\{
                      \lit{if} (\sTypeInit{t}{\multi{\lex{v}}, \lex{\phi}}.\specName{m}() \noteq
                      \signatureMeta[\emptyset]{M_u[\Psi \by \psi]}
                      )~ \sytxBrace{ \lit{panic}
                      }
                      ~\middle|~
                      \begin{matrix}
                          mM_u \in \multi{S}
                      \end{matrix}
                      \right\} ;~ \pfbreakline \quad \return \sTypeInit{t}{\multi{\lex{v}}, \lex{\phi}}
                      \pfbreakline
                      \prepostsim^*
                      \left\{
                      \lit{if} (\signatureMeta[\zeta']{M_t} \noteq
                      \signatureMeta[\emptyset]{M_u[\Psi \by \psi]}
                      )~ \sytxBrace{ \lit{panic}
                      }
                      ~\middle|~
                      \begin{matrix}
                          mM_u \in \multi{S}
                      \end{matrix}
                      \right\} ;~ \pfbreakline \quad \return \sTypeInit{t}{\multi{\lex{v}}, \lex{\phi}}
                  }
              \end{pfsteps*}
              }
              We can now use Lemma \ref{lem:methspec} to resolve $\zeta'$
              {\small
              \begin{pfsteps*}
                  \pf[15]{
                      \prepostsim^*
                      \left\{
                      \lit{if} (\signatureMeta[\emptyset]{M_t[\Phi \by \phi]} \noteq
                      \signatureMeta[\emptyset]{M_u[\Psi \by \psi]}
                      )~ \sytxBrace{ \lit{panic}
                      }
                      ~\middle|~
                      \begin{matrix}
                          mM_u \in \multi{S}
                      \end{matrix}
                      \right\} ;~ \pfbreakline\quad \return \sTypeInit{t}{\multi{\lex{v}}, \lex{\phi}}
                  }
              \end{pfsteps*}
              }
              Since $M_u[\Psi \by \psi] = M_t[\Phi \by \phi]$, no $\lit{if}$ is
              triggered.
              \begin{pfsteps*}
                  \pf[16]{
                      \prepostsim^*
                      \return \sTypeInit{t}{\multi{\lex{v}}, \lex{\phi}}
                      \pfbreakline
                      \prepostsim \sTypeInit{t}{\multi{\lex{v}}, \lex{\phi}}}
              \end{pfsteps*}
              which is the desired term.
              \pfSubcase{\tau = \sType{t}[\phi]}
              For \pfref{5} to hold if $\tau$ is a structure type then it must
              be precisely the same type as target of the assertion.
              \begin{pfsteps*}
                  \pf[17]{\typemeta[\emptyset]{\tau}.\trycast(\sTypeInit{t}{\multi{\lex{v}}, \lex{\phi}})
                  \pfbreakline
                  = \metadataName{\sType{t}}\sytxBrace{\typemeta[\emptyset]{\phi}}.\trycast(\sTypeInit{t}{\multi{\lex{v}}, \lex{\phi}})
                  \pfbreakline
                  \red \{\lit{if} ~ \metadataName{\sType{t}}\sytxBrace{\typemeta[\emptyset]{\phi}}.\typeField_i \noteq \sTypeInit{t}{\multi{\lex{v}}, \lex{\phi}}.(\sType{t}).\dictlit_i.\typeField~\sytxBrace{\lit{panic}}\}_{i<n} ; \return \sTypeInit{t}{\multi{\lex{v}}, \lex{\phi}}
                  }
              \end{pfsteps*}
              Once again we use $\assertpattern$ to resolve assertion. We also
              use the same proof simplification and ignore explicit sequentiality.
              {\small
              \begin{pfsteps*}
                  \pf[18]{\{\lit{if} ~ \metadataName{\sType{t}}\sytxBrace{\typemeta[\emptyset]{\phi}}.\typeField_i \noteq \sTypeInit{t}{\multi{\lex{v}}, \lex{\phi}}.(\sType{t}).\dictlit_i.\typeField~\sytxBrace{\lit{panic}}\}_{i<n} ; \return \sTypeInit{t}{\multi{\lex{v}}, \lex{\phi}}
                  \pfbreakline \prepostsim^*
                  \{\lit{if} ~ \typemeta[\emptyset]{\phi_i} \noteq \sTypeInit{t}{\multi{\lex{v}}, \lex{\phi}}.(\sType{t}).\dictlit_i.\typeField~\sytxBrace{\lit{panic}}\}_{i<n} ; \return \sTypeInit{t}{\multi{\lex{v}}, \lex{\phi}}
                  \pfbreakline \prepostsim^*
                  \{\lit{if} ~ \typemeta[\emptyset]{\phi_i} \noteq \sTypeInit{t}{\multi{\lex{v}}, \lex{\phi}}.\dictlit_i.\typeField~\sytxBrace{\lit{panic}}\}_{i<n} ; \return \sTypeInit{t}{\multi{\lex{v}}, \lex{\phi}}
                  \pfbreakline \prepostsim^*
                  \{\lit{if} ~ \typemeta[\emptyset]{\phi_i} \noteq  \lex{\phi_i}.\typeField~\sytxBrace{\lit{panic}}\}_{i<n} ; \return \sTypeInit{t}{\multi{\lex{v}}, \lex{\phi}}
                  \pfbreakline \prepostsim^*
                  \{\lit{if} ~ \typemeta[\emptyset]{\phi_i} \noteq  \typemeta[\emptyset]{\phi_i}~\sytxBrace{\lit{panic}}\}_{i<n} ; \return \sTypeInit{t}{\multi{\lex{v}}, \lex{\phi}}
                  \pfbreakline \prepostsim^*
                  \sTypeInit{t}{\multi{\lex{v}}, \lex{\phi}}
                  }
              \end{pfsteps*}
              }
        \item[] \caseof{\rassert} --- (b) direction\\
              This direction follows closely the (a) direction other than that it
              does not assume \\$\sType{t}\typeActualReceive\sytxBrace{\multi{v}}.(\tau)
                  \red \sType{t}\typeActualReceive\sytxBrace{\multi{v}}$.
              Yet by our assumption that $\sType{t}\typeActualReceive\sytxBrace{\multi{v}}.(\tau)$
              is not a type assertion error this reduction must exist. It
              then suffices to show that the source and target terms' reductions match,
              which is given in (a).
        \item[] \caseof{\rassert} --- (c) direction\\
              We first note that $e = v.(\tau)$ is the only case for (c)
              as no other term can produce a panic, and that
              $\dicttrans$ is defined as the greatest reduction available.
              As such for $\lex{e} \dicttrans e'$ there is no further $e' \prepostsim$.
              \begin{pfsteps*}
                  \pf[1]{v.(\tau)~\panic}
                  \pf[2]{\subtype[\emptyset]{\vtype(v)\not}{\tau}}
                  \pf[3]{\dict[]{v.(\tau)}{\typemeta[\emptyset]{\tau}.\trycast(\map[\emptyset;\emptyset;\emptyset]{v})}}
                  \pf[4]{v = \sType{t}\typeActualReceive\sytxBrace{\multi{v}}}
                  \pf[5]{\map[\emptyset;\emptyset;\emptyset]{v} =
                      \sTypeInit{t}{\multi{\lex{v}}, \makeDict[\emptyset;\emptyset]{\phi, \Phi}}}
              \end{pfsteps*}
              \pfSubcase{\tau = \iType{u}\typeActualMethod}\\
              For \pfref{2} to hold there must be at least one method
              $mM \in \methods_\emptyset(\tau)$
              such that $mM \not\in \methods_\emptyset(\vtype(v))$.

              Let $\Psi = (\typeFormal[\multi{\beta~\iType{\sigma}}])$,
              and the map $\zeta$ be $\{\multi{\beta \mapsto this.\typeField}\}$.
              \begin{pfsteps*}
                  \pf[9]{\dict[]{\type~\iType{u}\typeFormalMethod~\interface{\multi{S}}}{ \pfbreakline
                      \type \iType{u} \interface{\multi{\lex{S}},~\multi{\specMetadata{S}}}\pfbreakline
                      \funcDelc{\metadataName{\iType{u}}}{\trycast}{x~\any}{\any}{\pfbreakline
                          \qquad\left\{
                          \lit{if} (x.(\iType{u}).\specName{m}() \noteq
                          \signatureMeta{M_u}
                          )~ \sytxBrace{ \lit{panic}
                          }
                          ~\middle|~
                          \begin{matrix}
                              mM_u \in \multi{S}
                          \end{matrix}
                          \right\} ;\pfbreakline
                          \qquad \return x
                          \pfbreakline
                          }
                      }}
              \end{pfsteps*}

              The translated term will always be able to make the reduction
              \begin{pfsteps*}
                  \pf[8]{
                  \typemeta[\emptyset]{\tau}.\trycast(\map[\emptyset;\emptyset;\emptyset]{v})
                  \red
                  \left\{
                  \lit{if} (\map[\emptyset;\emptyset;\emptyset]{v}.(\iType{u}).\specName{m}() \noteq
                  \signatureMeta{M_u}
                  )~ \sytxBrace{ \lit{panic}
                  }
                  ~\middle|~
                  \begin{matrix}
                      mM_u \in \multi{S}
                  \end{matrix}
                  \right\} ; \pfbreakline\quad \return x
                  }
              \end{pfsteps*}
              For convenience we assume that the problematic method $m$ is the
              first to be checked. If this is not the case then we may reduce
              all ok checks using $\prepostsim$ as described in the \rulename{\rassert}
              (a) case.

              \begin{pfsteps*}
                  \pf[7]{
                  \left\{
                  \lit{if} (\map[\emptyset;\emptyset;\emptyset]{v}.(\iType{u}).\specName{m}() \noteq
                  \signatureMeta{M_u}
                  )~ \sytxBrace{ \lit{panic}
                  }
                  ~\middle|~
                  \begin{matrix}
                      mM_u \in \multi{S}
                  \end{matrix}
                  \right\} ; \pfbreakline\quad \return x
                  \pfbreakline
                  \prepostsim^*
                  \lit{if} (\map[\emptyset;\emptyset;\emptyset]{v}.(\iType{u}).\specName{m}() \noteq
                  \signatureMeta{M_u}
                  )~ \sytxBrace{ \lit{panic}
                  } ; \cdots
                  }

              \end{pfsteps*}

              We now need to consider the two possible cases in which
              $mM \not\in \methods_\emptyset(\vtype(v))$ could hold.
              Either there is no method called $m$ in $\methods_\emptyset(\vtype(v))$
              or there is a method $m$ but with a different method signatures.
              In the former case the assertion $E[\map[\emptyset;\emptyset;\emptyset]{v}.(\iType{u})]$
              will panic as by our assumption that translation will never introduce a
              name collision the method $m$ will not be in $\methods(\sType{t})$
              (the methods of $\vtype(\map[\emptyset;\emptyset;\emptyset]{v})$).

              In the latter we assume
              that for $mM_t[\Phi \by \phi] \in \methods_\emptyset(\vtype(v))$
              and $mM_u[\Psi \by \psi] \in \methods_\emptyset(\iType{u}\typeActualMethod)$
              such that $M_t[\Phi \by \phi] \noteq M_u[\Psi \by \psi]$
              then the $\lit{if}$ branches to $\lit{panic}$.

              Let $\Phi = (\typeFormal)$, $\phi = \multi{\tau}$, and the map
              $\zeta' = \{\alpha \mapsto \this.\dictlit_i.\typeField\}$.
              \begin{pfsteps*}
                  \pf[10]{\dict[]{
                          \func~(\this~\sType{t}\typeFormalReceive)~mM_t~\sytxBrace{\return e}\pfbreakline
                      }{
                          \func~(\this~\sType{t})~\specName{m}()~\fnMeta{n}~\sytxBrace{\return \signatureMeta[\zeta']{M_t}}
                      }}
                  \pf[11]{\map[\emptyset;\emptyset;\emptyset]{v}
                      =
                      \sTypeInit{t}{\multi{\lex{v}}, \multi{\dictName{\iType{t}}\sytxBrace{\multi{ptr}, \typemeta[\emptyset]{\tau}}  }}}
                  \pf[12]{
                  \lit{if} (\map[\emptyset;\emptyset;\emptyset]{v}.(\iType{u}).\specName{m}() \noteq \signatureMeta{M_u})~ \sytxBrace{ \lit{panic}} ; \cdots
                  \pfbreakline
                  \prepostsim
                  \lit{if} (\map[\emptyset;\emptyset;\emptyset]{v}.\specName{m}() \noteq \signatureMeta{M_u})~ \sytxBrace{ \lit{panic}} ; \cdots
                  \pfbreakline
                  \prepostsim
                  \lit{if} (\signatureMeta[\zeta']{M_t} \noteq \signatureMeta{M_u})~ \sytxBrace{ \lit{panic}} ; \cdots
                  }
              \end{pfsteps*}
              We can now apply Lemma~\ref{lem:methspec}, first to the lhs then rhs
              \begin{pfsteps*}
                  \pf[13]{
                      \lit{if} (\signatureMeta[\zeta']{M_t} \noteq \signatureMeta{M_u})~ \sytxBrace{ \lit{panic}} ; \cdots
                      \pfbreakline
                      \prepostsim^*
                      \lit{if} (\signatureMeta[\emptyset]{M_t[\Phi\by\phi]} \noteq \signatureMeta[\emptyset]{M_u[\Psi\by\psi]})~ \sytxBrace{ \lit{panic}} ; \cdots
                  }
              \end{pfsteps*}
              By $M_t[\Phi \by \phi] \noteq M_u[\Psi \by \psi]$ this reduces to
              the desired $\lit{panic}$.
              \pfSubcase{\tau = \sType{u}\typeActualMethod}
              {\small
              \begin{pfsteps*}
                  \pf[14]{\dict[]{\type \sType{u}[\Phi] \struct{\cdots}\pfbreakline}{
                  \type \metadataName{\sType{u}} \struct{\multi{\typeField~\typeMetadataLit}}
                  \pfbreakline
                  \funcDelc{\this~\metadataName{\sType{u}}}{\trycast}{x~\any}{\any}{
                  \pfbreakline
                  \qquad x.(\sType{u})~;~\{\lit{if} ~ \this.\typeField_i \noteq x.(\sType{u}).\dictlit_i.\typeField~\sytxBrace{\lit{panic}}\}_{i<n} ; \return x
                  \pfbreakline
                  }
                  }}
              \end{pfsteps*}
              }
              For $\vtype(v) = \sType{t}\typeActualReceive$.
              If $\tau$ is a struct then there are two case.
              Either
              $\sType{u} \noteq \sType{t}$, or
              $\sType{u} = \sType{t}$
              but for $\phi = \multi{\sigma}$ and $\psi = \multi{\tau}$ there
              exists an $i$ such that $\sigma_i \noteq \tau_i$.

              We first consider the case $\sType{u} \noteq \sType{t}$.
              Note that  $\vtype(\map[\emptyset;\emptyset;\emptyset]{v}) = \sType{t}\typeActualReceive$.
              \begin{pfsteps*}
                  \pf[20]{\typemeta[\emptyset]{\sType{u}\typeActualMethod}.\trycast(\sTypeInit{t}{\multi{\lex{v}}, \makeDict[\emptyset;\emptyset]{\phi, \Phi}})
                      \pfbreakline
                      = \metadataName{\sType{u}}\sytxBrace{\typemeta[\emptyset]{\psi}}.\trycast(\sTypeInit{t}{\multi{\lex{v}}, \makeDict[\emptyset;\emptyset]{\phi, \Phi}})
                      \pfbreakline
                      \red \sTypeInit{t}{\multi{\lex{v}}, \makeDict[\emptyset;\emptyset]{\phi, \Phi}}.(\sType{u})~;\cdots
                  }
              \end{pfsteps*}
              By our assumption $\sType{u} \noteq \sType{t}$ we get the desired
              \panic.

              We now consider the case of $\sType{u} = \sType{t}$
              but for $\phi = \multi{\sigma}$ and $\psi = \multi{\tau}$ there
              exists an $i$ such that $\sigma_i \noteq \tau_i$.

              \begin{pfsteps*}
                  \pf[21]{\typemeta[\emptyset]{\tau}.\trycast(\sTypeInit{t}{\multi{\lex{v}}, \makeDict[\emptyset;\emptyset]{\phi, \Phi}})
                  \pfbreakline
                  = \metadataName{\sType{t}}\sytxBrace{\typemeta[\emptyset]{\psi}}.\trycast(\sTypeInit{t}{\multi{\lex{v}}, \makeDict[\emptyset;\emptyset]{\phi, \Phi}})
                  \pfbreakline
                  \red \sTypeInit{t}{\multi{\lex{v}}, \makeDict[\emptyset;\emptyset]{\phi, \Phi}}.(\sType{t})~;~\{\lit{if} ~ \metadataName{\sType{u}}\sytxBrace{\typemeta[\emptyset]{\psi}}.\typeField_i \noteq
                  \pfbreakline \qquad \qquad \sTypeInit{t}{\multi{\lex{v}}, \makeDict[\emptyset;\emptyset]{\phi, \Phi}}.(\sType{t}).\dictlit_i.\typeField~\sytxBrace{\lit{panic}}\}_{i<n} ; \return \cdots
                  \pfbreakline
                  \{\lit{if} ~ \metadataName{\sType{u}}\sytxBrace{\typemeta[\emptyset]{\psi}}.\typeField_i \noteq
                  \pfbreakline \qquad \qquad \sTypeInit{t}{\multi{\lex{v}}, \makeDict[\emptyset;\emptyset]{\phi, \Phi}}.(\sType{t}).\dictlit_i.\typeField~\sytxBrace{\lit{panic}}\}_{i<n} ; \return \cdots
                  }
              \end{pfsteps*}
              We once again only need to consider the (lowest) $i$ for which  $\sigma_i \noteq \tau_i$.
              All prior $\lit{if}$ statement pass as per \rulename{\rassert} (a).

              \begin{pfsteps*}
                  \pf[22]{
                  \{\lit{if} ~ \metadataName{\sType{u}}\sytxBrace{\typemeta[\emptyset]{\psi}}.\typeField_i \noteq
                  \pfbreakline \qquad \qquad \sTypeInit{t}{\multi{\lex{v}}, \makeDict[\emptyset;\emptyset]{\phi, \Phi}}.(\sType{t}).\dictlit_i.\typeField~\sytxBrace{\lit{panic}}\}_{i<n} ; \return \cdots
                  \pfbreakline
                  \prepostsim^*
                  \lit{if} ~ \metadataName{\sType{u}}\sytxBrace{\typemeta[\emptyset]{\psi}}.\typeField_i \noteq
                  \pfbreakline \qquad \qquad \sTypeInit{t}{\multi{\lex{v}}, \makeDict[\emptyset;\emptyset]{\phi, \Phi}}.(\sType{t}).\dictlit_i.\typeField~\sytxBrace{\lit{panic}} ; \return \cdots
                  \pfbreakline
                  \prepostsim
                  \lit{if} ~ \typemeta[\emptyset]{\tau_i} \noteq
                  \sTypeInit{t}{\multi{\lex{v}}, \makeDict[\emptyset;\emptyset]{\phi, \Phi}}.(\sType{t}).\dictlit_i.\typeField~\sytxBrace{\lit{panic}} ; \return \cdots
                  \pfbreakline
                  \prepostsim
                  \lit{if} ~ \typemeta[\emptyset]{\tau_i} \noteq
                  \sTypeInit{t}{\multi{\lex{v}}, \makeDict[\emptyset;\emptyset]{\phi, \Phi}}.\dictlit_i.\typeField~\sytxBrace{\lit{panic}} ; \return \cdots
                  \pfbreakline
                  \prepostsim
                  \lit{if} ~ \typemeta[\emptyset]{\tau_i} \noteq
                  \makeDict[\emptyset;\emptyset]{\sigma_i, \Phi_i}.\typeField~\sytxBrace{\lit{panic}} ; \return \cdots
                  \pfbreakline
                  \prepostsim
                  \lit{if} ~ \typemeta[\emptyset]{\tau_i} \noteq
                  \typemeta[\emptyset]{\sigma_i}~\sytxBrace{\lit{panic}} ; \return \cdots
                  }
              \end{pfsteps*}
              by our assumption that $\tau_i \noteq \sigma_i$ we get the desired
              \panic.
        \item[] \caseof{\rassert} --- (d) direction\\
              Once again we need only consider $e= v.(\tau)$.
              This case follows from case (c), but we must first
              show that there exists at least an $\lex{e}\prepre^*\red d$
              reduction. This $d$ then reduces by $\prepostsim^*$ to $e'$, where
              $e'$ is a type assertion error.
              We know that $d$ exists by observing that the translation of $v.(\tau)$
              will always reduce ($\red$) by \rulename{r-call} on $\trycast$.
              This $d$ will then reduce ($\prepostsim^*$) to $e'$, which by
              the same logic as (c) is a type assertion error.
        \item[] \caseof{\rcontext} --- (a) direction\\
              The only non-immediate case for \rulename{\rcontext} is when 
              $E = \hole.m\typeActualMethod(\multi{v})$. 
              \begin{pfsteps*}
                  \pf[1]{\infer{
                    e.m\typeActualMethod(\multi{v}) 
                    \red 
                    d.m\typeActualMethod(\multi{v})
                  }{
                    e\red d
                  }}
                  \pf[2]{
                      \wellTyped[\emptyset; \emptyset]{ e.m\typeActualMethod(\multi{v}) }{\sigma}
                  }
                  \pfstep{\inversion{\tcall}}
                    {3}{\wellTyped[\emptyset; \emptyset]{e}{t[\phi]}} 
              \end{pfsteps*}
              By preservation (lemma~\ref{thm:fggTypePreservation})
              \begin{pfsteps*}
                \pf[4]{\wellTyped[\emptyset; \emptyset]{d}{u[\phi']}}
                \pf[5]{\subtype[\emptyset]{u[\phi']}{t[\phi]}} 
            \end{pfsteps*}
            Translating $E[e]$ and $E[d]$ we get 
            \begin{pfsteps*}
                \pf[6]{\map[\emptyset;\emptyset;\emptyset]{E[e]} = 
                \map[\emptyset;\emptyset;\emptyset]{e}.(t).m(\lex{\psi}, 
                \multi{\map[\emptyset;\emptyset;\emptyset]{e}}  ) 
                }
                \pf[7]{\map[\emptyset;\emptyset;\emptyset]{E[d]} = 
                \map[\emptyset;\emptyset;\emptyset]{d}.(u).m(\lex{\psi}, 
                \multi{\map[\emptyset;\emptyset;\emptyset]{e}}  ) 
                }
            \end{pfsteps*}
            By the induction hypothesis on $e\red d$ 
            \begin{pfsteps*}
                \pf[8]{\map[\emptyset;\emptyset;\emptyset]{e}
                \dicttrans \precongdict^\ast
                \map[\emptyset;\emptyset;\emptyset]{d}}
            \end{pfsteps*}
            Using the evaluation context $\hole.(t).m(\lex{\psi}, 
            \multi{\map[\emptyset;\emptyset;\emptyset]{e}})$
            \begin{pfsteps*}
                \pf[9]{\map[\emptyset;\emptyset;\emptyset]{e}.(t).m(\lex{\psi}, 
                \multi{\map[\emptyset;\emptyset;\emptyset]{e}})
                \dicttrans \precongdict^\ast
                \map[\emptyset;\emptyset;\emptyset]{d}.(t).m(\lex{\psi}, 
                \multi{\map[\emptyset;\emptyset;\emptyset]{e}})
                }
            \end{pfsteps*}
            Using synthetic assertion specialisation and \ref{lem:sub:pres} on \pfref{5}
            \begin{pfsteps*}
                \pf[10]{\map[\emptyset;\emptyset;\emptyset]{d}.(t).m(\lex{\psi}, 
                \multi{\map[\emptyset;\emptyset;\emptyset]{e}})
                \precongdict 
                \map[\emptyset;\emptyset;\emptyset]{d}.(u).m(\lex{\psi}, 
                \multi{\map[\emptyset;\emptyset;\emptyset]{e}})}
            \end{pfsteps*}
        \item[] \caseof{\rcontext} --- (a) direction\\
        The only non-immediate case for \rulename{\rcontext}
        is for $E = \hole.m\typeActualMethod(\multi{v})$
        \begin{pfsteps*}
            \pf[0]{\wellTyped[\emptyset;\emptyset]{e}{t[\phi]}}
            \pf[2]{\map[\emptyset;\emptyset;\emptyset]{E[e]} = 
                \map[\emptyset;\emptyset;\emptyset]{e}.(t).m(\lex{\psi}, 
                \multi{\map[\emptyset;\emptyset;\emptyset]{e}}  ) 
                }
            \pf[1]{ \map[\emptyset;\emptyset;\emptyset]{e}.(t).m(\lex{\psi}, 
                \multi{\map[\emptyset;\emptyset;\emptyset]{e}}) \dicttrans e'.(t).m(\lex{\psi}, 
                \multi{\map[\emptyset;\emptyset;\emptyset]{e}}) } 
        \end{pfsteps*}
        By inversion on the reduction $\dicttrans$ using context 
        $E' = \hole.(t).m(\lex{\psi}, 
        \multi{\map[\emptyset;\emptyset;\emptyset]{e}})$ 
        \begin{pfsteps*}
            \pf[3]{ \map[\emptyset;\emptyset;\emptyset]{e} \dicttrans e'}
        \end{pfsteps*}
        By the induction hypothesis on \pfref{3} there exists $d$ 
        \begin{pfsteps*}
            \pf[4]{e\red d}
            \pf[5]{e' \precongdict^* \map[\emptyset;\emptyset;\emptyset]{d}}
            \pfstep{lemma~\ref{thm:fggTypePreservation} \pfref{0}}
            {7}{\wellTyped[\emptyset;\emptyset]{d}{u[\phi']}}
            \pf[8]{\subtype[\emptyset]{u[\phi']}{t[\phi]}}
        \end{pfsteps*}
        Applying \pfref{5} on context $C=E'$ we get that 
        \begin{pfsteps*}
            \pf[6]{e'.(t).m(\lex{\psi}, 
            \multi{\map[\emptyset;\emptyset;\emptyset]{e}}) \precongdict^* 
            \map[\emptyset;\emptyset;\emptyset]{d}.(t).m(\lex{\psi}, 
                \multi{\map[\emptyset;\emptyset;\emptyset]{e}}) }
        \end{pfsteps*}
        Using synthetic assertion specialisation and \ref{lem:sub:pres} on \pfref{8}
        \begin{pfsteps*}
            \pf[6]{\map[\emptyset;\emptyset;\emptyset]{d}.(t).m(\lex{\psi}, 
            \multi{\map[\emptyset;\emptyset;\emptyset]{e}}) 
            \precongdict^* 
            \map[\emptyset;\emptyset;\emptyset]{d}.(u).m(\lex{\psi}, 
                \multi{\map[\emptyset;\emptyset;\emptyset]{e}}) }
        \end{pfsteps*}
        Using \rulename{\rcontext} on \pfref{4} and context $E$. 
        \begin{pfsteps*}
            \pf[10]{E[e]\red E[d]}
        \end{pfsteps*}
        Finally, using the typing of $d$ \pfref{7} we get the translation of $E[d]$
        \begin{pfsteps*}
            \pf[11]{\map[\emptyset;\emptyset;\emptyset]{E[d]} = 
            \map[\emptyset;\emptyset;\emptyset]{d}.(u).m(\lex{\psi}, 
            \multi{\map[\emptyset;\emptyset;\emptyset]{e}})}
        \end{pfsteps*}
    \end{itemize}
\end{proof}

\lemredrewrite*
  \begin{proof}
    (1) is immediate as if $d_1\red d_2$ then  $E[d_1] \red E[d_2] \red e_3$.\\
    (2) is by case analysis on the reduction $e_2\red e_3$.
    \begin{itemize}
        \item[] \caseof{r-field}:
              We have that $e_2 = \sTypeInit{t}{\multi{v}}$.
              There are two possible options for the congruence
              evaluation context $C$, either it is $\hole$ or it is
              of the form $\sTypeInit{t}{\multi{v}, \hole, \multi{v}}$.
              In either case we get a contradiction as both cases
              are captured by the standard evaluation context.
        \item[] \caseof{r-call}:
              Same logic as \rulename{r-field}.
        \item[] \caseof{r-assert}:
              Same logic as \rulename{r-field}.
        \item[] \caseof{r-context}:
              We begin with an assumption that the
              congruence context $C$ deviated from the standard context
              at the top level. Namely there does not exist $E'$, $C'$ such
              that $C = E'[C']$. We do this for clarity.
              In the situation that this does not
              hold, we may simply use add $E'$ where appropriate.
  
              There are two cases for $C\neq E$. Either
              $C$ is of the form
              $\sTypeInit{t}{\multi{v}, e, \multi{e_1}, C_{\text{sub}}, \multi{e_2}}$ or
              $v.m(\multi{v}, e, \multi{e_1}, C_{\text{sub}}, \multi{e_2})$.
              We focus on the former.
  
              The starting term $e_1$ is $\sTypeInit{t}{\multi{v}, e, \multi{e_1}, C_{\text{sub}}[d_1], \multi{e_2}}$,
              the subsequent term $e_2$ is $\sTypeInit{t}{\multi{v}, e, \multi{e_1}, C_{\text{sub}}[d_2], \multi{e_2}}$,
              and the final term $e_3$ is $\sTypeInit{t}{\multi{v}, d, \multi{e_1}, C_{\text{sub}}[d_2], \multi{e_2}}$
              for some $d$ such that $e\red d$.
  
              Our initial term may make a standard reduction to
              $e_2' = \sTypeInit{t}{\multi{v}, d, \multi{e_1}, C_{\text{sub}}[d_1], \multi{e_2}}$,
              followed by a $\dictred$ reduction using $C'' = \sTypeInit{t}{\multi{v}, d, \multi{e_1}, C_{\text{sub}}, \multi{e_2}}$
              to $e_3$.
  
    \end{itemize}
  \end{proof}

\lemredvalue*
  \begin{proof}
    Assume for contradiction that $\dictred \not\in \red$. 
    Either 
    \begin{itemize}
        \item[] \caseStd $C[e']\dictred C[d] = v$ where $C$ is not the 
            standard $\red$ reduction context. Since there 
            must be another $\red$ reduction from $C[d]$ using the 
            standard reduction context $E$ it cannot be a value. 
        \item[] \caseStd $C[e'.(u)]\dictred C[e'.(t)]$. Immediate as $ C[e'.(t)]$
            is not a value. 
    \end{itemize}
  \end{proof}

%% file: sections/appendix-dyn-ex.tex
\section{Appendix: Motivating Example Translated using erasure}
\label{sec:erasure-example}
\input{figs/dyn/fgg-nomono.tex}

%% file: figs/dyn/fgg-nomono.tex
\begin{figure}[H]
    \centering
\begin{minipage}[ht]{0.5\linewidth }
\begin{lstfcgg}
type List interface {
    permute() Any
    insert( val Any, i Any) Any
    map(func(Any) Any) Any
    len() Any }

type Cons struct {
    head Any
    tail Any
}
type Nil struct {}

// Naively produce all possible list orderings of list this. 
func (this Cons) permute() Any { 
    if this.len().(int) == 1 {
    return Cons{this, Nil{}}
    } else {
    return flatten(this.tail.permute().Map(
        func(l Any) Any{
        var l_new Any = Nil{}
        for i := 0; i <= l.(List).len().(int); i++ {
            l_new = Cons{l.(List).insert(this.head, i), l_new}
        }
        return l_new
}))}}
func (this Nil) permute() Any { return Nil{} }
\end{lstfcgg}
\end{minipage}
\caption{The \gls{erasure} translation of the example in \S~\ref{sec:introduction}}
\end{figure}